\documentclass{article}

\usepackage{amsfonts}
\usepackage{amssymb,epic,eepic}
\usepackage{amsmath}
\usepackage{dcolumn}
\usepackage{bm}
\usepackage{verbatim}
\usepackage{color}
\usepackage{amsthm}
\newcommand{\hr}{{\mathcal H}}
\newcommand{\cn}{{\mathcal N }}
\newcommand{\E}{{\mathcal E }}

\newcommand{\cs}{{\mathcal S}}
\newcommand{\crr}{{\mathcal R}}
\newcommand{\fr}{{\mathcal F}}
\newcommand{\gr}{{\mathcal G}}
\newcommand{\X}{{\mathcal X}}
\newcommand{\fri}{{\mathfrak I}}
\newcommand{\kr}{{\mathcal K}}
\newcommand{\bo}{{\mathcal B}}
\newcommand{\cc}{{\mathbb C}}
\newcommand{\rr}{{\mathbb R}}

\newcommand{\nn}{{\mathbb N}}
\newcommand{\idn}{\mathbf{1}}

\newcommand{\eps}{{\varepsilon}}        
\newcommand{\vphi}{{\varphi}}           

\newcommand{\A}{\mathcal A}
\newcommand{\B}{\mathcal B}
\newcommand{\D}{\mathcal D}
\newcommand{\cP}{\mathcal P}
\newcommand{\bS}{\mathbf S}

\newtheorem{theorem}{Theorem}

\newtheorem{corollary}[theorem]{Corollary}

\newtheorem{definition}[theorem]{Definition}

\newtheorem{lemma}[theorem]{Lemma}

\newtheorem{remark}[theorem]{Remark}

\newcommand{\tr}{\mathrm{tr}}

\DeclareMathOperator{\conv}{conv}
\DeclareMathOperator{\aff}{aff}
\DeclareMathOperator{\ri}{ri}
\DeclareMathOperator{\rebd}{rebd}

\begin{document}
\title{Quantum capacity under adversarial quantum noise: arbitrarily varying quantum channels}
\author{Rudolf Ahlswede $^{1}$\footnote{Tragically, Rudolf Ahlswede passed away during the preparation of the final version of the present paper in December 2010. We, the remaining authors, are thankful to have had the opportunity to experience and enjoy his boundless enthusiasm for science and his lively spirit.}, Igor Bjelakovi\'c $^{2}$, Holger Boche $^{3}$, Janis N\"otzel $^{2}$\\
\scriptsize{Electronic addresses: \{igor.bjelakovic, boche, janis.noetzel\}@tum.de}
\vspace{0.2cm}\\
$^{1}$ {\footnotesize Fakult\"at f\"ur Mathematik, Universit\"at Bielefeld,}\\
{\footnotesize Universit\"atsstr. 25, 33615 Bielefeld, Germany}\\
$^{2}$ {\footnotesize Theoretische Informationstechnik, Technische Universit\"at M\"unchen,}\\
{\footnotesize 80291 M\"unchen, Germany}\\
$^{3}$ {\footnotesize Lehrstuhl f\"ur Theoretische Informationstechnik, Technische Universit\"at M\"unchen,}\\
{\footnotesize 80291 M\"unchen, Germany }
}

\maketitle

\begin{abstract}We investigate entanglement transmission over an unknown channel in the presence of a third party (called the adversary), which is enabled to choose the 
channel from a given set of memoryless but non-stationary channels without informing the legitimate sender and receiver about the particular choice that he made. 
This channel model is called an arbitrarily varying quantum channel (AVQC).\\
We derive a quantum version of Ahlswede's dichotomy for classical arbitrarily varying channels. This includes a regularized formula for the common randomness-assisted capacity for entanglement transmission of an AVQC. Quite surprisingly and in contrast to the classical analog of the problem involving the maximal and average error probability, we find that the capacity for entanglement transmission of an AVQC always equals its strong subspace transmission capacity.\\
These results are accompanied by different notions of symmetrizability (zero-capacity conditions) as well as by conditions for an AVQC to have a capacity described by a single-letter formula. 
In the final part of the paper the capacity of the erasure-AVQC is computed and some light shed on the connection between AVQCs and zero-error capacities. Additionally, we show by entirely elementary and operational arguments motivated by the theory of AVQCs that the quantum, classical, and entanglement-assisted zero-error capacities of quantum channels are generically zero and are discontinuous at every positivity point.
\end{abstract}

\tableofcontents
\section{\label{sec:Introduction}Introduction}
System uncertainty is a basic feature of many information processing systems, regardless whether classical or quantum mechanical, which has a significant impact on structure and performance of protocols used to cope with limited system knowledge. While in case of quantum communication through  unknown quantum channels several important techniques, including channel detection and quantum channel tomography, have been developed to gain at least partial system knowledge, the assumptions needed for these techniques to work satisfactory seem to be rather limiting. Either involved channels have to be stationary (and memoryless) or there has to be additional assistance by a noiseless two-way classical side channel of potentially unlimited capacity or both.\\
An alternative approach consists of following the successful paradigm of classical information theory according to which one develops the techniques, tailored to clearly specified channel models, for identification of optimal communication parameters, e.g. achievable rates, without making any attempt to reduce system uncertainty. However, channel detection and/or tomography can be seen as auxiliary steps which help to specify the channel model prior to actual information processing. In this paper we follow this information-theoretic strategy, and consider the problem of entanglement transmission over adversarially selected quantum channels. \\
\newline
The basic setup consists of a set of quantum channels $\fri=\{\cn_s  \}_{s\in\bS}$ which is known to both the sender and receiver. The goal of the sender is to transmit one half of a maximally entangled pure state $\psi$, suitably encoded, by $l$-fold usage of the (unknown) channel. An entity, which we call the adversary for simplicity, can choose a sequence $s^l=(s_1,\ldots,s_l)\in\bS^l$ at her/his will which results in the selection of the channel $\cn_{s^l}=\otimes_{i=1}^l\cn_{s_i}$. The encoded version of $\psi$ is then fed into $\cn_{s^l}$ and the receiver's goal is to recover the input state, of course without knowing the sequence $s^l$ being selected by the adversary. Implicit in this informal description of the communication scenario is that we suppose that the adversary knows the code which is used for entanglement transmission. Therefore, the communicators are forced to use entanglement transmission protocols that are reliable for the whole family $\fri^{(l)}=\{ \cn_{s^l}  \}_{s^l\in\bS^l}$ of memoryless and partly non-stationary channels. In other words, the desired entanglement transmission protocol should be resistant to the effect of arbitrarily varying noise represented by the family $\fri^{(l)}=\{\cn_{s^l}  \}_{s^l\in\bS^l}$. Even in the simplest non-trivial case of a finite set $\fri$ with $|\fri|>1$ we have to deal for each block length $l$ with exponentially many quantum channels simultaneously.\\
The main contribution of this paper is a generalization of Ahlswede's dichotomy \cite{ahlswede-elimination} which can be stated as follows:\\
{\bf First}, the common-randomness-assisted entanglement transmission capacity of the AVQC $(\fri^{(l)})_{l\in\nn}$ is equal to the entanglement transmission capacity of the compound channel built up from $\conv(\fri)$, i.e. the uncountable family of stationary, memoryless channels that lie in the convex hull of $\fri$ (cf. \cite{bbn-1}, \cite{bbn-2} for more information on compound quantum channels).\\
{\bf Second}, if the deterministic capacity for transmission of messages with asymptotically vanishing average error over an AVQC is greater than zero, its capacity for transmission of entanglement with deterministic codes is equal to its common-randomness-assisted capacity for transmission of entanglement.\\\\
The proof of the direct part as well as the proof of the converse rely substantially on the corresponding results for compound quantum channels developed in \cite{bbn-1}, \cite{bbn-2}. The link between the compound and arbitrarily varying channel models needed in the achievability proofs is given by the powerful robustification technique of $\cite{ahlswede-coloring}$ and \cite{ahlswede-gelfand-pinsker}.\\
The idea behind the second part of the theorem is the following. If the deterministic capacity for message transmission, with average error probability as the success criterion, of $\fri$ is greater than zero, then sender and receiver can use a few (sub-exponentially many) bits to derandomize a given common-randomness-assisted code for transmission of entanglement.\\
As a supplement to the coding theorem, we derive a multi-letter necessary and sufficient condition for the deterministic capacity, with average error, for message transmission of a (finite) AVQC to be zero in Section \ref{sec:symmetrizability}. For sake of completeness, we also include a necessary and sufficient condition for the deterministic capacity for message transmission with \emph{maximal} error probability to be equal to zero. Moreover, we present a first attempt to derive a non-trivial sufficient condition for the common-randomness-assisted capacity for transmission of entanglement to be zero, which we call qc-symmetrizability. Our feeling in this matter is that the definition of that kind of symmetrizability is too narrow to have any chance to be necessary and sufficient. This is basically because according to that definition the adversary does not use all the freedom he is given by the channel model to prevent the common-randomness-assisted entanglement transmission.\\
We find a striking difference to the classical theory: entanglement transmission with entanglement fidelity as the criterion of success is widely acknowledged as a fully quantum counterpart to message transmission with average error as a criterion for failure of transmission, while the counterpart of strong subspace transmission should be maximal error probability.\\
The two classical criteria have been proven to be asymptotically equivalent e.g. for single memoryless channels. For transmission over an AVC they lead to different capacities, as can be seen from Example 2 in \cite{ahlswede-elimination}. The AVC given there has zero capacity for message transmission with asymptotically vanishing maximal error probability, but from Theorem 3, part a) it can be seen that it has positive capacity for message transmission with asymptotically vanishing average error.\\
In the quantum case, asymptotic equivalence of entanglement and strong subspace transmission for single quantum channels has already been proven in \cite{barnum-knill-nielsen}. Our results show, that they are - in contrast to the classical theory - also (asymptotically) equivalent criteria w.r.t. AVQCs.\\
It is no surprise then, that the connection between arbitrarily varying channels and zero-error capacities that is valid in the classical case \cite{ahlswede-note} only partly survives in the quantum regime. This connection is explored in the last part of the paper. Additionally, we show that quantum, classical, and entanglement-assisted zero-error capacities of quantum channels are generically zero and are discontinuous at every positivity point. This is obvious for the classical zero-error capacity in Shannon's original setting \cite{shannon}. In the quantum case we employ some simple facts from convex geometry combined with methods motivated by the theory of arbitrarily varying channels to obtain this conclusion in an extremely simple way directly from the corresponding definitions of zero-error quantum capacities. It should be mentioned at this point that these results can as well be obtained rather easily using the concept of non-commutative graphs (again accompanied by some convex geometry) that has been systematically explored in the recent work \cite{duan-severini-winter}. The fact that the quantum zero-error capacity is generically zero shows that the channels for which it is possible to satisfy the Knill-Laflamme condition \cite{knill-laflamme} on a subspace of dimension greater or equal than $2$ are exceptional.\\
We also list two properties that lead to a single-letter capacity formula of an AVQC and compute the (deterministic) entanglement transmission capacity of an erasure AVQC.\\
\subsection{\label{subsec:Related Work}Related Work}
The model of an arbitrarily varying channel has been introduced by Blackwell, Breiman and Thomasian \cite{bbt-avc} in 1960. They derived a formula for the capacity of an AVC with random codes and asymptotically vanishing average error probability. They also wrote down an explicit example of an AVC whose deterministic capacity is zero, while having nonzero capacity when using random codes.\\
Later landmarks in the development of coding theorems for AVCs have been the papers by Kiefer and Wolfowitz \cite{kiefer-wolfowitz}, who found a necessary and sufficient condition for an AVC to have nonzero capacity with deterministic codes and asymptotically vanishing maximal error probability.\\
The maximal error probability criterion was further investigated in \cite{ahlswede-wolfowitz-2} by Ahlswede and Wolfowitz, who completely determined the capacity of AVCs with binary output alphabet under that criterion. A solution for arbitrarily large alphabets does not seem to exist until now. It should be mentioned that such a solution would include the solution to Shannon's zero error capacity problem \cite{shannon}, as pointed out in \cite{ahlswede-note}.\\
In our approach we use the powerful elimination technique developed by the first author in 1978 \cite{ahlswede-elimination} that, together with the random coding results of \cite{bbt-avc} enabled him to prove the following dichotomy result for AVCs: It stated that the capacity of an AVC (under the average error probability criterion) is either zero or equals its random coding capacity. Together with the robustificaion technique \cite{ahlswede-coloring,ahlswede-gelfand-pinsker} of the first author, the elimination technique led to a rather streamlined approach that, in this combination, has first been successfully used by Ahlswede in \cite{ahlswede-gelfand-pinsker}.\\
After the discoveries of \cite{ahlswede-elimination}, an important open question was, when exactly the deterministic capacity with vanishing average error is equal to zero. In 1985, a first step towards a solution was made by Ericson \cite{ericson}, who came up with a sufficient condition that was proven to be necessary by Csiszar and Narayan \cite{csiszar-narayan} in 1989.\\
The model of an arbitrarily varying  channel with classical input and quantum output has first been considered in 2007 by the first author together with Blinovsky \cite{ahlswede-blinovsky}. They considered the transmission of messages under the average error criterion and gave a complete solution of the problem, i.e. a single-letter capacity formula, including a necessary and sufficient condition for the case of zero capacity.

\subsection{\label{subsec:Outline}Outline}
The notation we freely use throughout the paper is summarized in Section \ref{Notation and Conventions}. The definitions of codes and capacities that are needed in the sequel 
are given in Section \ref{sec:codes-and-capacity}. This section also contains our main result, a quantum version of Ahlswede's dichotomy.\\
A perhaps surprising result is proven in Section \ref{sec:Equivalence of strong subspace and entanglement transmission}: As can be seen from an application of a concentration 
inequality, the capacities for entanglement and strong subspace transmission are identical. This is in sharp contrast to the classical case, where their analogs - average and 
maximal error criterion - lead to different capacities \cite{ahlswede-elimination}.\\
The main part of the paper, up to Section \ref{sec:symmetrizability}, is mostly devoted to the proof of the main result and is organized as follows.\\
In Section \ref{sec:converse-avqc} we are concerned with the upper bound to the common-randomness-assisted capacity for entanglement transmission, i.e. with the converse part. Here, the basic problem is that we cannot employ the Minimax-Theorem (which has originally been proven by von Neumann \cite{von-neumann} and later put into a more general context by Kakutani \cite{kakutani}) like in the classical case to reduce the converse part to that of a single channel. We circumvent this obstacle by noting that the desired result follows from the optimal upper bound on the common-randomness-assisted capacity for entanglement transmission for \emph{compound} quantum channels. The latter is easily shown using the methods from \cite{bbn-2}.\\
Section \ref{sec:random-achievability} contains the achievability proofs for the common-randomness assisted entanglement transmission capacity $\A_{\textup{random}}(\fri)$. 
As we already mentioned above, we are in  the pleasant situation of 
having at our disposal the coding results for compound quantum channels from \cite{bbn-2} and the robustification technique from 
\cite{ahlswede-coloring, ahlswede-gelfand-pinsker}. Since the latter is a central tool for our results and because there is a short and simple proof of it in 
\cite{ahlswede-gelfand-pinsker}, we have decided to include the full account of the robustification technique. The technique on its own operates as follows: We start with 
a ``good'' code for the compound quantum channel built up from $\conv(\fri)$, then applying permutations to the encoding and decoding operations of that code we obtain a 
``good'' random code for the AVQC. The source of common randomness now helps coordinating the selection of permutations at sender's and receiver's side.\\
The last part of our main theorem is proven in Section \ref{sec:derandomization}. The first step in the proof is to show that not that much common randomness is needed to 
achieve $\A_{\textup{random}}(\fri)$. Basically for each block length $l$ we need roughly $O(\log l)$ random bits. This is shown by a slight modification of the elimination 
technique of \cite{ahlswede-elimination}. If we assume now that the capacity for transmission of classical messages using average error criterion is strictly greater than zero,
 the sender and receiver can generate the required $O(\log l)$ random bits by 
sending (classical) messages over the AVQC by spending negligible block length compared to $l$ and are therefore able to simulate reliable random codes by deterministic ones.\\
Section \ref{sec:symmetrizability} summarizes attempts to address the question when exactly a given (finite) AVQC has a capacity (for various types of transmission and 
criteria of success) equal to zero. Most important for our present work is a necessary and sufficient condition for the message transmission capacity with deterministic 
codes and average error criterion to be greater than zero. Together with our results on random entanglement transmission codes it enables us to prove our main theorem. We 
also give a necessary and sufficient condition for the classical deterministic capacity with maximal error criterion to be greater than zero and end the section with a first 
attempt to find non-trivial conditions for the common-randomness-assisted capacity for transmission of entanglement of an AVQC to be equal to zero.\\
Section \ref{sec:Conditions for single-letter-capacities} is devoted to single letter characterizations of the entanglement transmission capacity $\A_{\textup{det}}(\fri)$. 
We give two conditions that lead to single letter formulas, both of which demand certain properties to be valid all over the convex hull $\textrm{conv}(\fri)$ of $\fri$.\\
Finally in Section \ref{sec:Applications and examples} we explicitly calculate the capacities for entanglement transmission of the erasure-AVQC and exploit the connection 
between AVQC's and zero-error capacities.
\section{\label{Notation and Conventions}Notation and conventions}
All Hilbert spaces are assumed to have finite dimension and are over the field $\cc$. $\mathcal{S}(\hr)$ is the set of states, i.e. 
positive semi-definite operators with trace $1$ acting on the Hilbert space $\hr$. Pure states are given by projections onto one-dimensional subspaces. 
A vector of unit length spanning such a subspace will therefore be referred to as a state vector. If $\fr\subset \hr$ is a subspace of $\hr$ then we write $\pi_{\fr}$ 
for the maximally mixed state on $\fr$, i.e. $\pi_{\fr}=\frac{p_{\fr}}{\tr(p_{\fr})}$ where $p_{\fr}$ stands for the projection onto $\fr$. $\mathcal{B}(\hr)$ denotes 
the set of linear operators acting on $\hr$. For a finite set $A$ the notation $\mathfrak{P}(A)$ is reserved for the set of probability distributions on $A$.\\
For a given Hilbert space with inner product $\langle\cdot,\cdot\rangle$ we write $S(\hr)$ for its unit sphere: $S(\hr):=\{x\in\hr:\langle x,x\rangle=1\}$.\\
The set of completely positive trace preserving (CPTP) maps
between the operator spaces $\mathcal{B}(\hr)$ and
$\mathcal{B}(\kr)$ is denoted by $\mathcal{C}(\hr,\kr)$. It is contained in the set  $\mathcal{C}^\downarrow(\hr,\kr)$ of completely positive trace non-increasing maps from $\mathcal{B}(\hr)$ to $\mathcal{B}(\kr)$.\\
We use the base two
logarithm which is denoted by $\log$. The von Neumann entropy of
a state $\rho\in\mathcal{S}(\hr)$ is given by
\begin{equation}S(\rho):=-\textrm{tr}(\rho \log\rho).  \end{equation}
The coherent information for $\cn\in \mathcal{C}(\hr,\kr) $ and
$\rho\in\mathcal{S}(\hr)$ is defined by
\begin{equation}I_c(\rho, \cn):=S(\cn (\rho))- S( (id_{\mathcal{B}(\hr)}\otimes \cn)(|\psi\rangle\langle \psi|)  ),  \end{equation}
where $\psi\in\hr\otimes \hr$ is an arbitrary purification of the state $\rho$. Following the usual conventions we let $S_e(\rho,\cn):=S( (id_{\mathcal{B}(\hr)}\otimes \cn)(|\psi\rangle\langle \psi|)  )$ denote the entropy exchange.\\
As a measure of closeness between two states $\rho,\sigma\in\mathcal S(\hr)$ we use the fidelity $F(\rho,\sigma):=||\sqrt{\rho}\sqrt{\sigma}||^2_1$. The fidelity is symmetric in the input and for a pure state $\rho=|\phi\rangle\langle\phi|$ we have $F(|\phi\rangle\langle\phi|,\sigma)=\langle\phi,\sigma\phi\rangle$.\\
A closely related quantity is the entanglement fidelity. For $\rho\in\mathcal{S}(\hr)$ and $\cn\in
\mathcal{C}^{\downarrow}(\hr,\hr)$ it is given by
\begin{equation}F_e(\rho,\cn):=\langle\psi, (id_{\mathcal{B}(\hr)}\otimes \cn)(|\psi\rangle\langle \psi|)     \psi\rangle,  \end{equation}
with $\psi\in\hr\otimes \hr$ being an arbitrary purification of the state $\rho$.\\
We use the diamond norm $||\cdot||_\lozenge$ as a measure of closeness in the set of quantum channels, which is given by
\begin{equation}\label{def:diamond-norm}
||\cn||_{\lozenge}:=\sup_{n\in \nn}\max_{a\in \mathcal{B}(\cc^n\otimes\hr),||a||_1=1}||(id_{n}\otimes \mathcal{N})(a)||_1,   \end{equation}
where $id_n:\mathcal{B}(\cc^n)\to \mathcal{B}(\cc^n)$ is the identity channel, and $\mathcal{N}:\mathcal{B}(\hr)\to \mathcal{B}(\kr)$ is any linear map, 
not necessarily completely positive. The merits of $||\cdot||_{\lozenge}$ are due to the following facts (cf. \cite{kitaev}). First, $||\cn||_{\lozenge}=1$ for 
all $\cn\in\mathcal{C}(\hr,\kr)$. Thus, $\mathcal{C}(\hr,\kr)\subset S_{\lozenge}$, where $S_{\lozenge}$ denotes the unit sphere of the normed space 
$(\mathcal{B}(\mathcal{B}(\hr),\mathcal{B}(\kr)),||\cdot||_{\lozenge} )$. Moreover, $||\cn_1\otimes \cn_2||_{\lozenge}=||\cn_1||_{\lozenge}||\cn_2||_{\lozenge}$ 
for arbitrary linear maps $\cn_1,\cn_2:\mathcal{B}(\hr)\to \mathcal{B}(\kr) $. Finally, the supremum in (\ref{def:diamond-norm}) needs only be taken over $n$ that range over $\{1,2,\ldots,\dim\hr   \}.$\\
We further use the diamond norm to define the function $D_\lozenge(\cdot,\cdot)$ on $\{(\fri,\fri'):\fri,\fri'\subset\mathcal C(\hr,\kr)\}$, which is for 
$\fri,\fri'\subset\mathcal C(\hr,\kr)$ given by
\begin{equation}D_\lozenge(\fri,\fri'):=\max\{\sup_{\cn\in\fri}\inf_{\cn'\in\fri'}||\cn-\cn'||_\lozenge,\sup_{\cn'\in\fri'}\inf_{\cn\in\fri}||\cn-\cn'||_\lozenge\}.\end{equation}
For $\fri\subset\mathcal C(\hr,\kr)$ let $\overline{\fri}$ denote the closure of $\fri$ in $||\cdot||_\lozenge$. Then $D_\lozenge$ defines a metric on $\{\fri:\fri\subset\mathcal C(\hr,\kr),\ \fri=\bar\fri\}$ which is basically the Hausdorff distance induced by the diamond norm.\\
Obviously, for arbitrary $\fri,\fri'\subset\mathcal C(\hr,\kr)$, $D_\lozenge(\fri,\fri')\leq\epsilon$ implies that for every $\cn\in\fri$ ($\cn'\in\fri'$) there exists $\cn'\in\fri'$ ($\cn\in\fri)$ such that $||\cn-\cn'||_\lozenge\leq2\epsilon$. If $\fri=\bar\fri,\ \fri'=\bar{\fri'}$ holds we even have $||\cn-\cn'||_\lozenge\leq\epsilon$.
In this way $D_\lozenge$ gives a measure of distance between sets of channels.\\
For any set $\fri\subset \mathcal{C}(\hr,\kr) $ and $l\in\nn$ we set
\begin{equation}\fri^{\otimes l}:=\{\cn^{\otimes l}: \cn\in\fri  \}.  \end{equation}
For an arbitrary set $\bS$, $\bS^l:=\{(s_1,\ldots,s_l):s_i\in\bS\ \forall i\in\{1,\ldots,l\}\}.$ We also write $s^l$ for the elements of $\bS^l$.\\
For an arbitrary set $\fri$ of CPTP maps we denote by $\conv(\fri)$ its convex hull (see \cite{webster} for the definition) and note that in case that $\fri=\{\cn_s\}_{s\in\bS}$ is a finite set we have
\begin{equation}\label{eq:conv-hull}
 \conv(\fri)=\left\{\cn_{q}\in \mathcal{C}(\hr,\kr): \cn_q=\sum_{s\in \bS}q(s)\cn_s,\ q\in\mathfrak{P}(\bS)  \right\},
\end{equation}
an equality that we will make use of in the approximation of infinite AVQC's by finite ones.\\
Finally, we need some simple topological notions for convex sets in finite dimensional normed space $(V, ||\cdot||)$ over the field of real or complex numbers which we borrow from \cite{webster}. Let $F\subset V$ be convex. $x\in F$ is said to be a relative interior point of $F$ if there is $r>0$ such that $B(x,r)\cap \aff F\subset F$. Here $B(x,r)$ denotes the open ball of radius $r$ with the center $x$ and $\aff F$ stands for the affine hull of $F$. The set of relative interior points of $F$ is called the relative interior of $F$ and is denoted by $\ri F$.\\
The relative boundary of $F$, $\rebd F$, is the set difference between the closure of $F$ and $\ri F$.\\
For a set $A\subset V$ and $\delta\ge 0$ we define the parallel set or the blow-up $(A)_{\delta}$ of $A$ by
\begin{equation}(A)_{\delta}:=\left\{x\in V: ||x-y||\le \delta \textrm{ for some } y\in A  \right\}.  \end{equation}
\section{\label{sec:codes-and-capacity}Basic definitions and main results}
In this section we define the quantities that we will be dealing with in the rest of the paper: Arbitrarily varying quantum channels and codes for transmission of entanglement and subspaces. Since they will be of importance for our derandomization arguments, we will also include definitions of the capacities for message transmission with average and maximal error probability criterion.\\
Our most basic object is the arbitrarily varying quantum channel (AVQC). It is generated by a set $\fri=\{\cn_s  \}_{s\in \bS}$ of CPTP maps with input Hilbert space $\hr$ and output Hilbert space $\kr$ and given by the family of CPTP maps $\{\cn_{s^l}:\mathcal{B}(\hr)^{\otimes l}\to\mathcal{B}(\kr)^{\otimes l}  \}_{l\in\nn,s^l\in \bS^{l}}$, where
\begin{equation}\cn_{s^l}:=\cn_{s_1}\otimes \ldots\otimes \cn_{s_l}\qquad \qquad (s^l\in\bS^l).  \end{equation}
Thus, even in the case of a finite set $\fri=\{\cn_s  \}_{s\in\bS}$, showing the existence of reliable codes for the AVQC determined by $\fri$ is a non-trivial task: 
For each block length $l\in\nn$ we have to deal with $|\fri|^{l}$, i.e. exponentially many, memoryless partly non-stationary quantum channels simultaneously.\\\\
In order to relieve ourselves from the burden of complicated notation we will simply write $\fri=\{\cn_s \}_{s\in\bS}$ for the AVQC.
\subsection{\label{subsec:Entanglement-transmission}Entanglement transmission}
For the rest of this subsection, let $\fri=\{\cn_s\}_{s\in\bS}$ be an AVQC.
\begin{definition}
An $(l,k_l)-$\emph{random entanglement transmission code} for $\fri$ is a probability measure $\mu_l$ on $(\mathcal C(\fr_l,\hr^{\otimes l})\times\mathcal C(\kr^{\otimes l},\fr_l'),\sigma_l)$, 
where $\fr_l,\fr_l'$ are Hilbert spaces, $\dim\fr_l=k_l$, $\fr_l\subset\fr_l'$ and the sigma-algebra $\sigma_l$ is chosen such that the function $(\cP_l,\crr_l)\mapsto F_e(\pi_{\fr_l},\crr_l\circ\cn_{s^l}\circ\cP_l)$ is measurable 
w.r.t. $\sigma_l$ for every $s^l\in\bS^l$.\\
Moreover, we assume that $\sigma_l$ contains all singleton sets. An example of such a sigma-algebra $\sigma_l$ is given by 
the product of sigma-algebras of Borel sets induced on $\mathcal C(\fr_l,\hr) $ and $\mathcal C(\kr,\fr_l') $ by the standard topologies of the ambient spaces.
\end{definition}
\begin{definition}\label{def:random-cap-ent-trans}
A non-negative number $R$ is said to be an achievable entanglement transmission rate for the AVQC $\fri=\{\cn_s  \}_{s\in\bS}$ with random codes if there is a sequence of $(l,k_l)-$random entanglement transmission codes such that
\begin{enumerate}
\item $\liminf_{l\rightarrow\infty}\frac{1}{l}\log k_l\geq R$ and
\item $\lim_{l\rightarrow\infty}\inf_{s^l\in\bS^l}\int F_e(\pi_{\fr_l},\crr^l\circ\cn_{s^l}\circ\cP^l)d\mu_l(\cP^l,\crr^l)=1$.
\end{enumerate}
The random entanglement transmission capacity $\A_{\textup{random}}(\fri)$ of $\fri$ is defined by
\begin{equation}\A_{\textup{random}}(\fri):=\sup\{R:R \textrm{ is an achievable entanglement transmission rate for } \fri \textrm{ with random codes}\}.\end{equation}
\end{definition}
Having defined random codes and random code capacity for entanglement transmission we are in the position to introduce their deterministic counterparts: An $(l,k_l)-$code for entanglement transmission over $\fri$ is an $(l,k_l)-$random code for $\fri$ with $\mu_l(\{(\mathcal{P}^l,\crr^l)  \}  )=1$ for some encoder-decoder pair $(\mathcal{P}^l,\crr^l)$\footnote{This explains our requirement on $\sigma_l$ to contain all singleton sets.} and $\mu_l(A)=0$ for any $A\in\sigma_l$ with $(\mathcal{P}^l,\crr^l)\notin A $. We will refer to such measures as point measures in what follows.
\begin{definition}
A non-negative number $R$ is a deterministically achievable entanglement transmission rate for the AVQC $\fri=\{\cn_s  \}_{s\in \bS}$ if it is achievable in the sense of Definition \ref{def:random-cap-ent-trans} for random codes  with \emph{point measures} $\mu_l$.\\
The deterministic entanglement transmission capacity $\A_{\textup{det}}(\fri)$ of $\fri$ is given by
\begin{equation}\A_{\textup{det}}(\fri):=\sup \{R: R \textrm{ is a deterministically achievable entanglement transmission rate for }\fri  \}.  \end{equation}
\end{definition}
Finally, we shall need the notion of the classical deterministic capacity $C_{\textrm{det}}(\fri)$ of the AVQC $\fri=\{\cn_s  \}_{s\in\bS}$ with \emph{average} error criterion.
\begin{definition}\label{def:message-trans-with-average-error}
An $(l,M_l)$-(deterministic) code for message transmission is a family of pairs $\mathfrak{C}_l=(\rho_i, D_i)_{i=1}^{M_l}$
where $\rho_1,\ldots ,\rho_{M_l}\in\cs(\hr^{\otimes l})$, and positive semi-definite operators $D_1,\ldots, D_{M_l}\in\mathcal{B}(\kr^{\otimes l})$ satisfying $\sum_{i=1}^{M_l}D_i=\idn_{\kr^{\otimes l}} $.\\
The worst-case average probability of error of a code $\mathfrak{C}_l$ is given by
\begin{equation}\label{def-av-error}
\bar P_{e,l}(\fri):=\sup_{s^l\in \bS^l}\bar P_e (\mathfrak{C}_l,s^l),
\end{equation}
where for $s^l\in\bS^l$ we set
\begin{equation} \bar P_e (\mathfrak{C}_l,s^l):=\frac{1}{M_l}\sum_{i=1}^{M_l}\left(1- \tr(\cn_{s^l}(\rho_i)D_i)  \right).  \end{equation}
The achievable rates and the classical deterministic capacity  $C_{\textrm{det}}(\fri)$ of $\fri$, with respect to the error criterion given in (\ref{def-av-error}), are then defined in the usual way.\\
\end{definition}
For any AVQC (finite or infinite), the compound quantum channel generated by the set $\conv(\fri)$ (cf. \cite{bbn-2} for the relevant definition) shall play the crucial role in our derivation of the coding results below. 
In the relevant cases we will have $|\fri|>1$ and, therefore, $\conv(\fri)$ will be \emph{infinite}.\\
Our main result, a quantum version of Ahlswede's dichotomy for finite AVQCs, goes as follows:
\begin{theorem}\label{quant-ahlswede-dichotomy}
Let $\fri=\{\cn_s  \}_{s\in \bS}$ be an AVQC.
\begin{enumerate}
\item With $\textrm{conv}(\fri)$ denoting the convex hull of $\fri$ we have
  \begin{equation}\label{eq:ahlswede-dichotomy-1}
    \A_{\textup{random}}(\fri)=\lim_{l\to\infty}\frac{1}{l}\max_{\rho\in\cs(\hr^{\otimes l})}\inf_{\cn\in \conv(\fri)}I_c(\rho, \cn^{\otimes l}).
  \end{equation}
\item Either $C_{\textup{det}}(\fri)=0 $ or else $\A_{\textup{det}}(\fri)= \A_{\textup{random}}(\fri)$.
\end{enumerate}
\end{theorem}
\begin{proof} 
The claim made in (\ref{eq:ahlswede-dichotomy-1}) follows from Theorem \ref{converse:finite-avc} and Corollary \ref{achievability-finite-avqc}.\\
The proof that $C_{\textup{det}}(\fri)>0$ implies $\A_{\textup{det}}(\fri)= \A_{\textup{random}}(\fri) $ requires a derandomization argument which is presented in Section \ref{sec:derandomization}.
\end{proof}
We conclude this section with some explaining remarks:\\
1. Coherent information depends continuously on the state, therefore $\rho\mapsto\inf_{\cn\in \conv(\fri)}I_c(\rho, \cn^{\otimes l})$ is upper semicontinuous 
and thus $\max_{\rho\in\cs(\hr^{\otimes l})}\inf_{\cn\in \conv(\fri)}I_c(\rho, \cn^{\otimes l})$ exists due to the compactness of $\cs(\hr^{\otimes l})$.\\
The limit in (\ref{eq:ahlswede-dichotomy-1}) exists due to superadditivity of the sequence $(\max_{\rho\in\cs(\hr^{\otimes l})}\inf_{\cn\in \conv(\fri)}I_c(\rho, \cn^{\otimes l}))_{l\in\nn}$.\\
2. It is clear that $\A_{\textup{det}}(\fri)\le C_{\textup{det}}(\fri) $, so that $C_{\textup{det}}(\fri)=0 $ implies  $\A_{\textup{det}}(\fri)=0 $. Therefore, Theorem \ref{quant-ahlswede-dichotomy} gives a regularized formula for  $\A_{\textup{det}}(\fri) $ in form of (\ref{eq:ahlswede-dichotomy-1}), and the question remains when $C_{\textup{det}}(\fri)=0$ happens. We derive a non-single-letter necessary and sufficient condition for the latter in Section \ref{sec:symmetrizability}.\\
3. Continuous dependence of the coherent information on the channel reveals that for each $l\in\nn$ and $\rho\in\cs(\hr^{\otimes l}) $
\begin{equation} \inf_{\cn\in \conv(\fri)}I_c(\rho, \cn^{\otimes l})=\min_{\cn\in \overline{\conv(\fri)}}I_c(\rho, \cn^{\otimes l}), \end{equation}
\subsection{\label{subsec:strong-subspace-transmission}Strong subspace transmission}
Let $\fri=\{\cn_s\}_{s\in\bS}$ be an AVQC. An $(l,k_l)-$\emph{random strong subspace transmission code} for $\fri$ is a probability measure $\mu_l$ on 
$(\mathcal C(\fr_l,\hr^{\otimes l})\times\mathcal C(\kr^{\otimes l},\fr_l'),\sigma_l)$, where $\fr_l,\ \fr_l'$ are Hilbert spaces, $\dim\fr_l=k_l$, $\fr_l\subset\fr_l'$ and the sigma-algebra 
$\sigma_l$ is chosen such that for every $\psi\in S(\fr_l)$ the function $(\cP^l,\crr^l)\mapsto F(|\psi\rangle\langle\psi|,\crr^l\circ\cn_{s^l}\circ\cP^l(|\psi\rangle\langle\psi|))$ 
is measurable w.r.t. $\sigma_l$ for every $s^l\in\bS^l$. Again, we assume that $\sigma_l$ contains all singleton sets.
\begin{definition}\label{def:random-cap-strsub-trans}
A non-negative number $R$ is said to be an achievable strong subspace transmission rate for the AVQC $\fri=\{\cn_s  \}_{s\in\bS}$ with random codes if there is a sequence of $(l,k_l)-$random strong subspace transmission codes such that
\begin{enumerate}
\item $\liminf_{l\rightarrow\infty}\frac{1}{l}\log k_l\geq R$ and
\item $\lim_{l\rightarrow\infty}\inf_{s^l\in\bS^l} \min_{\psi\in S(\fr_l)}\int F(|\psi\rangle\langle\psi|,\crr^l\circ\cn_{s^l}\circ\cP^l(|\psi\rangle\langle\psi|))d\mu_l(\cP^l,\crr^l)=1$.
\end{enumerate}
The random strong subspace transmission capacity $\A_{\textup{s,random}}(\fri)$ of $\fri$ is defined by
\begin{equation}\A_{\textup{s,random}}(\fri):=\sup\{R:R \textrm{ is an achievable strong subspace transmission rate for } \fri \textrm{ with random codes}\}.\end{equation}
\end{definition}
As before we also define deterministic codes: A \emph{deterministic} $(l,k_l)-$strong subspace transmission code for $\fri$ is an $(l,k_l)-$random strong subspace transmission code for $\fri$ with $\mu_l(\{(\mathcal{P}^l,\crr^l)  \}  )=1$ for some encoder-decoder pair $(\mathcal{P}^l,\crr^l)$ and $\mu_l(A)=0$ for any $A\in\sigma_l$ with $(\mathcal{P}^l,\crr^l)\notin A $. We will refer to such measures as point measures in what follows.
\begin{definition}
A non-negative number $R$ is a deterministically achievable strong subspace transmission rate for the AVQC $\fri=\{\cn_s  \}_{s\in \bS}$ if it is achievable in the sense of Definition \ref{def:random-cap-strsub-trans} for random codes  with \emph{point measures} $\mu_l$.\\
The deterministic capacity $\A_{\textup{s,det}}(\fri)$ for strong subspace transmission over an AVQC $\fri$ is  given by
\begin{equation}\A_{\textup{s,det}}(\fri):=\sup \{R: R \textrm{ is a deterministically achievable rate for }\fri  \}.  \end{equation}
\end{definition}
If we want to transmit classical messages, then the error criterion that is most closely related to strong subspace transmission is that of maximal error probability. It leads to the notion of classical deterministic capacity with maximal error:
\begin{definition}\label{def:message-trans-with-max-error}Let $\mathfrak C_l$ be an $(l,M_l)$-(deterministic) code for message transmission as given in Definition \ref{def:message-trans-with-average-error}. The worst-case maximal probability of error of the code $\mathfrak{C}_l$ is given by
\begin{equation}\label{def-av-maxerror}
P_{e,l}(\fri):=\sup_{s^l\in \bS^l} P_e (\mathfrak{C}_l,s^l),
\end{equation}
where for $s^l\in\bS^l$ we set
\begin{equation} P_e (\mathfrak{C}_l,s^l):=\max_{i\in M_l}\left(1- \tr(\cn_{s^l}(\rho_i)D_i)  \right).  \end{equation}
The achievable rates and the classical deterministic capacity  $C_{\det,\max}(\fri)$ of $\fri$, with respect to the error criterion given in (\ref{def-av-maxerror}), are then defined in the usual way.\\
\end{definition}
The perhaps surprising result is that the strong subspace transmission capacity of a (finite) AVQC always equals its entanglement transmission capacity:
\begin{theorem}\label{theorem:equivalence-of-capacities}
For every AVQC $\fri=\{\cn_s\}_{s\in\bS}$ we have the equalities
\begin{align}
\A_{\textup{s,random}}(\fri)&=\A_{\textup{random}}(\fri),&\\
\A_{\textup{s,det}}(\fri)&=\A_{\textup{det}}(\fri).&
\end{align}
\end{theorem}
\subsection{\label{subsec:Zero-error capacities}Zero-error capacities}
In this subsection we only give definitions of zero-error capacities. 
Through the ideas of \cite{ahlswede-note} these capacities are connected to arbitrarily varying channels, 
though this connection is not as strong as in the classical setting.\\
Results concerning these capacities are stated in subsections \ref{subsec:Qualitative behavior of zero-error capacities} and \ref{subsec:lovasz}.
\begin{definition}
An $(l,k)$ zero-error quantum code (QC for short) $(\fr,\cP,\crr)$ for $\cn\in \mathcal{C}(\hr,\kr) $ consists of a Hilbert space $\fr$, $\cP \in \mathcal{C}(\fr,\hr^{\otimes l})$, $\crr\in \mathcal{C}(\kr^{\otimes l},\fr)$ with $\dim \fr=k$ such that
\begin{equation}\label{def-q-0-code}
  \min_{x\in\fr, ||x||=1}\langle x, \crr\circ \cn^{\otimes l}\circ \cP(|x\rangle\langle x|) x\rangle =1.
\end{equation}
For fixed block length $l\in\nn$ define
\begin{equation}\label{q-ind-number}
  k(l,\cn):=\max\{k : \exists (l,k) \textrm{ zero-error QC for }\cn \}. 
\end{equation}
The zero-error quantum capacity $Q_0(\cn)$ of $\cn\in\mathcal{C}(\hr,\kr)$ is then defined by
\begin{equation}\label{0-error-q-capacity}
  Q_0(\cn):=\lim_{l\to\infty}\frac{1}{l}\log k(l,\cn).
\end{equation}
The existence of the limit follows from standard arguments based on Fekete's Lemma on subadditive sequences. 
\end{definition}
\begin{remark}
Fekete's Lemma seems to have explicitly appeared first in (\cite{polya-szegoe}, page 23). 
The authors of \cite{polya-szegoe} cite Fekete's work (\cite{fekete}, page 233, Satz II) as a special case of their statement 98. We adopt today's convention of naming it
Fekete's Lemma, although there seems to be no clear way of attributing it to any specific group of authors.
\end{remark}

Next we pass to the zero-error classical capacities of quantum channels.\\
\begin{definition}
Let $\sigma_{\fr\fr'}$ be a bipartite state on $\fr\otimes \fr'$ where $\fr'$ denotes a unitary copy of the Hilbert space $\fr$. An $(l,M)$ entanglement assisted code (ea-code for short) $(\sigma_{\fr\fr'},\{\cP_m, D_m  \}_{m=1}^M )$ consists of a bipartite state $\sigma_{\fr\fr'}$, $\cP_m\in\mathcal{C}(\fr,\hr^{\otimes l })$, $m=1,\ldots, M$, and a POVM $\{D_m \}_{m=1}^M$ on $\fr'\otimes \kr^{\otimes l}$. A given $(l,M)$ entanglement assisted code $(\sigma_{\fr\fr'},\{\cP_m, D_m  \}_{m=1}^M )$ is a zero-error code for $\cn\in \mathcal{C}(\hr,\kr)$ if
\begin{equation}\label{ea-0-error-code}
  \tr ((\cn^{\otimes l}\circ \cP_m \otimes     \textrm{id}_{\fr'})(\sigma_{\fr\fr'})D_m )=1
\end{equation}
holds for all $m\in [M]:=\{1,\ldots,M\}$. For $l\in\nn$ we set
\begin{equation}\label{ea-ind-number}
  M_{\textrm{EA}}(l,\cn):=\max\{M: \exists \ \textrm{ zero-error } (l,M) \textrm{ ea-code for } \cn   \}.
\end{equation}
\end{definition}
\begin{definition}
The entanglement assisted classical zero-error capacity $C_{0\textrm{EA}}(\cn)$ of $\cn\in\mathcal{C}(\hr,\kr)$ is given by
\begin{equation}\label{ea-0-error-capacity}
  C_{0\textrm{EA}}(\cn):=\lim_{l\to\infty}\frac{1}{l}\log M_\textrm{EA}(l,\cn).
\end{equation}
\end{definition}
If we restrict the definition of zero-error ea-code to states $\sigma_{\fr\fr'}$ with $\dim \fr=\dim \fr'=1$ we obtain the perfomance parameter $M(l,\cn)$ as a special case of $M_{\textrm{EA}}(l,\cn)$ in (\ref{ea-ind-number}) and the classical zero-error capacity $C_{0}(\cn)$ of a quantum channel $\cn$.\\
\begin{definition}
Given a bipartite state $\rho\in\cs (\hr_A\otimes \hr_B)$. An $(l,k_l)$ zero-error entanglement distillation protocol (EDP for short) for $\rho$ consists of an LOCC operation $\D \in\mathcal{C}(\hr_A^{\otimes l}\otimes\hr_B^{\otimes l}, \cc^{k_l}\otimes \cc^{k_l})$ and a maximally entangled state vector $\vphi_{k_l}=\frac{1}{\sqrt{k_l}}\sum_{i=1}^{k_l}e_i\otimes e_i\in \cc^{k_l}\otimes \cc^{k_l}$  with an orthonormal basis $\{e_1,\ldots, e_{k_l}  \}$ of $\cc^{k_l}$ such that
\begin{equation}\label{distillation-1}
  \langle \vphi_{k_l}, \D (\rho^{\otimes l})\vphi_{k_l}\rangle=1.
\end{equation}
Let for $l\in\nn$
\begin{equation}\label{distillation-2}
  d(l,\rho):=\max\{k_l: \exists (l,k_l) \textrm{ zero-error EDP for } \rho  \},
\end{equation}
and we define the zero-error distillable entanglement of $\rho\in \cs (\hr_A\otimes \hr_B) $ as
\begin{equation}\label{distillation-3}
  D_0(\rho):=\lim_{l\to\infty}\frac{1}{l}\log  d(l,\rho).
\end{equation}
 \end{definition}
\section{\label{sec:Equivalence of strong subspace and entanglement transmission}Equivalence of strong subspace and entanglement transmission}
We will now use results from convex high-dimensional geometry to show that every sequence of asymptotically perfect (random) codes for entanglement transmission for $\fri$ yields another sequence of (random) codes that guarantees asymptotically perfect strong subspace transmission.\\
First, we state the following theorem which is the complex version of a theorem that can essentially be picked up from \cite{milman-schechtman}, Theorem 2.4 and Remark 2.7:

\begin{theorem}\label{theorem-milman-schechtman}
For $\delta,\Theta>0$ and an integer $n$ let $k(\delta,\Theta,n)=\lfloor \delta^2(n-1)/(2\log(4/\Theta))\rfloor$. Let $f:S(\mathbb C^{n})\rightarrow\mathbb R$ be a continuous 
function and $\nu_k$ the uniform measure induced on the Grassmannian $G_{n,k}:=\{G\subset\mathbb C^n: G\textrm{\ is\ subspace\ and\ }\dim G=k\}$ 
by the normalized Haar measure on the unitary group on $\mathbb C^{n}$ then, for all $\delta,\Theta>0$, the measure of the set $E_k\subset G_{n,k}$ of all subspaces 
$E\subset\mathbb C^{n}$ satisfying the three conditions
\begin{enumerate}
\item $\dim E=k(\delta,\Theta,n)$
\item There is a $\Theta-$net $N$ in $S(E)=S(\mathbb C^{n})\bigcap E$ such that $|f(x)-M_f|\leq\omega_f(\delta)$ for all $x\in N$
\item $|f(x)-M_f|\leq\omega_f(\delta)+\omega_f(\Theta)$ for all $x\in S(E)$
\end{enumerate}
satisfies $\nu_k(E_k)\geq1-\sqrt{2/\pi}e^{-\delta^2(n-1)/2}$.
\\\\
Here, $S(\mathbb C^{n})$ is the unit sphere in $\mathbb C^{n}$, $\omega_f(\delta):=\sup\{|f(x)-f(y)|:D(x,y)\leq\delta\}$ is the modulus of continuity, $D$ the geodesic metric on $S(\mathbb C^{n})$ and $M_f$ the median of $f$, which is the number such that with $\nu$ the Haar measure on $S(\mathbb C^{n})$ both $\nu(\{x:f(x)\leq M_f\})\geq1/2$ and $\nu(\{x:f(x)\geq M_f\})\geq1/2$ hold.
\end{theorem}
\begin{remark}The proof of Theorem \ref{theorem-milman-schechtman} uses the identification $\mathbb C^n\simeq\mathbb R^{2n}$ under the map $\sum_{i=1}^nz_ie_i\mapsto\sum_{i=1}^n(\mathfrak{Re}\{z_i\}e_i+\mathfrak{Im}\{z_i\}e_{i+n})$, where $\{e_1,\ldots,e_n\}$ and $\{e_1,\ldots,e_{2n}\}$ denote the standard bases in $\mathbb C^n$ and $\mathbb R^{2n}$.
\end{remark}
Second, we use the following lemma which first appeared in \cite{horodecki-horodecki-horodecki}:
\begin{lemma}\label{lemma-connection-between-strong-subspace-and-entanglement-fidelity}
Let $\cn\in\mathcal C(\gr,\hr)$, $\gr$ a $d$-dimensional subspace of $\hr$ and $\phi\in\gr$ with euclidean norm $\|\phi\|=1$. Then
\begin{equation}
\int_{\mathfrak U(\gr)} \langle U\phi,\cn(U|\phi\rangle\langle\phi|U^\ast)U\phi\rangle dU=\frac{1}{d+1}(d\cdot F_e(\pi_\gr,\cn)+1).
\end{equation}
\end{lemma}
Third, we need a well behaving relation between the median and the expectation of a function $f:S(\mathbb C^n)\rightarrow\mathbb R$. 
This is given by Proposition 14.3.3 taken from \cite{matousek}:
\begin{lemma}\label{lemma-connection-between-expectation-and-median}
Let $f:S(\mathbb C^n)\rightarrow\mathbb R$ be Lipschitz with constant one (w.r.t. the geodesic metric). Then
\begin{equation}|M_f-\mathbb E(f)|\leq\frac{12}{\sqrt{2(n-1)}}.\end{equation}
\end{lemma}
\begin{remark}
Obviously, this implies $|M_f-\mathbb E(f)|\leq\frac{12\cdot L}{\sqrt{2(n-1)}}$ for Lipschitz functions with constant $L\in\mathbb R_+$.
\end{remark}
The function that we will apply Lemma \ref{lemma-connection-between-expectation-and-median} to is given by the following:
\begin{lemma}\label{lemma-lipschitz-property-of-Fe}
Let $\Lambda\in\mathcal C(\hr,\hr)$. Define $f_\Lambda:S(\hr)\rightarrow\mathbb R$ by
\begin{equation}f_\Lambda(x):=\langle x,\Lambda(|x\rangle\langle x|)x\rangle,\ x\in S(\hr).\end{equation}
Then $f_\Lambda$ is Lipschitz with constant $L=4$ (w.r.t. the geodesic metric).
\end{lemma}
\begin{proof} Let $x,y\in S(\hr)$. Then by H\"older's inequality,
\begin{align}
|f_\Lambda(x)-f_\Lambda(y)|&=|\tr(|x\rangle\langle x|\Lambda(|x\rangle\langle x|))-\tr(|y\rangle\langle y|\Lambda(|y\rangle\langle y|))|&\\
&=|\tr(|x\rangle\langle x|\Lambda(|x\rangle\langle x|-|y\rangle\langle y|))|+|\tr((|x\rangle\langle x|-|y\rangle\langle y|)\Lambda(|y\rangle\langle y|))|&\\
&\leq\|\ |x\rangle\langle x|\ \|_{\infty}\cdot\|\Lambda(|x\rangle\langle x|-|y\rangle\langle y|)\|_1+\|\ |x\rangle\langle x|-|y\rangle\langle y|\ ||_1\cdot\|\Lambda(|y\rangle\langle y|)\|_\infty&\\
&\leq\|\Lambda(|x\rangle\langle x|-|y\rangle\langle y|)\|_1+\||x\rangle\langle x|-|y\rangle\langle y|\|_1&\\
&\leq2\|\ |x\rangle\langle x|-|y\rangle\langle y|\ \|_1.&
\end{align}
It further holds, with $\|\cdot\|$ denoting the euclidean norm,
\begin{align}
\|\ |x\rangle\langle x|-|y\rangle\langle y|\ \|_1\leq2\|x-y\|\leq2 D(x,y).
\end{align}
\end{proof}
We now state the main ingredient of this section.
\begin{lemma}\label{theorem-strong-subspace-codes}
Let $\fri=\{\cn_s\}_{s\in\bS}$ be a finite set of channels and $(\mu_l)_{l\in\nn}$ any sequence of (random or deterministic) entanglement transmission codes that satisfies
\begin{itemize}
\item[A1] $\min_{s^l\in\bS^l}\int F_e(\pi_{\fr_l},\crr^l\circ\cn_{s^l}\circ\cP^l)d\mu_l(\crr^l,\cP^l)=1-f_l$,
\end{itemize}
where $(f_l)_{l\in\nn}$ is any sequence of real numbers in the interval $[0,1]$.\\
Let $(\eps_l)_{l\in\nn}$ be a sequence with $\eps_l\in(0,1]\ \forall l\in\nn$ satisfying
\begin{itemize}
\item[A2] There is $\hat l\in \nn$ such that $|\bS|^l\sqrt{2/\pi}e^{-\eps_l^2(k_l-1)/128}<1$ and $k_l\geq2$ hold for all $l\geq\hat l$\\ (where, as usual, $k_l=\dim\fr_l$).
\end{itemize}
Then for any $l\geq\hat l$ there is a subspace $\hat\fr_l\subset\fr_l$ with the properties
\begin{itemize}
\item[P1] $\dim\hat\fr_l=\lfloor\frac{\eps_l^2}{256\log(32/\eps_l)}\cdot k_l\rfloor$,
\item[P2] $\min_{s^l\in\bS^l}\min_{\phi\in S(\hat\fr_l)}\int\langle\phi,\crr^l\circ\cn_{s^l}\circ\cP^l(|\phi\rangle\langle\phi|)\phi\rangle d\mu_l(\crr^l,\cP^l)\geq1-f_l-\frac{4\cdot12}{\sqrt{2(k_l-1)}}-\eps_l$.
\end{itemize}
\end{lemma}
\begin{proof} Let $l\in\nn$. For an arbitrary $s^l\in\bS^l$ define $f_{s^l}:S(\fr_l)\rightarrow\mathbb R$ by
\begin{equation}
f_{s^l}(\phi):=\int\langle\phi,\crr^l\circ\cn_{s^l}\circ\cP^l(|\phi\rangle\langle\phi|)\phi\rangle d\mu_l(\crr^l,\cP^l)\ \ (\phi\in S(\fr_l)).
\end{equation}
Since $f_{s^l}$ is an affine combination of functions with Lipschitz-constant $L=4$, it is itself Lipschitz with $L=4$.\\
Also, by the Theorem of Fubini, Lemma \ref{lemma-connection-between-strong-subspace-and-entanglement-fidelity} and our assumption \emph{A1} we have
\begin{align}
\mathbb E(f_s{^l})&=\int_{\mathfrak{U}(\fr_l)} f_{s^l}(U\phi)dU&\\
&=\int_{\mathfrak{U}(\fr_l)} [\int\langle U\phi,\crr^l\circ\cn_{s^l}\circ\cP^l(|U\phi\rangle\langle U\phi|)U\phi\rangle d\mu_l(\crr^l,\cP^l)]dU&\\
&=\int\int_{\mathfrak{U}(\fr_l)} [\langle U\phi,\crr^l\circ\cn_{s^l}\circ\cP^l(|U\phi\rangle\langle U\phi|)U\phi\rangle dU]d\mu_l(\crr^l,\cP^l)&\\
&=\int\frac{k_lF_e(\pi_{\fr_l},\crr^l\circ\cn_{s^l}\circ\cP^l)+1}{k_l+1}d\mu_l(\crr^l,\cP^l)&\\
&\geq\int F_e(\pi_{\fr_l},\crr^l\circ\cn_{s^l}\circ\cP^l)d\mu_l(\crr^l,\cP^l)&\\
&=1-f_l.&
\end{align}
By Lemma \ref{lemma-connection-between-expectation-and-median} and Lemma \ref{lemma-lipschitz-property-of-Fe} we now get a good lower bound on the median of $f_{s^l}$:
\begin{align}
M_{f_{s^l}}&\geq\mathbb E(f_{s^l})-\frac{4\cdot12}{\sqrt{2(k_l-1)}}&\\
&\geq1-f_l-\frac{4\cdot12}{\sqrt{2(k_l-1)}}.\label{eqn-theorem-strong-subspace-capacity-1}&
\end{align}
We now apply Theorem \ref{theorem-milman-schechtman} with $n=k_l$ and $\delta=\Theta=\eps_l/8$ to $f_{s^l}$. 
Then $k(\eps_l/8,\eps_l/8,k_l)\geq\lfloor\frac{\eps_l^2}{256\log(32/\eps_l)}k_l\rfloor$ holds due to the second estimate in \emph{A2}.\\
We set $k_l':=\lfloor\frac{\eps_l^2}{256\log(32/\eps_l)}k_l\rfloor$.\\
Since the fact that $f_{s^l}$ is $4$-Lipschitz implies $\omega_{f_{s^l}}(\delta)\leq 4 \delta$ we get the following:
\begin{align}
\nu_k(\{E\in G_{k_l,k(\eps_l/8,\eps_l/8,k_l)}:|f_{s^l}(\phi)-M_{f_{s^l}}|\leq\eps_l\ \forall \phi\in S(E)\})\geq1-\sqrt{2/\pi}e^{-\eps_l^2(k_l-1)/128}.
\end{align}
The last inequality is valid for each choice of $s^l$, so we can conclude that
\begin{align}
\nu_k(\{E\in G_{k_l,k(\frac{\eps_l}{8},\frac{\eps_l}{8},k_l)}:|f_{s^l}(\phi)-M_{f_{s^l}}|\leq\eps_l\ \forall \phi\in S(E),\ s^l\in\bS^l\})&\geq1-|\bS|^l\sqrt{\frac{2}{\pi}}e^{-\eps_l^2(k_l-1)/128}.&
\end{align}
Thus for all $l\geq\hat l$ we have
\begin{align}
\nu_k(\{E\in G_{k_l,k(\eps_l/8,\eps_l/8,k_l)}:|f_{s^l}(\phi)-M_{f_{s^l}}|\leq\eps_l\ \forall \phi\in S(E),\ \forall s^l\in\bS^l\})&>0&
\end{align}
by assumption \emph{A2}, implying the existence of a subspace $E\subset\fr_l$ of dimension $\dim E=k(\eps_l/8,\eps_l/8,k_l)$ such that 
\begin{equation}|f_{s^l}(\phi)-M_{f_{s^l}}|\leq\eps_l\ \forall\ \phi\in S(E),\ s^l\in\bS^l.\label{eqn-theorem-strong-subspace-capacity-11}\end{equation}
Let $\hat\fr_l\subset E$ be any subspace of dimension $k'_l$. Then \emph{P1} holds and equation (\ref{eqn-theorem-strong-subspace-capacity-1}) together with (\ref{eqn-theorem-strong-subspace-capacity-11}) establishes \emph{P2}:
\begin{align}
f_{s^l}(\phi)\geq1-f_l-\frac{4\cdot12}{\sqrt{2(k_l-1)}}-\eps_l\ \ \forall\ \phi\in S(\hat\fr_l),\ \forall s^l\in\bS^l.
\end{align}
\end{proof}
\begin{proof}[Proof of Theorem \ref{theorem:equivalence-of-capacities}] First, assuming that $R>0$ is an achievable rate for entanglement transmission over a \emph{finite} AVQC $\fri$ (with random codes), we show that it is also an achievable strong subspace transmission rate (with random codes) for $\fri$. The proof does not depend on the form of the sequence of probability distributions assigned to the codes, so it applies to the case of deterministically achievable rates as well.\\
So, let there be a sequence of $(l,k_l)$ random entanglement transmission codes with
\begin{align}
\liminf_{l\rightarrow\infty}\frac{1}{l}\log k_l\geq R,\end{align}
\begin{align}
\min_{s^l\in\bS^l}\int F_e(\pi_{\fr_l},\crr^l\circ\cn_{s^l}\circ\cP^l)d\mu_l(\cP^l,\crr^l)=1-f_l,\textrm{ where }f_l\searrow0.
\end{align}
Thus, there is $l'\in\nn$ such that $k_l\geq2^{l3R/4}+1$ for all $l\geq l'$. Choose $\eps_l=2^{-lR/4}$ (this is just one of many possible choices). Obviously, since $R$ is \emph{strictly} greater than zero there is $\hat l\in \nn$ such that 
\begin{eqnarray}
|\bS|^l\sqrt{2/\pi}e^{-\eps_l^2(k_l-1)/128}&\leq&|\bS|^l\sqrt{2/\pi}e^{-\eps_l^22^{l3R/4}/128}\\
&=&|\bS|^l\sqrt{2/\pi}e^{-2^{lR/4}/128}\\
&<&1
\end{eqnarray}
holds for all $l\geq\hat l$. Application of Lemma \ref{theorem-strong-subspace-codes} then yields a sequence of subspaces $\hat\fr_l$ with dimensions $\hat k_l$ such that
\begin{align}
\liminf_{l\rightarrow\infty}\frac{1}{l}\log \hat k_l=\liminf_{l\rightarrow\infty}\frac{1}{l}\log k_l\geq R,
\end{align}
\begin{align}
\min_{s^l\in\bS^l}\min_{\phi\in S(\hat\fr_l)}\int\langle\phi,\crr^l\circ\cn_{s^l}\circ\cP^l(|\phi\rangle\langle\phi|)\phi\rangle d\mu_l(\crr^l,\cP^l)\geq1-f_l-\frac{4\cdot12}{\sqrt{2(k_l-1)}}-\frac{1}{l}\ \ \forall\ l\geq \max\{l',\hat l\}\label{eqn:equivalence-of-capacities-1}.
\end{align}
Since the right hand side of (\ref{eqn:equivalence-of-capacities-1}) goes to zero for $l$ going to infinity, we have shown that $R$ is an achievable rate for strong subspace transmission (with random codes).\\
In case that $|\fri|=\infty$ holds we have to take care of some extra issues that arise from approximating $\fri$ by a finite AVQC. Such an approximation is carried out in detail in the proof of Lemma \ref{random-code-reduction}.\\
Now let $R=0$ be an achievable rate for entanglement transmission with (random) codes. We show that it is achievable for strong subspace transmission by demonstrating that we can \emph{always} achieve a strong subspace transmission rate of zero:\\
Choose any sequence $(|x_l\rangle\langle x_l|)_{l\in\nn}$ of pure states such that $|x_l\rangle\langle x_l|\in\cs(\hr^{\otimes l})\ \forall l\in\nn$. Set $\fr_l:=\mathbb C\cdot x_l$ $(l\in\nn)$. Define a sequence of recovery operations by $\crr^l(a):=\tr(a)\cdot|x_l\rangle\langle x_l|$ $(a\in\B(\kr^{\otimes l}),\ l\in\nn)$. Then $F_e(\pi_{\fr_l},\crr^l\circ\cn_{s^l})=1$ for all $l\in\nn,\ s^l\in\bS^l$ and $\liminf_{l\rightarrow\infty}\frac{1}{l}\log(\dim\fr_l)=0$.
\\\\
Now let $R\geq0$ be an achievable rate for strong subspace transmission over some AVQC $\fri$ (with random codes). Thus, there exists a sequence of (random) strong subspace transmission codes with
\begin{align}
\liminf_{l\rightarrow\infty}\frac{1}{l}\log k_l\geq R,\end{align}
\begin{align}
\inf_{s^l\in\bS^l}\min_{\phi\in \cs(\fr_l)}\int\langle\phi,\crr^l\circ\cn_{s^l}\circ\cP^l(|\phi\rangle\langle\phi|)\phi\rangle d\mu_l(\crr^l,\cP^l)=1-f_l\ \forall l\in\nn,\textrm{ where }f_l\searrow0.\label{eqn:equivalence-of-capacities-2}
\end{align}
Now consider, for every $l\in\nn$ and $s^l\in\bS^l$, the channels $\int\crr^l\circ\cn_{s^l}\circ\cP^ld\mu_l(\crr^l,\cP^l)$. Then (\ref{eqn:equivalence-of-capacities-2}) implies that for these channels we have the estimate
\begin{align}
\inf_{s^l\in\bS^l}\min_{\phi\in \cs(\fr_l)}\langle\phi,\int\crr^l\circ\cn_{s^l}\circ\cP^ld\mu_l(\crr^l,\cP^l)(|\phi\rangle\langle\phi|)\phi\rangle =1-f_l,\label{eqn:equivalence-of-capacities-3}
\end{align}
and by a well-known result (\cite{barnum-knill-nielsen}, Theorem 2) we get
\begin{align}
\inf_{s^l\in\bS^l}F_e(\pi_{\fr_l},\int\crr^l\circ\cn_{s^l}\circ\cP^ld\mu_l(\crr^l,\cP^l))\geq1-\frac{3}{2}f_l,\label{eqn:equivalence-of-capacities-4}
\end{align}
which by convex-linearity of the entanglement fidelity in the channel implies
\begin{align}
\inf_{s^l\in\bS^l}\int F_e(\pi_{\fr_l},\crr^l\circ\cn_{s^l}\circ\cP^l)d\mu_l(\crr^l,\cP^l)\geq1-\frac{3}{2}f_l.\label{eqn:equivalence-of-capacities-5}
\end{align}
But $\lim_{l\rightarrow\infty}\frac{3}{2}f_l=0$ by assumption, implying that $R$ is an achievable rate for entanglement transmission (with random codes) as well.
\end{proof}
\section{\label{sec:converse-avqc}Proof of the converse part}
The basic technical obstacle we are faced with is that the converse part of the coding theorem for an AVQC cannot be reduced immediately to that of the single stationary memoryless quantum channel via Minimax Theorem (cf. \cite{bbt-avc} and \cite{csiszar}). In order to circumvent this problem we derive a relation between $\A_{\textrm{random}}(\fri)$ and the corresponding random capacity of a suitable compound channel.\\
To be explicit, let us consider a finite AVQC $\fri=\{\cn_s  \}_{s\in \bS}$ and let $(\mu_l)_{l\in\nn}$ be a sequence of random $(l,k_l)-$ codes for the AVQC $\fri$ with
\begin{equation}\label{eq:intro-converse-1}
\lim_{l\rightarrow\infty}\inf_{s^l\in\bS^l}\int F_e(\pi_{\fr_l},\crr^l\circ\cn_{s^l}\circ\cP^l)d\mu_l(\cP^l,\crr^l)=1.
\end{equation}
On the other hand, for the infinite channel set $\conv(\fri)$, defined in (\ref{eq:conv-hull}), and each $\cn_q\in \conv(\fri)$ we obtain
\begin{eqnarray}\label{eq:intro-converse-2}
  \int F_e(\pi_{\fr_l},\crr^l\circ\cn_{q}^{\otimes l}\circ\cP^l)d\mu_l(\cP^l,\crr^l)&=& \sum_{s^l\in\bS^l}q(s_1)\cdot \ldots \cdot q(s_l)\int F_e(\pi_{\fr_l},\crr^l\circ\cn_{s^l}\circ\cP^l)d\mu_l(\cP^l,\crr^l)\\
&\ge & \inf_{s^l\in\bS^l}\int F_e(\pi_{\fr_l},\crr^l\circ\cn_{s^l}\circ\cP^l)d\mu_l(\cP^l,\crr^l).
\end{eqnarray}
Consequently, (\ref{eq:intro-converse-1}) and (\ref{eq:intro-converse-2}) imply
\begin{equation}\label{eq:intro-convers-3}
  \lim_{l\to\infty}\inf_{q\in\mathfrak{P}(\bS)}\int F_e(\pi_{\fr_l},\crr^l\circ\cn_{q}^{\otimes l}\circ\cP^l)d\mu_l(\cP^l,\crr^l)=1.
\end{equation}
Defining the random entanglement transmission capacity $Q_{\textrm{comp, random}}(\conv(\fri))$ for the \emph{compound quantum channel} (cf. \cite{bbn-2}) built up from $\conv(\fri)$ in a similar fashion to $\A_{\textrm{random}}(\fri)$ we can infer from the considerations presented above that
\begin{equation}\label{eq:intro-converse-4}
  \A_{\textrm{random}}(\fri)\le Q_{\textrm{comp, random}}(\conv(\fri)).
\end{equation}
Since the inequality $\A_{\textrm{det}}(\fri)\le \A_{\textrm{random}}(\fri)$ is obvious, we obtain the following basic lemma.
\begin{lemma}\label{basic-relations-converse}
Let $\fri=\{\cn_s  \}_{s\in \bS}$ be any finite set of channels and let $\conv(\fri)$ be the associated infinite set given in (\ref{eq:conv-hull}). Then
\begin{equation}\label{eq:basic-relations-converse}
\A_{\textup{det}}(\fri)\le \A_{\textup{random}}(\fri)\le Q_{\textup{comp, random}}(\conv(\fri)).
\end{equation}
\end{lemma}
Thus, our remaining task is to show that right-most capacity in (\ref{eq:basic-relations-converse}) is upper bounded by the last term in (\ref{eq:ahlswede-dichotomy-1}). This is done in the following two subsections for finite and infinite AVQCs respectively.
\subsection{\label{subsec:converse-for-the-finite-avc}Converse for the finite AVQC}
First, we prove the converse to the coding theorem for finite compound quantum channels with random codes.
\begin{theorem}[Converse Part: Compound Channel, $|\fri|<\infty$]\label{theorem:converse-random-compound}
Let $\fri=\{\cn_1,\ldots,\cn_N\}\subset\mathcal C(\hr,\kr)$ be a finite compound channel. The capacity $Q_{\textup{comp, random}}(\fri)$ of $\fri$ is bounded from above by
\begin{equation}Q_{\textup{comp, random}}(\fri)\leq\lim_{l\rightarrow\infty}\max_{\rho\in\mathcal S(\hr^{\otimes l})}\min_{\cn_i\in\fri}\frac{1}{l}I_c(\rho,\cn_i^{\otimes l}).\end{equation}
\end{theorem}
\begin{proof} Let for arbitrary $l\in\mathbb N$ a random $(l,k_l)$ code for a compound channel $\fri=\{\cn_1,\ldots,\cn_N\}$ with the property
\begin{equation}\min_{1\leq i\leq N}\int F_e(\pi_{\fr_l},\crr^l\circ\cn_i^{\otimes l}\circ\mathcal P^l)d\mu_l(\cP^l,\crr^l)\geq 1-\eps_l\end{equation}
be given, where $\eps_l\in[0,1]$ and $\lim_{l\to\infty}\eps_l=0$. Obviously, the above code then satisfies
\begin{eqnarray}
\int F_e(\pi_{\fr_l},\crr^l\circ\frac{1}{N}\sum_{i=1}^N\cn_i^{\otimes l}\circ\mathcal P^l)d\mu_l(\cP^l,\crr^l)&=&\frac{1}{N}\sum_{i=1}^N\int F_e(\pi_{\fr_l},\crr^l\circ\cn_i^{\otimes l}\circ\mathcal P^l)d\mu_l(\cP^l,\crr^l)\\
&\geq&1-\eps_l.\label{eq:converse-uninformed-1}
\end{eqnarray}
This implies the existence of at least one pair $(\crr^l,\cP^l)$ such that
\begin{equation}F_e(\pi_{\fr_l},\crr^l\circ\frac{1}{N}\sum_{i=1}^N\cn_i^{\otimes l}\circ\mathcal P^l)\geq1-\eps_l,\end{equation}
hence for all $i=1,\ldots,N$
\begin{equation}F_e(\pi_{\fr_l},\crr^l\circ\cn_i^{\otimes l}\circ\mathcal P^l)\geq1-N\eps_l.\end{equation} 
The rest of the proof is identical to that of Theorem 9 in \cite{bbn-2}.
\end{proof}
Using the approximation techniques developed in \cite{bbn-2}, we will now prove the converse for random codes and general compound channels.
\begin{theorem}[Converse Part: Compound Channel]\label{theorem:converse-random-compound-general}
Let $\fri\subset\mathcal C(\hr,\kr)$ be an arbitrary compound quantum channel. The capacity $Q_{\textup{comp, random}}(\fri)$ of $\fri$ is bounded from above by
\begin{equation}Q_{\textup{comp, random}}(\fri)\leq\lim_{l\rightarrow\infty}\max_{\rho\in\mathcal S(\hr^{\otimes l})}\inf_{\cn\in\fri}\frac{1}{l}I_c(\rho,\cn^{\otimes l}).\end{equation}
\end{theorem}
For the proof of this theorem, we will need the following lemma:
\begin{lemma}[Cf. \cite{bbn-2}]\label{lemma:estimate-for-coherent-information}
Let $\hr,\kr$ be finite dimensional Hilbert spaces. There is a function $\nu:[0,1]\rightarrow\mathbb R_+$ with $\lim_{x\rightarrow0}\nu(x)=0$ such that for every $\fri,\fri'\subseteq \mathcal{C}(\hr,\kr) $ with $D_\lozenge(\fri,\fri')\leq\tau\leq1/2$ and every $l\in\mathbb N$ we have the estimate
\begin{equation}\left|\frac{1}{l}\inf_{\cn\in\fri}I_c(\rho,\cn^{\otimes l})-\frac{1}{l}\inf_{\cn'\in\fri'}I_c(\rho,\cn'^{\otimes l})\right|\leq\nu(2\tau)\ \ \ \forall \rho\in\mathcal S(\hr^{\otimes l})\end{equation}
The function $\nu$ is given by $\nu(x)=x+8x\log(\dim \kr)+4h(x)$. Here, $h(\cdot)$ denotes the binary entropy.
\end{lemma}
\begin{proof}[Proof of Theorem \ref{theorem:converse-random-compound-general}] Let a sequence $(l,k_l)_{l\in\mathbb N}$ of random codes for $\fri$ be given such that
\begin{itemize}
\item $\liminf_{l\rightarrow\infty}\frac{1}{l}\log\dim\fr_l=R$
\item $\inf_{\cn\in\fri}\int F_e(\fr_l,\crr^l\circ\cn^{\otimes l}\circ\cP^l)d\mu_l(\cP^l,\crr^l)=1-\eps_l$,
\end{itemize}
where the sequence $(\eps_l)_{l\in\mathbb N}$ satisfies $\lim_{l\to\infty}\eps_l=0$. Let, for some $\tau>0$, $\bigcup_{\cn\in\fri}B_\lozenge(\cn,\tau)$ be an open cover for $\fri$.
Clearly, it also covers the compact set $\bar\fri$. Thus, there exist finitely many channels $\cn_1,\ldots,\cn_{M_\tau}$ such that $\bigcup_{i=1}^{M_\tau}B_\lozenge(\cn_i,\tau)\supset\bar\fri$
 and, therefore, $\mathcal M_\tau:=\{\cn_1,\ldots,\cn_{M_\tau}\}$ is a $\tau$-net for $\fri$.\\
By $\mathcal M_\tau\subset\fri$ we get, for every $\tau>0$, the following result:
\begin{itemize}
\item $\liminf_{l\rightarrow\infty}\frac{1}{l}\log\dim\fr_l=R$
\item $\min_{\cn_i\in\mathcal M_\tau}\int F_e(\fr_l,\crr^l\circ\cn_i^{\otimes l}\circ\cP^l)d\mu_l(\cP^l,\crr^l)\geq1-\eps_l$.
\end{itemize}
By Theorem \ref{theorem:converse-random-compound}, this immediately implies
\begin{equation}R\leq\lim_{l\rightarrow\infty}\frac{1}{l}\max_{\rho\in\cs(\hr^{\otimes l})}\min_{\cn_i\in\mathcal M_\tau}I_c(\rho,\cn_i^{\otimes l}).\end{equation}
From Lemma \ref{lemma:estimate-for-coherent-information} we get, by noting that $D_\lozenge(\fri,\mathcal M_\tau)\leq\tau$ the estimate
\begin{equation}R\leq\lim_{l\rightarrow\infty}\frac{1}{l}\max_{\rho\in\cs(\hr^{\otimes l})}\inf_{\cn\in\fri}I_c(\rho,\cn^{\otimes l})+\nu(2\tau).\end{equation}
Taking the limit $\tau\rightarrow0$ proves the theorem.
\end{proof}
\begin{theorem}[Converse: finite AVQC]\label{converse:finite-avc} Let $\fri=\{\cn_s\}_{s\in\bS}$ be a finite AVQC. Then
\begin{equation}\A_{\textup{random}}(\fri)\leq Q_{\textup{comp, random}}(\conv(\fri))\leq\lim_{l\rightarrow\infty}\frac{1}{l}\max_{\rho\in\cs(\hr^{\otimes l})}\inf_{\cn\in \conv(\fri)}I_c(\rho,\cn^{\otimes l}).\end{equation}
\end{theorem}
\emph{Proof.} Just combine Lemma \ref{basic-relations-converse} and Theorem \ref{theorem:converse-random-compound-general} applied to $\conv(\fri)$.
\begin{flushright}$\Box$\end{flushright}
\subsection{\label{subsec:general-converse}Case $|\fri|=\infty$}
The proof of the converse part of Theorem \ref{quant-ahlswede-dichotomy} requires just a bit of additional work. Let $\fri=\{ \cn_s \}_{s\in \bS}$ be an arbitrary AVQC and let $\mathfrak{P}_{\textrm{fin}}(S)$ denote the set of probability distributions on $S$ with \emph{finite} support. Then
\begin{equation}\conv(\fri)=\left\{\cn_q\in\mathcal{C}(\hr,\kr): \cn_q:=\sum_{s\in\bS}q(s)\cn_s, \textrm{ and }q\in\mathfrak{P}_{\textrm{fin}}(S)   \right\}.  \end{equation}
The argument that led us to the inequality (\ref{eq:intro-converse-2}) accompanied by the continuity of the entanglement fidelity with respect to $||\cdot ||_{\lozenge}$ and 
an application of the dominated convergence theorem show that for each $\cn\in \overline{\mathrm{conv}(\fri)}$ 
\begin{equation}\int F_e(\pi_{\fr_l},\crr^l\circ\cn^{\otimes l}\circ\cP^l)d\mu_l(\cP^l,\crr^l)\ge \inf_{s^l\in\bS^l}\int F_e(\pi_{\fr_l},\crr^l\circ\cn_{s^l}\circ\cP^l)d\mu_l(\cP^l,\crr^l)  \end{equation}
holds. Then Lemma \ref{basic-relations-converse} holds \emph{mutatis mutandis} with $\mathrm{conv}(\fri)$ replaced by $\overline{\mathrm{conv}(\fri)}$. 
Additionally, if we apply Theorem \ref{theorem:converse-random-compound-general} to $\overline{\mathrm{conv}(\fri)}$ we are led to the following theorem.
\begin{theorem}[Converse: general AVC]\label{converse:general-avc} Let $\fri=\{\cn_s\}_{s\in\bS}$ be an arbitrary AVQC. Then
\begin{equation}\A_{\textup{random}}(\fri)\leq\lim_{l\rightarrow\infty}\frac{1}{l}\max_{\rho\in\cs(\hr^{\otimes l})}\min_{\cn\in\overline{\mathrm{conv}(\fri)}}I_c(\rho,\cn^{\otimes l}).\end{equation}
\end{theorem}

\section{\label{sec:random-achievability}Achievability of entanglement transmission rate I: Random codes}
We show in this section how the achievability results for compound quantum channels from our previous paper \cite{bbn-2} imply existence of reliable random codes for AVQC via Ahlswede's robustification technique \cite{ahlswede-coloring}.

Let $l\in\nn$ and let $\textrm{Perm}_l$ denote the set of permutations acting on $\{1,\ldots, l\}$. Let us further suppose that we are given a finite set $\bS$. Then each permutation $\sigma\in \textrm{Perm}_l$ induces a natural action on $\bS^l$ by  $\sigma:\mathbf S^l\rightarrow\mathbf S^l$, $\sigma(s^l)_i:=s_{\sigma(i)}$. Moreover, let $T(l,\bS)$ denote the set of types on $\bS$ induced by the elements of $\bS^l$, i.e. the set of empirical distributions on $\bS$ generated by sequences in $\bS^l$. Then Ahlswede's robustification can be stated as follows.
\begin{theorem}[Robustification technique, cf. Theorem 6 in \cite{ahlswede-coloring}]\label{robustification-technique}
If a function $f:\bS^l\to [0,1]$ satisfies
\begin{equation}\label{eq:robustification-1}
 \sum_{s^l\in\bS^l}f(s^l)q(s_1)\cdot\ldots\cdot q(s_l)\ge 1-\gamma
\end{equation}
for all $q\in T(l,\bS)$ and some $\gamma\in [0,1]$, then
\begin{equation}\label{eq:robustification-2}
  \frac{1}{l!}\sum_{\sigma\in\textup{Perm}_l}f(\sigma(s^l))\ge 1-(l+1)^{|\bS  |}\cdot \gamma\qquad \forall s^l\in \bS^l.
\end{equation}
\end{theorem}
\begin{remark} Ahlswede's original approach in \cite{ahlswede-coloring} gives
\begin{equation}\frac{1}{l!}\sum_{\sigma\in\textup{Perm}_l}f(\sigma(s^l))\ge 1-3\cdot (l+1)^{|\bS  |}\cdot \sqrt{\gamma}\qquad \forall s^l\in \bS^l.  \end{equation}
The better bound (\ref{eq:robustification-2}) is from \cite{ahlswede-gelfand-pinsker}.
\end{remark}
\begin{proof} Because the result of Theorem \ref{robustification-technique} is a central tool in our paper and the proof given in \cite{ahlswede-gelfand-pinsker} is particularly simple we reproduce it here in full for reader's convenience.\\
Notice first that (\ref{eq:robustification-1}) is equivalent to
\begin{equation}\sum_{s^l\in\bS^l}(1-f(s^l))q(s_1)\cdot\ldots\cdot q(s_l)\le \gamma \qquad \forall q\in T(l,\bS),  \end{equation}
which in turn is equivalent to
\begin{equation}\sum_{s^l\in\bS^l}(1-f(\sigma(s^l)))q(s_{\sigma(1)})\cdot\ldots\cdot q(s_{\sigma(l)})\le \gamma \qquad \forall q\in T(l,\bS),   \end{equation}
and $\sigma\in \textrm{Perm}_l$, since $\sigma$ is bijective. Clearly, we have
\begin{equation}q(s_{\sigma(1)})\cdot\ldots\cdot q(s_{\sigma(l)})= q(s_1)\cdot\ldots\cdot q(s_l)\qquad \forall \sigma \in \textrm{Perm}_l,\forall s^l\in\bS^l, \end{equation}
and therefore, we obtain
\begin{equation}\label{eq:rob-1}
  \sum_{s^l\in\bS^l}\left( 1-\frac{1}{l!}\sum_{\sigma\in\textrm{Perm}_l}f(\sigma(s^l))  \right)q(s_1)\cdot\ldots\cdot q(s_l)\le \gamma  \qquad \forall q\in T(l,\bS).
\end{equation}
Now, for $q\in T(l,\bS)  $ let $T_q^l\subset \bS^l$ denote the set of sequences whose empirical distribution is $q$. Since $f$ takes values in $[0,1]$ we have $1-\frac{1}{l!}\sum_{\sigma\in\textrm{Perm}_l}f(\sigma(s^l))\ge 0 $ and thus from (\ref{eq:rob-1})
\begin{equation}\label{eq:rob-2}
 \sum_{s^l\in T_q^l}\left( 1-\frac{1}{l!}\sum_{\sigma\in\textrm{Perm}_l}f(\sigma(s^l))  \right)q(s_1)\cdot\ldots\cdot q(s_l)\le \gamma  \qquad \forall q\in T(l,\bS).
\end{equation}
It is clear from definition that for each $s^l\in T_q^l$ we have $\bigcup_{\sigma\in\textrm{Perm}_l}\{ \sigma(s^l) \}=T_q^l$ and, consequently, $\sum_{\sigma\in\textrm{Perm}_l}f(\sigma(s^l))$ does not depend on $s^l\in T_q^l$. Therefore, from (\ref{eq:rob-2}) we obtain
\begin{equation}\label{eq:rob-3}
 \left( 1-\frac{1}{l!}\sum_{\sigma\in\textrm{Perm}_l}f(\sigma(s^l))  \right)q^{\otimes l}(T_q^l)\le \gamma\qquad \forall  q\in T(l,\bS),\  \forall s^l\in T_q^l.
\end{equation}
On the other hand
\begin{equation}\label{eq:rob-4}
  q^{\otimes l}(T_q^l)\ge \frac{1}{(l+1)^{|\bS|}}\qquad \forall q\in T(l,\bS)
\end{equation}
holds (cf. \cite{csiszar} page 30), which, by (\ref{eq:rob-3}), implies
\begin{equation} \left( 1-\frac{1}{l!}\sum_{\sigma\in\textrm{Perm}_l}f(\sigma(s^l))  \right)\le (l+1)^{|\bS|}\cdot  \gamma\qquad \forall  q\in T(l,\bS),\  \forall s^l\in T_q^l. \end{equation}
This is the inequality we aimed to prove since $\bS^l=\bigcup_{q\in T(l,\bS) }T_q^l$.
\end{proof}
The function $f$ appearing in Theorem \ref{robustification-technique} will be built up from the entanglement fidelities of the channels constituting a finite AVQC that approximates our AVQC $\fri=\{\cn_s\}_{s\in\bS}$.\\
As another ingredient for the arguments to follow we need an achievability result for compound channels. 
\begin{lemma}\label{compound-achiev-input}
Let $k\in \nn$ and $\mathfrak T\subset \mathcal C(\hr,\kr)$. For each $\eta>0$ there is a sequence of $(l,k_l)$-codes $(\mathcal{P}^l, \crr^l)_{l\in\nn}$ and an 
$l_0(\eta)\in\nn$ 
such that for all $l\ge l_0(\eta)$
\begin{equation}\label{eq:comp-input-1}
  F_{e}(\pi_{\fr_l}, \crr^l\circ \cn^{\otimes l}\circ \mathcal{P}^l)\ge 1- 2^{-lc}\qquad \forall \cn\in \mathfrak{T},
\end{equation}
and
\begin{equation}
  \frac{1}{l}\log \dim \fr_l\ge \frac{1}{k}\max_{\rho\in\cs(\hr^{\otimes k})}\inf_{\cn\in\mathfrak T}I_c(\rho,\cn^{\otimes k})  -\eta,
\end{equation}
hold with a real constant $c=c(k,\dim \hr,\dim\kr,\mathfrak T,\eta)>0$.
\end{lemma}
\begin{remark}
Lemma \ref{compound-achiev-input} is a strengthening of Theorem 7 in \cite{bbn-2} insofar as it explicitly points out the following: 
For any compound channel, at any rate below its capacity for transmission of entanglement, there exist sequences of codes such that 
entanglement fidelity goes to one exponentially fast.\\
The importance of this result for the investigations at hand can be understood by looking at equation (\ref{conversion-2}) in Theorem \ref{conversion-of-compound-codes}.
\end{remark}
\begin{proof} We will give the details for the case $k=1$ only. The proof for arbitrary $k$ follows by an almost identical argument.\\
According to the compound BSST Lemma (cf. \cite{bbn-1}, Lemma 6.1) to any $\eta>0$ we can find $m=m(\mathfrak T,\eta)\in\nn $ and a subspace $\gr\subset \hr^{\otimes m}$ such that
\begin{equation}\label{eq:compound-input-1}
\frac{1}{m}\inf_{\cn\in\mathfrak T}I_c(\pi_{\gr},\cn^{\otimes m})\ge \max_{\rho\in\cs (\hr)}\inf_{\cn\in\mathfrak T}I_c(\rho,\cn)-\frac{\eta}{3}.
\end{equation}
Explicitly stated, $\gr$ is the eigenspace to eigenvalue one of a frequency typical projection of the maximizer on the r.h.s. of (\ref{eq:compound-input-1}) with suitably 
chosen parameters.\\
Let us consider the compound quantum channel built up from $\{\cn^{\otimes m}:\cn\in\mathfrak T\}$.\\
Looking at the last equation in the proof of Lemma 9 in \cite{bbn-2} we see, that for this channel, for large enough $t\in\nn$, there exists
a sequence of $(t,k_t)$-codes $(\bar{\mathcal{P}}^t,\bar{\crr}^t)_{t\in\nn}$, $\bar{\mathcal{P}}^t\in\mathcal{C}(\fr_t,\hr^{\otimes mt})$, 
$\bar{\crr}^t\in \mathcal{C}(\kr^{\otimes mt},\fr'_t)$ with
\begin{equation}\label{eqn-revision-1}
 \inf_{\cn\in\mathfrak T}F_e(\pi_{\fr_t},\bar\crr^t\circ\cn^{\otimes mt}\circ\bar\cP^t)\geq1-\sqrt{3N_{\tau_t}\eps_t}-t\tau_t
\end{equation}
holds. Here, $\eps_t\leq18N_{\tau_t}2^{-tc_1}$ for some $c_1>0$ (compare equation (34) in \cite{bbn-2}), while $N_{\tau_t}\leq(6/\tau_t)^{2(d\cdot d')^2}$, 
where $d=\dim\hr$ and $d'=\dim\kr$. Inserting these estimates into (\ref{eqn-revision-1}), we know that there exists $t_0\in\nn$ such that for all $t\geq t_0$ we have
\begin{equation}\label{eqn-revision-2}
\inf_{\cn\in\mathfrak T}F_e(\pi_{\fr_t},\bar\crr^t\circ\cn^{\otimes mt}\circ\bar\cP^t)\geq1-18N_{\tau_t}2^{-lc_1/2}-t\tau_t.
\end{equation}
Setting $t_1:=\max\{t_0,18\cdot6^{2(dd')^2}\}$ we even get that for all $t\geq t_1$ the estimate 
\begin{equation}\label{eqn-revision-3}
\inf_{\cn\in\mathfrak T}F_e(\pi_{\fr_t},\bar\crr^t\circ\cn^{\otimes mt}\circ\bar\cP^t)\geq1-t(2^{-lc_1/2}(1/\tau_t)^{2(dd')^2}+\tau_t)
\end{equation}
holds.
Although the proof of Lemma 9 uses a subexponential growth of $N_{\tau_t}$, this is not at all necessary. Set
\begin{equation}c_2:=c_1/8(d\cdot d')^2,\qquad\qquad\tau_t:=2^{-tc_2}\ \ (t\in\nn).\end{equation}
Then
\begin{equation}
 \inf_{\cn\in\mathfrak T}F_e(\pi_{\fr_t},\bar\crr^t\circ\cn^{\otimes mt}\circ\bar\cP^t)\geq1-t\cdot(2^{-tc_1/4}+2^{-tc_1/8(dd')^2})
\end{equation}
and thus, defining $c':=c_1/(3\cdot d\cdot d')^2$, we know that for each $\eta>0$ there exists $t(\eta)\in\nn$ such that
 \begin{equation}\label{eq:compound-input-2}
  \inf_{\cn\in\mathfrak T}F_e(\pi_{\fr_t},\bar{\crr}^t\circ \cn^{\otimes mt}\circ \bar{\mathcal{P}}^t)\ge 1-2^{-tc'}\qquad \forall t\geq t(\eta),
\end{equation}
as well as
\begin{equation}\label{eq:compound-input-3}
 \frac{1}{t}\log k_t= \frac{1}{t}\log \dim \fr_t\ge \inf_{\cn\in\mathfrak T}I_c(\pi_{\gr},\cn^{\otimes m})-\frac{\eta}{3}\qquad \forall t\geq t(\eta)
\end{equation}
where, clearly, $c'=c'(\dim\hr,\dim \kr,\mathfrak T,\eta)$.\\
For $t,l\in\nn$ let $r\in\{0,1,\ldots, m-1  \}$ be the unique non-negative integer such that $l=mt+r$. Furthermore, let us choose for each $r\in\{0,1,\ldots, m-1  \} $ a state vector $x_r\in\hr^{\otimes r}$ and set
\begin{equation}
\fr_l:=\fr_t\otimes \cc\cdot \{ x_r \}.\label{eq:compound-input-7}
\end{equation}
Then
\begin{equation}\label{eq:compound-input-5a}
\pi_{\fr_l}=\pi_{\fr_t}\otimes |x_r\rangle\langle x_r|.
\end{equation}
Moreover we set
\begin{equation}\label{eq:compound-input-6}
\mathcal{P}^l:=\bar{\mathcal{P}}^t\otimes id_{\mathcal{B}(\hr^{\otimes r})}\quad \textrm{and}\quad \crr^l:=\bar{\crr}^t\otimes T^r ,
\end{equation}
where $T^r\in\mathcal{C}(\kr^{\otimes r},\hr^{\otimes r})$ is given by $T^r(a):=\tr(a)|x_r\rangle\langle x_r|$. Then it is clear that
\begin{eqnarray}\label{eq:compound-input-4}
  F_e(\pi_{\fr_l},\crr^l\circ \cn^{\otimes l}\circ \mathcal{P}^l)&=& F_e(\pi_{\fr_t}, \bar{\crr}^t\circ\cn^{\otimes mt}\circ \bar{\mathcal{P}}^t)\\
&\ge& 1-2^{-tc'}\\
&=& 1-2^{-\frac{l-r}{m}c'}\\
&\ge & 1-2^{-l c}\qquad \forall \cn\in\mathfrak T
\end{eqnarray}
for all $l\ge l_1(\eta)$ with $c:=\frac{c'}{2m}$, and where in the second line we have used (\ref{eq:compound-input-2}).\\
On the other hand, from equations (\ref{eq:compound-input-1}), (\ref{eq:compound-input-3}) and (\ref{eq:compound-input-7}) we obtain for $t\ge t(\eta)$
\begin{eqnarray}\label{eq:compound-input-5}
  \frac{1}{l}\log\dim\fr_l&=& \frac{1}{tm+r}\log\dim \fr_t\\
&\ge & \frac{1}{1+\frac{r}{tm}}(\max_{\rho\in\cs (\hr)}\inf_{\cn\in\mathfrak T}I_c(\rho,\cn)-\frac{\eta}{3}-\frac{\eta}{3m}) \\
&\ge& \max_{\rho\in\cs (\hr)}\inf_{\cn\in \mathfrak T}I_c(\rho,\cn)-\eta
\end{eqnarray}
if $t$ and consequently $l$ is sufficiently large. Therefore there is an $l_0(\eta)\in\nn$ such that (\ref{eq:compound-input-4}) and (\ref{eq:compound-input-5}) hold simultaneously for all $l\ge l_0(\eta)$ which concludes the proof in the case $k=1$.
\end{proof}
In the next step we will combine the robustification technique and Lemma \ref{compound-achiev-input} to prove the existence of good random codes for the AVQC $\fri=\{\cn_s  \}_{s\in\bS}$.\\
Recall that there is a canonical action of $\textrm{Perm}_l$ on $\mathcal{B}(\hr)^{\otimes l}$ given by $A_{\sigma,\hr}(a_1\otimes\ldots\otimes a_l):=a_{\sigma^{-1}(1)}\otimes\ldots\otimes a_{\sigma^{-1}(l)}$. It is easy to see that $A_{\sigma,\hr}(a)=U_{\sigma}aU_{\sigma}^{\ast},\ (a\in\B(\hr)^{\otimes l})$ with the unitary operator $U_{\sigma}:\hr^{\otimes l}\to\hr^{\otimes l}$ defined by $U_{\sigma}(x_1\otimes \ldots\otimes x_l)=x_{\sigma^{-1}(1)}\otimes \ldots\otimes x_{\sigma^{-1}(l)}$.\\
\begin{theorem}[Conversion of compound codes]\label{conversion-of-compound-codes}
Let $\fri=\{\cn_s  \}_{s\in\bS}$ be an AVQC. For each $k\in\nn$ and any sufficiently small $\eta>0$ there is a sequence of codes $(\mathcal{P}^l,\crr^l)_{l\in\mathbb N}$, $\mathcal{P}^l\in\mathcal{C}(\fr_l,\hr^{\otimes l}),\crr^l\in\mathcal{C}(\kr^{\otimes l},\fr'_l)$, for the compound channel built up from $\conv(\fri)$ (cf. (\ref{eq:conv-hull})) satisfying
\begin{equation}\label{conversion-1}
\frac{1}{l}\log \dim \fr_l\ge \frac{1}{k}\max_{\rho\in\cs(\hr^{\otimes k})}\inf_{\cn\in \conv(\fri)}I_c(\rho,\cn^{\otimes k})  -2\cdot h(8\eta),
\end{equation}
\begin{equation}\label{conversion-2}
  \frac{1}{l!}\sum_{\sigma\in\textup{Perm}_l}F_e(\pi_{\fr_l}, \crr^l\circ A_{\sigma^{-1},\kr}\circ \cn_{s^l}\circ A_{\sigma,\hr}\circ \mathcal{P}^l)\ge 1- (l+1)^{N_\eta}\cdot 2^{-lc}\qquad \forall s^l\in\bS^l
\end{equation}
for all sufficiently large $l$ with a positive number $c=c(k,\dim\hr,\dim\kr,\conv(\fri),\eta)$, $\nu:[0,1]\rightarrow\mathbb R$ defined by $\nu(x):=x+8x\log(d_\kr)+4\cdot h(x)$ ($h(\cdot)$ being the binary entropy) and an integer $N_\eta$ which depends on the set $\fri$ as well.
\end{theorem}
\begin{remark}
Notice that (\ref{conversion-2}) guarantees us the existence of a sequence of random codes with recovery operations $\crr^l\circ A_{\sigma^{-1},\kr}$, encoding operations 
$A_{\sigma,\hr}\circ\cP^l$, and measure being the uniform distribution on $\textrm{Perm}_l$.
\end{remark}

The idea of the proof is the following. We want to approximate the set $\conv(\fri)$ from the outside by using a polytope $P_\eta$ with $N_\eta$ extreme points.
Then, our results for compound codes and an application of the robustification technique yield a sequence of codes which have asymptotically optimal performance
for the AVQC $P_\eta$. Since $\conv(\fri)\subset P_\eta$, they will also have asymptotically optimal performance for $\fri$.\\
A problem occurs if $\conv(\fri)$ touches the boundary of the set of quantum channels because parts of that boundary are curved and the approximating polytope $P_{\eta}$ may contain maps that are not channels.
Therefore, an intermediate step consists of slightly moving $\conv(\fri)$ away from the boundary. This may be seen as application of a completely positive map and can therefore be absorbed into the recovery operation.\\
During the proof we are going to make use of the following Lemma, that will be proven first:
\begin{lemma}\label{lemma-direct-infinite-2}
Let $A,B$ be compact convex sets in $\mathbb C^n$ with $A\subset B$ and
\begin{equation}
d(\rebd  B, A):=\inf\{||b-a||:b\in\rebd  B,a\in A\}=t>0,
\end{equation}
where $||\cdot||$ denotes any norm.\\
Let $P\supset A$ be a polytope with $D(A,P)\leq\delta$, where $\delta\in(0,t]$ and $D$ is the Hausdorff distance induced by $||\cdot||$. Then $P':=P\cap \aff  B$ is a polytope and $P' \subset B$.
\end{lemma}
\begin{proof}[Proof of Lemma \ref{lemma-direct-infinite-2}] The assertion that $P'$ is a polytope is clear. Suppose $\exists p\in P'\backslash B$. Then since $D(A,P)\leq\delta$ we have $P\subset (A)_\delta$ (cf. \cite{webster}, Theorem 2.7.3).
But this means, since $P'\subset P$, that to our $p\in P'\backslash B$ we can find $a_\delta\in A$ with
\begin{equation}
||p-a_\delta||\leq\delta.\label{eqn-direct-infinite-4}
\end{equation}
For $\lambda\in[0,1]$ define
\begin{equation}
x_\lambda:=(1-\lambda)a_\delta+\lambda p.\label{eqn-direct-infinite-5}
\end{equation}
Then there is $\lambda^*\in(0,1)$ such that
\begin{equation}
x:=x_\lambda^*\ \ \in\ \ \rebd B.
\end{equation}
This is seen as follows: Since $d(\rebd B,A)=t>0$ we have $A\subset \ri B$. Set
\begin{equation}
L:=\{\lambda\in(0,1]:(1-\lambda)a_\delta+\lambda p\in B\}.
\end{equation}
From $a_\delta\in \ri  B$ it follows that $L\neq\emptyset$ and from the fact that $B$ is compact and convex we then get that $L=(0,\lambda^*]$. Now,
\begin{eqnarray}
||x-a_\delta||&=&||(1-\lambda^*)a_\delta+\lambda^*p-(1-\lambda^*)a_\delta-\lambda^*a_\delta||\\
&=&\lambda^*||p-a_\delta||\\
&\leq&\lambda^*\cdot\delta\\
&<&t,
\end{eqnarray}
where the last line follows from $\lambda\in(0,1)$. This is a contradiction to $d(\rebd  B,A)=t$.
\end{proof}
\begin{proof}[Proof of Theorem \ref{conversion-of-compound-codes}] We can suppose that
\begin{equation}\max_{\rho\in\cs(\hr^{\otimes k})}\inf_{\cn\in \conv(\fri)}I_c(\rho,\cn^{\otimes k})>0,\label{eqn:direct-infinite-20}\end{equation}
because otherwise our claim is obviously true. We will further assume that $\fri$, and therefore $\conv(\fri)$ as well, is compact. Since the Hausdorff-distance of $\conv(\fri)$ to its closure (in $||\cdot||_\lozenge$) is zero, this does not change the left hand side of equation (\ref{eqn:direct-infinite-20}), due to the estimates in Lemma \ref{lemma:estimate-for-coherent-information}. Since $\fri$ is a subset of its norm-closure, good codes for the norm-closure will also work for $\fri$. Thus, our assumption is a pure technicality and, indeed, without loss of generality.\\
Now let us, for $\eps\leq1$, by $\mathfrak D_\eps$ denote the operation $\mathfrak D_\eps(\cdot):=(1-\eps)\textup{id}_{\mathcal{B}(\kr)}(\cdot)+\frac{\eps}{\dim\kr}\mathbf1_\kr\tr(\cdot)$. If $\eps\geq0$, this is nothing but a depolarizing channel.\\
By Lemma 2.3.3 in \cite{webster} and since $\mathfrak D_1\circ\cn\notin \rebd \mathcal C(\hr,\kr)$ for arbitrary $\cn\in\mathcal C(\hr,\kr)$ and $\eta>0$,
we have 
\begin{equation}
\mathfrak D_\eta(\conv(\fri))\subset \ri \mathcal C(\hr,\kr).
\end{equation}
Since $\mathfrak D_\eta(\conv(\fri))$ is compact, we know that
\begin{equation}
c':=\min\{||\cn-\cn'||_\lozenge:\cn\in\mathfrak D_\eta(\conv(\fri)), \cn'\in \rebd \mathcal C(\hr,\kr)\}
\end{equation}
satisfies $c'>0$. Thus, by Lemma \ref{lemma-direct-infinite-2} and Theorem 3.1.6 in \cite{webster} there exists a polytope $P_\eta\subset\mathcal C(\hr,\kr)$ such that $\mathfrak D_\eta(\conv(\fri))\subset P_\eta$ and 
\begin{equation}
 D_\lozenge(\mathfrak D_\eta(\conv(\fri)),P_\eta)\leq2\eta.
\end{equation}
The set of extremal points of $P_\eta$ we denote by $\textup{ext}(P_\eta)=\{\cn_e\}_{e\in E_\eta}$, where $E_\eta$ is a finite set indexing the extremal points, the number of which we label $N_\eta$. 
Consider the compound quantum channel (the papers \cite{bbn-1,bbn-2} give proper definitions of this object)  $P_\eta$. It follows from Lemma \ref{compound-achiev-input} that there exists a sequence of $(l,k_l)$-codes $(\mathcal{P}^l,\crr^l)_{l\in\nn}$ such that for all $l\ge l_0(\eta)$
\begin{equation}\label{eq:conversion-1}
 F_{e}(\pi_{\fr_l}, \crr^l\circ \cn^{\otimes l}\circ \mathcal{P}^l)\ge 1- 2^{-lc}\qquad \forall \cn\in P_\eta,
\end{equation}
and
\begin{equation}\label{eq:conversion-2}
  \frac{1}{l}\log \dim \fr_l\ge \frac{1}{k}\max_{\rho\in\cs(\hr^{\otimes k})}\inf_{\cn\in P_\eta}I_c(\rho,\cn^{\otimes k})  -\eta,
\end{equation}
with a positive number $c=c(k,\dim\hr,\dim\kr,\conv(\fri),\eta)$.\\
Let us define $f:E_\eta^l\to [0,1]$ by
\begin{equation} f(e^l):= F_{e}(\pi_{\fr_l}, \crr^l\circ \cn_{e^l}\circ \mathcal{P}^l). \end{equation}
Then (\ref{eq:conversion-1}) implies that
\begin{equation}\label{eq:conversion-3}
 \sum_{e^l\in E_\eta^l}f(e^l)q(e_1)\cdot\ldots\cdot q(e_l)\ge 1- 2^{-lc} \qquad \forall q\in T(l,E_\eta).
\end{equation}
But (\ref{eq:conversion-3}) and Theorem \ref{robustification-technique} yield
\begin{equation}\label{eq:conversion-4}
\frac{1}{l!}\sum_{\sigma\in\textrm{Perm}_l} F_{e}(\pi_{\fr_l}, \crr^l\circ \cn_{\sigma(e^l)}\circ \mathcal{P}^l)\ge 1-(l+1)^{N_\eta}\cdot 2^{-lc}\qquad \forall e^l\in E_\eta^l.
\end{equation}
By (\ref{eq:conversion-2}) and (\ref{eq:conversion-4}) we are guaranteed the existence of a good random code for $P_\eta$ if we can somehow consider permutations as part of the encoding and recovery procedure.
More precisely, we will now show that
\begin{equation}\label{eq:conversion-5}
  \cn_{\sigma (e^l)}= A_{\sigma^{-1},\kr}\circ \cn_{e^l}\circ A_{\sigma,\hr}\qquad \forall e^l\in E_\eta^l.
\end{equation}
 To this end, let $\psi=\psi_1\otimes\ldots\otimes\psi_l,\ \varphi=\varphi_1\otimes\ldots\otimes \varphi_l\in\hr^{\otimes l}$. Then
\begin{eqnarray}
A_{\sigma^{-1},\kr}\circ\cn_{e^l}\circ A_{\sigma,\hr}(|\psi\rangle\langle\varphi|)&=&(A_{\sigma^{-1},\kr}\circ\cn_{e^l})(|\psi_{\sigma^{-1}(1)}\rangle\langle\varphi_{\sigma^{-1}(1)}|\otimes\ldots\otimes|\psi_{\sigma^{-1}(l)}\rangle\langle\varphi_{\sigma^{-1}(l)}|)\\
&=&A_{\sigma^{-1},\kr}(\otimes_{i=1}^{l} \cn_{s_i}(|\psi_{\sigma^{-1}(i)}\rangle\langle\varphi_{\sigma^{-1}(i)}|))\\
&=&\otimes_{i=1}^l\cn_{s_{\sigma(i)}}(|\psi_{i}\rangle\langle\varphi_{i}|)\\
&=&\cn_{\sigma(e^l)}(\otimes_{i=1}^{l}|\psi_{i}\rangle\langle\varphi_{i}|)\\
&=& \cn_{\sigma(e^l)}(|\psi\rangle\langle\varphi|).
\end{eqnarray}
Therefore,
\begin{equation} A_{\sigma^{-1},\kr}\circ\cn_{e^l}\circ A_{\sigma,\hr}=\cn_{\sigma(e^l)}.\end{equation}
By construction of $P_\eta$ we know that for every $\cn_s\in\fri$ there exists a probability distribution $q(\cdot|s)\in\mathfrak P(E_\eta)$ such that
\begin{equation}
 \mathfrak D_\eta\circ\cn_s=\sum_{e\in E_\eta}q(e|s)\cn_e\label{eq:conversion-6}
\end{equation}
holds. We define 
\begin{equation}
\tilde\crr^l_\sigma:=\crr^l\circ A_{\sigma^{-1},\kr}\circ\mathfrak D_\eta^{\otimes l},\qquad \tilde\cP^l_\sigma:=A_{\sigma,\hr}\circ\cP^l.\label{eq:conversion-7}
\end{equation}
Combining the equations (\ref{eq:conversion-4}),(\ref{eq:conversion-5},),(\ref{eq:conversion-6}),(\ref{eq:conversion-7}) we get for every $s^l\in\bS^l$:
\begin{eqnarray}
\sum_{\sigma\in\textrm{Perm}_l}F_e(\pi_{\fr_l},\tilde\crr^l_\sigma\circ\cn_{s^l}\circ\tilde\cP^l_\sigma)&=&\sum_{\sigma\in\textrm{Perm}_l}F_e(\pi_{\fr_l},\crr^l\circ A_{\sigma^{-1},\kr}\circ\mathfrak D_\eta^{\otimes l}\circ\cn_{s^l}\circ A_{\sigma,\hr}\circ\cP^l)\\
&=&\sum_{\sigma\in\textrm{Perm}_l}F_e(\pi_{\fr_l},\crr^l\circ A_{\sigma^{-1},\kr}\circ\sum_{e^l\in E_\eta^l}\prod_{i=1}^lq(e_i|s_i)\otimes_{j=1}^l\cn_{e_j}\circ A_{\sigma,\hr}\circ\cP^l)\\
&=&\sum_{e^l\in E_\eta^l}\prod_{i=1}^lq(e_i|s_i)\sum_{\sigma\in\textrm{Perm}_l}F_e(\pi_{\fr_l},\crr^l\circ A_{\sigma^{-1},\kr}\circ\cn_{e^l}\circ A_{\sigma,\hr}\circ\cP^l)\\
&=&\sum_{e^l\in E_\eta^l}\prod_{i=1}^lq(e_i|s_i)\sum_{\sigma\in\textrm{Perm}_l}F_e(\pi_{\fr_l},\crr^l\circ\cn_{\sigma(e^l)}\circ\cP^l)\\
&\geq&(l!)(1-(l+1)^{N_\eta}\cdot 2^{-lc})
\end{eqnarray}
Now, defining a discretely supported probability measure $\mu_l$, $l\in\nn$ by
\begin{equation}\mu_l:=\frac{1}{l!}\sum_{\sigma\in\textrm{Perm}_l}\delta_{(\tilde\crr^l_\sigma,\tilde\cP^l_\sigma)}, \end{equation}
where $\delta_{(\tilde\crr^l_\sigma,\tilde\cP^l_\sigma)}$ denotes the probability measure that puts measure $1$ on the point $(\tilde\crr^l_\sigma, \cP^l_\sigma)$,
we obtain for each $k\in\nn$ a sequence of $(l,k_l)$-random codes for $\fri$ achieving
\begin{equation} \frac{1}{k}\max_{\rho\in\cs(\hr^{\otimes k})}\inf_{\cn\in P_\eta}I_c(\rho,\cn^{\otimes k})-\eta. \end{equation}
It remains to show that this last number is close to (\ref{conversion-1}).
This in turn is true mostly because, by construction, $D_\lozenge(P_\eta,\mathfrak D_\eta(\conv(\fri)))\leq2\eta$ holds and, as will be shown, $D_\lozenge(\conv(\fri),\mathfrak D_\eta(\conv(\fri)))\leq2\eta$ holds. 
\\
We start with the upper bound on $D_\lozenge(\conv(\fri),\mathfrak D_{\eta}(\conv(\fri)))$, which will be derived in a slightly more general way. For arbitrary $s\leq1$ and a compact $A\subset\mathcal C(\hr,\kr)$ 
\begin{equation}
D_\lozenge(\mathfrak D_s(A),A)\leq |s|\cdot\max_{x\in A}||x-\mathfrak D_1\circ x||\leq2|s|
\end{equation}
holds, where the second inequality follows from $A\subset\mathcal C(\hr,\kr)$ in an obvious way and we only prove the first one:
\begin{eqnarray}
\max_{x\in\mathfrak D_s(A)}\min_{y\in A}||x-y||_\lozenge&=&\max_{x\in A}\min_{y\in A}||\mathfrak D_s(x)-y||_\lozenge\\
&=&\max_{x\in A}\min_{y\in A}||(1-s)x+s\mathfrak D_1\circ x-(1-s)y-sy||_\lozenge\\
&\leq&\max_{x\in A}\min_{y\in A}(||(1-s)x-(1-s)y||_\lozenge+||sy-s\mathfrak D_1\circ x||_\lozenge)\\
&\leq&\max_{x\in A}|s|\cdot||x-\mathfrak D_1\circ x||_\lozenge.\label{eqn-direct-infinite-10}
\end{eqnarray}
A similar calculation leads to
\begin{equation}
\max_{x\in A}\min_{y\in \mathfrak D_s(A)}||x-y||_\lozenge\leq|s|\cdot\max_{x\in A}\cdot||x-\mathfrak D_1\circ x||_\lozenge.\label{eqn-direct-infinite-11}
\end{equation}
Application of the triangle inequality for $D_\lozenge$ gives us the estimate
\begin{equation}
D_\lozenge(P_\eta,\conv(\fri))\leq4\eta.
\end{equation}
Lemma 16 in \cite{bbn-2} (originating back to \cite{leung-smith}), finally makes the connection between our set-theoretic approximations and the capacity formula:
\begin{equation}
|\frac{1}{k}\max_{\rho\in\cs(\hr^{\otimes k})}\inf_{\cn\in P_\eta}I_c(\rho,\cn^{\otimes k})-\frac{1}{k}\max_{\rho\in\cs(\hr^{\otimes k})}\inf_{\cn\in \conv(\fri)}I_c(\rho,\cn^{\otimes k})|\leq\nu(8\eta)
\end{equation}
with $\nu(x)=x+8x\log(d_\kr)+4h(x)$. It is obvious that $-\eta\geq-\nu(8\eta)$ holds, therefore for $l$ large enough
\begin{equation}
\frac{1}{l}\log\dim\fr_l\geq\frac{1}{k}\max_{\rho\in\cs(\hr^{\otimes k})}\inf_{\cn\in \conv(\fri)}I_c(\rho,\cn^{\otimes k})-2\nu(8\eta).
\end{equation}
\end{proof}
This leads to the following corollary to Theorem \ref{conversion-of-compound-codes}.
\begin{corollary}\label{achievability-finite-avqc}
For any AVQC $\fri=\{\cn_s  \}_{s\in\bS}$ we have
\begin{equation}\mathcal{A}_{\textup{random}}(\fri)\ge \lim_{l\to\infty}\frac{1}{l}\max_{\rho\in\cs(\hr^{\otimes l})}\inf_{\cn\in \conv(\fri)}I_c(\rho,\cn^{\otimes l}).  \end{equation}
\end{corollary}
Together with Theorem \ref{theorem:converse-random-compound-general} this proves the first part of Theorem \ref{quant-ahlswede-dichotomy}.
\section{\label{sec:derandomization}Achievability of entanglement transmission rate II: Derandomization}

In this section we will prove the second claim made in Theorem \ref{quant-ahlswede-dichotomy} by following Ahlswede's elimination technique. The main result of this section is the following Theorem.
\begin{theorem}\label{dichotomy}
Let  $\fri=\{\cn_s  \}_{s\in\bS}$ be an AVQC. Then $C_{\textup{det}}(\fri)>0$ implies $\A_{\textup{det}}(\fri)=\A_{\textup{random}}(\fri) $.
\end{theorem}
The proof of Theorem \ref{dichotomy} is based mainly on the following lemma, which shows that not much of common randomness is needed to achieve $\A_{\textrm{random}}(\fri)$.
\begin{lemma}[Random Code Reduction]\label{random-code-reduction}
Let $\fri=\{\cn_s\}_{s\in\mathbf S}$ be an AVQC, $l\in\nn$, and $\mu_l$ an $(l,k_l)$-random code for the AVQC $\fri$ with
\begin{equation}\label{eq:random-code-reduction}
e(\mu_l,\fri):=\inf_{s^l\in\bS^l}\int F_e(\pi_{\fr_l},\crr^l\circ\cn_{s^l}\circ\cP^l)d\mu_l(\cP^l,\crr^l)\ge 1-\eps_l
\end{equation}
for a sequence $(\eps_l)_{l\in\nn}$ such that $\eps_l\searrow0$.\\
Let $\eps\in (0,1)$. Then for all sufficiently large $l\in\nn$ there exist $l^2$ codes $\{(\cP^l_i,\crr^l_i):i=1,\ldots ,l^2\}\subset \mathcal C(\fr_l,\hr^{\otimes l})\times\mathcal C(\kr^{\otimes l},\fr'_l)$ such that
\begin{equation}\label{eq:random-code-reduction-a}
\frac{1}{l^2}\sum_{i=1}^{l^2} F_e(\pi_{\fr_l},\crr^l_i\circ\cn_{s^l}\circ\cP^l_i)>1-\eps \qquad  \forall s^l\in\mathbf S^l.
\end{equation}
\end{lemma}
\begin{proof} Before we get into the details, we should note that the whole proof can be read much more easily if one restricts to the case $|\fri|<\infty$ and sets each of the approximating sets occurring in the sequel equal to $\fri$.\\
Let $(\Lambda_i,\Omega_i)$, $i=1,\ldots, K$, be independent random variables with values in $\mathcal C(\fr_l,\hr^{\otimes l})\times\mathcal C(\kr^{\otimes l},\fr'_l) $ which are distributed according to $\mu_l^{\otimes K}$. Let $(P_l)_{l\in\nn}$ be a sequence of polytopes with, for all $l\in\nn$, the properties 
\begin{enumerate}
 \item $P_l\subset \conv(\fri)$
 \item $D_\lozenge(P_l,\conv(\fri))\leq1/l^2$.
\end{enumerate}
Denote by $ext(P_l)$ the extremal points of $P_l$. Consider an indexing such that we can write $ext(P_l)=\{\cn_e\}_{e\in E_l}$ and note that the polytope $P_l$ can be chosen in such a way that $N_l:=|E_l|$ satisfies $N_l\leq(6l)^{4\dim(\hr)^2\dim(\kr)^2}$ (see, for example, Lemma 5.2 in \cite{bbn-1}).\\
For every $e^l\in E_l^l$ and corresponding channel $\cn_{e^l}$, an application of Markov's inequality yields for any $\eps\in (0,1)$ and any $\gamma>0$ the following:
\begin{eqnarray}\label{eq:random-code-reduction-1}
\mathbb P\left( 1-\frac{1}{K}\sum_{i=1}^K F_e(\pi_{\fr_l},\Lambda_i\circ\cn_{e^l}\circ\Omega_i)\geq\eps/2\right)&=&\mathbb P\left(2^{K\gamma-\gamma\sum_{i=1}^K F_e(\pi_{\fr_l},\Lambda_i\circ\cn_{e^l}\circ\Omega_i)}\ge 2^{K\gamma(\eps/2)}\right)\\
&\leq&2^{-K\gamma(\eps/2)}\cdot \mathbb E\left(2^{\gamma(K-\sum_{i=1}^K F_e(\pi_{\fr_l},\Lambda_i\circ\cn_{e^l}\circ\Omega_i))}\right).
\end{eqnarray}
We will derive an upper bound on the expectation in the preceding line:
\begin{eqnarray}\label{eq:random-code-reduction-2}
\mathbb E\left(2^{\gamma(K-\sum_{i=1}^K F_e(\pi_{\fr_l},\Lambda_i\circ\cn_{e^l}\circ\Omega_i))}\right)&=&\mathbb E\left(2^{\gamma(\sum_{i=1}^K(1-F_e(\pi_{\fr_l},\Lambda_i\circ\cn_{e^l}\circ\Omega_i)))}\right)\\
&\overset{(a)}{=}&\left[\mathbb E\left(2^{\gamma(1-F_e(\pi_{\fr_l},\Lambda_1\circ\cn_{e^l}\circ\Omega_1))}\right)\right]^K\\
&\overset{(b)}{\leq}&[\mathbb E(1+2^{\gamma}(1-F_e(\pi_{\fr_l},\Lambda_1\circ\cn_{e^l}\circ\Omega_1)))]^K\\
&\overset{(c)}{\leq}&[1+2^{\gamma}\eps_l]^K.
\end{eqnarray}
We used $(a)$ independence of the $(\Lambda_i,\Omega_i)$, $(b)$ the inequality $2^{\gamma t}\leq (1-t)2^{\gamma\cdot  0}+t2^{\gamma}\le 1+t2^{\gamma},\ t\in[0,1]$, where the first inequality is simply the convexity of $[0,1]\ni t\mapsto 2^{\gamma t}$,  $(c)$ holds by (\ref{eq:random-code-reduction}) and by $P_l\subset \conv(\fri)$.\\
Now, for $K=l^2$, $\gamma=2$ there is an $l_0(\eps)\in\nn$ such that for all $l\ge l_0(\eps)$ we have
\begin{equation}\label{eq:random-code-reduction-3}
  (1+2^{2}\eps_l)^{l^2}\le 2^{l^2 (\eps/2)}.
\end{equation}
Therefore, we obtain from (\ref{eq:random-code-reduction-1}), (\ref{eq:random-code-reduction-2}), and (\ref{eq:random-code-reduction-3}) that for all sufficiently large $l\in\nn$
\begin{equation}\label{eq:random-code-reduction-4}
 \mathbb P\left( 1-\frac{1}{l^2}\sum_{i=1}^{l^2} F_e(\pi_{\fr_l},\Lambda_i\circ\cn_{e^l}\circ\Omega_i)\geq(\eps/2)\right) \le 2^{-l^2(\eps/2)}
\end{equation}
uniformly in $e^l\in E_l^l$.
It follows from (\ref{eq:random-code-reduction-4}) that
\begin{eqnarray}
\mathbb P\left(\frac{1}{l^2}\sum_{i=1}^{l^2} F_e(\pi_{\fr_l},\Lambda_i\circ\cn_{e^l}\circ\Omega_i)>1-\eps/2\ \forall e^l\in E_l^l\right) &\geq&1-N_l^l\cdot 2^{-l^2(\eps/2)}
\end{eqnarray}
implying the existence of a realization $(\cP^l_i,\crr^l_i)_{i=1}^{l^2}$ with
\begin{equation}\frac{1}{l^2}\sum_{i=1}^{l^2} F_e(\pi_{\fr_l},\crr^l_i\circ\cn_{e^l}\circ\cP^l_i)>1-\eps/2 \qquad  \forall e^l\in\mathbf E_l^l  \end{equation}
whenever $N_l^l\cdot 2^{-l^2\eps}<1$, which is clearly fulfilled for all sufficiently large $l\in\nn$.\\
Finally, we note that for every $l\in\nn$ and $\cn_s\in\fri$ there is $\cn_e\in E_l$ such that $||\cn_s-\cn_e||_\lozenge\leq\frac{1}{l^2}$ and, therefore, to every $\cn_{s^l}$ there exists $\cn_{e^l}$ (with each $\cn_{e_i}\in E_l$) such that (see the proof of Lemma 5.2 in \cite{bbn-1} for details)
\begin{equation}
 ||\cn_{s^l}-\cn_{e^l}||_\lozenge\leq\sum_{i=1}^l||\cn_{s_i}-\cn_{e_i}||_\lozenge\leq\frac{1}{l},
\end{equation}
and therefore for every $s^l\in\bS^l$ we have, for a maybe even larger $l$ as before (satisfying $1/l<\eps/2$, additionally),
\begin{equation}\frac{1}{l^2}\sum_{i=1}^{l^2} F_e(\pi_{\fr_l},\crr^l_i\circ\cn_{s^l}\circ\cP^l_i)>1-\eps \qquad  \forall s^l\in\mathbf S^l.  \end{equation}
\end{proof}
We proceed with the proof of Theorem \ref{dichotomy}. Since $C_{\textrm{det}}(\fri)>0$ according to the assumption of the theorem, there is an $(m_l,l^2)$-deterministic code $\mathfrak{C}_{m_l}=(\rho_i, D_i)_{i=1}^{l^2}$ with $\rho_1,\ldots,\rho_{l^2}\in\cs(\hr^{\otimes m_l})$, $D_1,\ldots, D_{l^2}\in\mathcal{B}(\kr^{\otimes m_l})$ with $m_l=o(l)$ and
\begin{equation}\label{eq:dichotomy-1}
  \bar P_{e,m_l}=\sup_{s^{m_l}\in\bS^{m_l}}P_{e}(\mathfrak{C}_{m_l},s^{m_l})\le \eps.
\end{equation}
On the other hand, let us consider an $(l,k_l)$-random code as in Lemma \ref{random-code-reduction}, i.e. with
\begin{equation}\label{eq:dichotomy-2}
 \frac{1}{l^2}\sum_{i=1}^{l^2} F_e(\pi_{\fr_l},\crr^l_i\circ\cn_{s^l}\circ\cP^l_i)>1-\eps \qquad  \forall s^l\in\mathbf S^l.
\end{equation}
Define CPTP maps $\cP^{l+m_l}\in\mathcal{C}(\fr_l, \hr^{\otimes l+m_l})$, $\crr^{l+m_l}\in \mathcal{C}(\kr^{\otimes l+ m_l},\fr'_l)$ by
\begin{equation}\label{eq:dichotomy-3}
 \cP^{l+m_l}(a):=\frac{1}{l^2}\sum_{i=1}^{l^2} \cP^l_i(a)\otimes \rho_i\quad \textrm{and}\quad \crr^{l+m_l}(b\otimes d):=\sum_{i=1}^{l^2}\tr(D_i d)\crr^l_i(b).
\end{equation}
Then for each $s^{l+m_l}=(v^l,u^{m_l},)\in\bS^{l+m_l}$
\begin{eqnarray}\label{eq:dichotomy-4}
  F_e(\pi_{\fr_l},\crr^{l+m_l}\circ (\cn_{v^l}\otimes \cn_{u^{m_l}}  )\circ \cP^{l+m_l})&=&\frac{1}{l^2}\sum_{i,j=1}^{l^2}\tr (D_j \cn_{u^{m_l}}(\rho_i))F_e(\pi_{\fr_l},\crr_j^l\circ \cn_{v^l}\circ \cP_i^l )\\
&\ge & \frac{1}{l^2}\sum_{i=1}^{l^2}\tr (D_i \cn_{u^{m_l}}(\rho_i))F_e(\pi_{\fr_l},\crr_i^l\circ \cn_{v^l}\circ \cP_i^l ),
\end{eqnarray}
where in the last line we have used that all involved terms are non-negative. In order to show that the fidelity on the left-hand side of (\ref{eq:dichotomy-4}) is at least $1-2\eps$ we need the following lemma from \cite{ahlswede-elimination}.
\begin{lemma}\label{innerproduct-lemma}
Let $K\in\nn$ and real numbers $a_1,\ldots,a_K,b_1,\ldots, b_K\in [0,1]$ be given. Assume that
\begin{equation}\frac{1}{K}\sum_{i=1}^Ka_i\ge 1-\eps\qquad \textrm{and} \qquad \frac{1}{K}\sum_{i=1}^K b_i\ge 1-\eps,  \end{equation}
hold. Then
\begin{equation}\frac{1}{K}\sum_{i=1}^Ka_ib_i\ge 1-2\eps.  \end{equation}
\end{lemma}
Applying this lemma with $K=l^2$,
\begin{equation} a_i= \tr (D_i \cn_{u^{m_l}}(\rho_i)),\quad\textrm{and}\quad b_i=F_e(\pi_{\fr_l},\crr_i^l\circ \cn_{v^l}\circ \cP_i^l ) \end{equation}
along with (\ref{eq:dichotomy-1}), (\ref{eq:dichotomy-2}), and (\ref{eq:dichotomy-4}) shows that
\begin{equation}\label{eq:dichotomy-5}
  F_e(\pi_{\fr_l},\crr^{l+m_l}\circ (\cn_{v^l}\otimes \cn_{u^{m_l}}  )\circ \cP^{l+m_l})\ge 1-2\eps.
\end{equation}
On the other hand we know from Theorem \ref{conversion-of-compound-codes} that for each sufficiently small $\eta>0$ there is a random code $\mu_l$ for the AVQC $\fri$ with
\begin{equation}\label{eq:dichotomy-6}
 \frac{1}{l}\log \dim \fr_l\ge \frac{1}{k}\max_{\rho\in\cs(\hr^{\otimes k})}\inf_{\cn\in\conv(\fri)}I_c(\rho,\cn^{\otimes k})  -\eta,
\end{equation}
and
\begin{equation}e(\mu_l,\fri)=\inf_{s^l\in\bS^l}\int F_e(\pi_{\fr_l},\crr^l\circ\cn_{s^l}\circ\cP^l)d\mu_l(\cP^l,\crr^l)\ge 1-(l+1)^{N_{\eta}}2^{-lc}  \end{equation}
for all sufficiently large $l$ with $c=c(k,\dim\hr,\dim\kr,\conv(\fri),\eta)$ and $N_{\eta}\in\nn$. Thus the arguments that led us to (\ref{eq:dichotomy-5}) show that for all sufficiently large $l$ there is a deterministic $(l+m_l,k_{l})$-code for the AVQC $\fri$ with
\begin{equation}F_e(\pi_{\fr_l},\crr^{l+m_l}\circ (\cn_{v^l}\otimes \cn_{u^{m_l}}  )\circ \cP^{l+m_l})\ge 1-2\eps,\end{equation}
and
\begin{equation}\frac{1}{l+m_l}\log\dim\fr_l\ge\frac{1}{k}\max_{\rho\in\cs(\hr^{\otimes k})}\inf_{\cn\in\conv(\fri)}I_c(\rho,\cn^{\otimes k})  -2\eta \end{equation}
by (\ref{eq:dichotomy-6}) and since $m_l=o(l)$. This shows that $\A_{\textup{det}}(\fri)\ge\A_{\textup{random}}(\fri)$. Since the reverse inequality is trivially true we are done.
\section{\label{sec:symmetrizability}Zero-capacity-conditions: Symmetrizability}
The most basic quality feature of an information processing system is whether it can be used for communication at a positive rate or not. This applies especially to such rather complex systems as AVCs or AVQCs. The notion of symmetrizability stems from the theory of classical AVCs and it addresses exactly that question. A classical AVC has deterministic capacity for message transmission equal to zero if and only if it is symmetrizable (with the definition of symmetrizability adjusted to the two different scenarios 'average error criterion' and 'maximal error criterion') \cite{ericson} and \cite{csiszar-narayan}, \cite{kiefer-wolfowitz}.\\
Of course, a similar statement for $\A_{\textup{det}}$ would be of great interest.
\\\\
In this section we give three different conditions for three different capacities of an AVQC to be equal to zero. We restrict ourselves to the case $|\fri|<\infty$. Starting with the statement that has the weakest information theoretic consequences, we proceed to stronger statements. The case $|\fri|=\infty$ requires some involved continuity issues which shall be carried out elsewhere.\\
All three conditions have in common that they enable the adversary to simulate, on average over some probability distribution, a different output at the receiver side than the one that was originally put into the channel by the sender. The first two conditions, dealing with message transmission, exhibit a possibly nonlinear dependence between message set and probability distribution. They are direct (but not single-letter) analogs of their classical counterparts.\\
The third one is a sufficient condition for $\A_{\textup{random}}$ to be equal to zero. It employs a linear dependence between input state and probability distribution. If this condition is valid, the adversary is not only able to simulate a wrong output, he can also simulate an entanglement breaking channel between sender and receiver. In contrast to the first two criteria, this third one is a single-letter criterion.\\
There is a fourth and, at first sight, trivial condition, given by the following: An AVQC $\fri=\{\cn_s\}_{s\in\bS}$ has (deterministic \emph{and} random) capacity for transmission of entanglement equal to zero if there is an $s\in\bS$ such that $\cn_s$ has zero capacity for transmission of entanglement.\\
We note that this fourth condition is nontrivial only because of the following reason: there is, until now, no way of telling exactly when a given (memoryless) quantum channel has a capacity greater than zero (except for calculating (\ref{eq:ahlswede-dichotomy-1}) for a single channel, an awkward task in general). This is in sharp contrast to the classical case, where the question can be trivially answered: A classical memoryless channel has a nonzero capacity if and only if there are at least two input states that lead to different output states.
\\
Since our results do not answer the question whether it can happen that $C_{\textup{det}}(\fri)=0$, $\A_{\textup{det}}(\fri)=0$ and $\A_{\textup{random}}(\fri)>0$ hold simultaneously for a given AVQC $\fri$, we are left with two interesting and intimately related questions:\\
First, there is the zero-capacity question for single memoryless channels. Second, we need to find a criterion telling us exactly when $\A_{\textup{det}}$ is equal to zero.
\subsection{\label{subsec:classical-deterministic-average-error}Classical capacity with deterministic codes and average error}
We now introduce a notion of symmetrizability which is a sufficient and necessary condition for $C_{\textrm{det}}(\fri)=0$. 
Our approach is motivated by the corresponding concept for arbitrarily varying channels with classical input and quantum output (cq-AVC) given in \cite{ahlswede-blinovsky}.\\
A nontrivial example of a non-symmetrizable AVQC can be found in subsection \ref{subsec:Erasure-AVQC}, see step $\mathbf D$ in the proof of Lemma \ref{lemma-erasure-example}.
\begin{definition}\label{def:c-symmetrizability}
Let $\bS$ be a finite set and $\fri=\{\cn_s  \}_{s\in\bS}$ an AVQC.
\begin{enumerate}
\item $\fri$ is called $l$-symmetrizable, $l\in\nn$, if for each finite set $\{\rho_1,\ldots,\rho_K  \}\subset \cs(\hr^{\otimes l})$, $K\in \nn$, there is a map $p:\{\rho_1,\ldots, \rho_K  \}\to \mathfrak{P}(\bS^l)$ such that for all $i,j\in\{1,\ldots, K  \}$
\begin{equation}\label{eq:c-symmetrizable}
\sum_{s^l\in\bS^l}p(\rho_i)(s^l)\cn_{s^l}(\rho_j)= \sum_{s^l\in\bS^l}p(\rho_j)(s^l)\cn_{s^l}(\rho_i)
\end{equation}
holds.
\item We call $\fri$ symmetrizable if it is $l$-symmetrizable for all $l\in\nn$.
\end{enumerate}
\end{definition}
We now state the main statement of this section.
\begin{theorem}\label{symm-equiv-C-0}
Let $\fri=\{\cn_s  \}_{s\in \bS}$, $|\bS|<\infty$, be an AVQC. Then $\fri$ is symmetrizable if and only if $C_{\textup{det}}(\fri)=0$.
\end{theorem}
\begin{proof} 1. ``Symmetrizability implies $C_{\textrm{det}}(\fri)=0$''.\\
The proof follows closely the corresponding arguments given in \cite{ericson}, \cite{csiszar-narayan}, and \cite{ahlswede-blinovsky}. We give the full proof for reader's convenience. Suppose that $\fri=\{\cn_s  \}_{s\in \bS} $ is symmetrizable and let $(\rho_i, D_i)_{i=1}^M$, $M\ge 2$, be a code for transmission of messages over $\fri$ with $\{\rho_1,\ldots, \rho_M  \}\subset\cs(\hr^{\otimes l})$ and POVM $\{D_i\}_{i=1}^M$ on $\hr^{\otimes l}$. Since $\fri$ is symmetrizable there is a map $p:\{\rho_1,\ldots, \rho_M \}\to \mathfrak{P}(\bS^l)$ such that for all $i,j\in\{1,\ldots, M  \}$
\begin{equation}\label{eq:symm-1}
 \sum_{s^l\in\bS^l}p(\rho_i)(s^l)\cn_{s^l}(\rho_j)= \sum_{s^l\in\bS^l}p(\rho_j)(s^l)\cn_{s^l}(\rho_i)
\end{equation}
For $s^l\in\bS^l$ and $i\in\{1,\ldots, M\}$ we set
\begin{equation}\label{eq:symm-2}
  e(i,s^l):=1-\tr (\cn_{s^l}(\rho_i)D_i)=\sum_{\substack{j=1 \\ j\neq i}}^M \tr(\cn_{s^l}(\rho_i)D_j).
\end{equation}
For $k\in\{1,\ldots,M  \}$ let $S_k^l$ be a random variable taking values in $\bS^l$ and which is distributed according to $(p(\rho_k )(s^l))_{s^l\in\bS^l}$. Then using relation (\ref{eq:symm-2}) we can write
\begin{eqnarray}\label{eq:symm-3}
  \mathbb{E}(e(i, S_k^l) )&=& \sum_{s^l\in\bS^l}\sum_{\substack{j=1 \\ j\neq i}}^M p( \rho_k )(s^l)  \tr(\cn_{s^l}(\rho_i)D_j)\\
&=& \sum_{\substack{j=1 \\ j\neq i}}^M\tr \{\sum_{s^l\in\bS^l}p(\rho_k )(s^l)\cn_{s^l}(\rho_i)D_j    \}\\
&=& \sum_{\substack{j=1 \\ j\neq i}}^M\tr ( \sum_{s^l\in\bS^l}p(\rho_i )(s^l)\cn_{s^l}(\rho_k)   D_j )\\
&=& \sum_{\substack{j=1 \\ j\neq i}}^M\sum_{s^l\in\bS^l}p( \rho_i)(s^l)  \tr(\cn_{s^l}(\rho_k)D_j),
\end{eqnarray}
where the third line is by (\ref{eq:symm-1}). On the other hand we have
\begin{equation}\label{eq:symm-4}
  \mathbb{E}(e(k,S_i^l) )=\sum_{s^l\in\bS^l}\sum_{\substack{j=1 \\ j\neq k}}^M p(\rho_i)(s^l)  \tr(\cn_{s^l}(\rho_k)D_j).
\end{equation}
Since $\{D_i\}_{i=1}^M $ is a POVM (\ref{eq:symm-3}) and (\ref{eq:symm-4}) imply that for $i\neq k$
\begin{equation}\label{eq:symm-5}
  \mathbb{E}(e(i, S_k^l) )+ \mathbb{E}(e(k,S_i^l) )\ge 1
\end{equation}
holds. Let us abbreviate $\mathfrak{C}:=(\rho_i,D_i)_{i=1}^M$, then with
\begin{equation}\bar{P}_e(\mathfrak{C},s^l)=\frac{1}{M}\sum_{k=1}^M(1-\textrm{tr}(\cn_{s^l}(\rho_k)D_k ))  \end{equation}
for $s^l\in\bS^l$ we obtain
\begin{eqnarray}\label{eq:symm-6}
  \mathbb{E}(\bar{P}_e(\mathfrak{C},S_j^l) )&=& \sum_{s^l\in\bS^l}p(\rho_j)(s^l)\frac{1}{M}\sum_{k=1}^M(1-\tr(\cn_{s^l}(\rho_k)D_k ))\\
&=& \frac{1}{M}\sum_{k=1}^M \mathbb{E}(e(k, S_j^l)).
\end{eqnarray}
(\ref{eq:symm-5}) and (\ref{eq:symm-6}) yield
\begin{eqnarray}
  \frac{1}{M}\sum_{j=1}^M \mathbb{E}(\bar{P}_e(\mathfrak{C},S_j^l) )&=& \frac{1}{M^2}\sum_{i,j=1}^{M}\mathbb{E}(e(k,S_j^l) )\\
&\ge& \frac{1}{M^2} \binom{M}{2}\\
&=& \frac{M-1}{2M}\ge \frac{1}{4}
\end{eqnarray}
for $M\ge 2$. Thus it follows that there is at least one $j\in\{1,\ldots,M  \}$ with
\begin{equation} \mathbb{E}(\bar{P}_e(\mathfrak{C},S_j^l) )\ge \frac{1}{4} \end{equation}
and consequently there is at least one $s^l\in\bS^l$ with
\begin{equation} \bar{P}_e(\mathfrak{C},s^l)\ge \frac{1}{4}  \end{equation}
implying that $C_{\textup{det}}(\fri)=0 $.\\
2. ``$C_{\textrm{det}}(\fri)=0$ implies symmetrizability''.\\
Suppose that $\fri$ is non-symmetrizable. Then there is an $\hat l\in\nn$ and a finite set $\{\rho_x\}_{x\in\X}\subset \cs(\hr^{\otimes \hat l})$ such that for no map $p:\{\rho_x\}_{x\in\X}\to\mathfrak{P}(\bS^{\hat l})$ the relation (\ref{eq:c-symmetrizable}) holds. Let us define for each $s^{\hat l}\in\bS^{\hat l}$ a cq-channel $\X\ni x\mapsto W_{s^{\hat l}}(x):=\cn_{s^{\hat l}} (\rho_x)\in \cs(\kr^{\otimes \hat l})$, and consider the cq-AVC generated by the set $\fri_{cq}:=\{ W_{s^{\hat l}} \}_{s^{\hat l}\in \bS^{\hat l}}$. Then, due to the assumed non-symmetrizability of $\fri$, our new cq-AVC $\fri_{cq}$ is non-symmetrizable in the sense of \cite{ahlswede-blinovsky}.\\
Since $\fri_{cq}$ is non-symmetrizable the reduction argument from \cite{ahlswede-blinovsky} to the results  of \cite{csiszar-narayan} show that the cq-AVC $ \fri_{cq}$ has positive capacity.
This implies the existence of a sequence $(K_m,f_m,D_m,\eps_m)_{m\in\nn}$, where $K_m\in\nn$, $f_m:\{1,\ldots,K_m\}\rightarrow\X^m$, $D_m\in\B_+(\hr^{\otimes l\cdot m})$, $\lim_{m\rightarrow\infty}\eps_m\searrow0$, $\liminf_{m\rightarrow\infty}\frac{1}{m}\log K_m=c>0$ and $\frac{1}{K_m}\sum_{i=1}^{K_m}(1-\tr(D_iW^m_{s^{\hat l}}(f(i))))=\eps_m$.\\
We may use this sequence to construct another sequence $(\rho_i, D_i)_{i=1}^{M_l}$ of deterministic codes for message transmission over $\fri$, thereby achieving a capacity of $\frac{1}{\hat l}c>0$. A similar construction is carried out explicitly at the end of the proof of the following Theorem \ref{theorem:c-det=0-for-maximal-error}.
\end{proof}
\begin{corollary}\label{q-det-=0}
If the AVQC $\fri=\{\cn  \}_{s\in \bS}$ is symmetrizable then $\mathcal{A}_{\textup{det}}(\fri)=0$.
\end{corollary}
\begin{proof} Note that $\mathcal{A}_{\textrm{det}}(\fri)\le C_{\textrm{det}}(\fri)  $ and apply Theorem \ref{symm-equiv-C-0}.
\end{proof}
\subsection{\label{subsec:Classical-capacity-with-deterministic-codes-and-maximal error}Classical capacity with deterministic codes and maximal error}
We will now investigate, when exactly it is possible to send classical messages at positive rate over a finite AVQC, with the error criterion being that of maximal rather than average error.\\
\begin{theorem}\label{theorem:c-det=0-for-maximal-error}
Let $\fri=\{\cn_s\}_{s\in\bS}\subset\mathcal C(\hr,\kr)$ be a finite AVQC. The classical deterministic maximal error capacity $C_{\det,\max}(\fri)$ of $\fri$ is equal to zero if and only if for every $l\in\nn$ and every set $\{\rho_1,\rho_2\}\subset\cs(\hr^{\otimes l})$ we have
\begin{equation}\label{eq:equiv-max-1}
\conv(\{\cn_{s^l}(\rho_1)\}_{s^l\in\bS^l})\cap \conv(\{\cn_{s^l}(\rho_2)\}_{s^l\in\bS^l})\neq\emptyset.
\end{equation}
\end{theorem}
\begin{proof} We closely follow the line of proof given in \cite{kiefer-wolfowitz}. Let us begin with the 'if' part. Let $K,l\in\nn$, $\{\rho_1,\ldots,\rho_K\}\subset\cs(\hr^{\otimes l})$ and $D_1,\ldots,D_K\in\mathcal B_+(\kr^{\otimes l})$ with $\sum_{i=1}^KD_i=\mathbf1_{\kr^{\otimes l}}$ be a code for transmission of classical messages over $\fri$.\\
We show that the maximal error probability of this code is bounded away from zero for large enough $l$.\\
Let, without loss of generality, $l$ be such that
\begin{align}
&\tr(D_1\cn_{s^l}(\rho_1))>1/2\ \ \ \ \ \forall\ s^l\in\bS^l&\label{eqn-mep-1}\\
&\tr(D_2\cn_{s^l}(\rho_2))>1/2\ \ \ \ \ \forall\ s^l\in\bS^l.&\label{eqn-mep-2}
\end{align}
We show that there is a contradiction between (\ref{eqn-mep-1}) and (\ref{eqn-mep-2}).
By assumption, there exist probability distributions $p_1,p_2\in\mathfrak P(\bS^l)$ such that
\begin{equation}\label{eqn-mep-3}
\sum_{s^l\in\bS^l}p_1(s^l)\cn_{s^l}(\rho_1)=\sum_{s^l\in\bS^l}p_2(s^l)\cn_{s^l}(\rho_2).
\end{equation}
Of course, (\ref{eqn-mep-1}) implies
\begin{align}
\sum_{s^l\in\bS^l}p_1(s^l)\tr(D_1\cn_{s^l}(\rho_1))>1/2.
\end{align}
Together with (\ref{eqn-mep-3}) this leads to
\begin{align}
1/2&<\sum_{s^l\in\bS^l}p_1(s^l)\tr(D_1\cn_{s^l}(\rho_1))&\\
&=\sum_{s^l\in\bS^l}p_2(s^l)\tr(D_1\cn_{s^l}(\rho_2))&\\
&\leq\sum_{s^l\in\bS^l}p_2(s^l)\tr((D_1+\sum_{i=3}^KD_i)\cn_{s^l}(\rho_2))&\\
&=\sum_{s^l\in\bS^l}p_2(s^l)\tr((\mathbf1-D_2)\cn_{s^l}(\rho_2))&\\
&=1-\sum_{s^l\in\bS^l}p_2(s^l)\tr(D_2\cn_{s^l}(\rho_2))&\\
&<1-1/2,&\label{eqn-mep-4}
\end{align}
a clear contradiction. Thus, for every code the maximal error probability is bounded from below by $1/2$.\\
Let us turn to the 'only if' part.\\
Assume there is an ${\hat{l}}\in\nn$ and a set $\{\rho_{1},\rho_{2}\}\subset\cs(\hr^{\otimes {\hat{l}}})$ such that
\begin{align}
\conv(\{\cn_{s^{\hat{l}}}(\rho_{1})\}_{s^{\hat{l}}\in\bS^{\hat{l}}})\cap \conv(\{\cn_{s^{\hat{l}}}(\rho_{2})\}_{s^{\hat{l}}\in\bS^{\hat{l}}})=\emptyset.\label{eqn-mep-5}
\end{align}
Thus, there exists a self adjoint operator $A\in\mathcal B(\kr^{\otimes {\hat{l}}})$ such that
\begin{align}
\tr(A\rho)<0\ \ \forall\ \rho\in \conv(\{\cn_{s^{\hat{l}}}(\rho_{1})\}_{s^{\hat{l}}\in\bS^{\hat{l}}}),\ \ \ \ \tr(A\rho)>0\ \ \forall\ \rho\in \conv(\{\cn_{s^{\hat{l}}}(\rho_{2})\}_{s^{\hat{l}}\in\bS^{\hat{l}}})\label{eqn-mep-6}.
\end{align}
Let $A$ have a decomposition $A=\sum_{x=1
}^da_xA_x$, where $a_x$ are real numbers (including the possibility of $a_x=0$ for some $x$) and $A_x$ are one dimensional projections fulfilling $\sum_{x=1}^dA_x=\mathbf1_{\kr^{\otimes {\hat{l}}}}$.
For every $m\in\nn$, define
\begin{align}
P_1^m:=\sum_{x^m : \frac{1}{m}\sum_{i=1}^ma_{x_i}<0}A_{x_1}\otimes\ldots\otimes A_{x_m},\ \ \ \ P_2^m:=\sum_{x^m : \frac{1}{m}\sum_{i=1}^ma_{x_i}\geq0}A_{x_1}\otimes\ldots\otimes A_{x_m}\label{eqn-mep-7}.
\end{align}
Then $P_1^m+P_2^m=\mathbf1_{\kr^{\otimes {\hat{l}}\cdot m}}$. Let us denote elements of $\bS^{{\hat{l}}m}$ by $s^{{\hat{l}}m}=(s_1^{\hat{l}},\ldots,s_m^{\hat{l}})$, where each $s_i^{\hat{l}}\in\bS^{\hat{l}}$.\\
To every $s^{\hat{l}}\in\bS^{\hat{l}}$, define probability distributions $p_{s^{\hat{l}}},q_{s^{\hat{l}}}\in\cs(\{1,\ldots,d\})$ according to
\begin{align}
p_{s^{\hat{l}}}(x):=\tr(A_x\cn_{s^{\hat{l}}}(\rho_{1})),\ \ \ \ \ q_{s^{\hat{l}}}(x):=\tr(A_x\cn_{s^{\hat{l}}}(\rho_{2})),\ \ \forall\ x\in\{1,\ldots,d\}
\end{align}
and to every $s^{{\hat{l}}m}\in\bS^{{\hat{l}}m}$ we associate two real numbers $\bar A_{s^{{\hat{l}}m}}(\rho_{1}),\bar A_{s^{{\hat{l}}m}}(\rho_{2})$ by
\begin{align}\bar A_{s^{{\hat{l}}m}}(\rho_{1}):=\sum_{s^{\hat{l}}\in\bS^{\hat{l}}}\frac{1}{m}N(s^{\hat{l}}|s^{{\hat{l}}m})\tr(A\cn_{s^{\hat{l}}}(\rho_{1})),\ \ \ \ \ \bar A_{s^{{\hat{l}}m}}(\rho_{2}):=\sum_{s^{\hat{l}}\in\bS^{\hat{l}}}\frac{1}{m}N(s^{\hat{l}}|s^{{\hat{l}}m})\tr(A\cn_{s^{\hat{l}}}(\rho_{2})),\end{align}
with natural numbers $N(s^{\hat{l}}|s^{{\hat{l}}m}):=|\{i:s^{\hat{l}}_i=s^{\hat{l}},\ i\in\{1,\ldots,m\}\}|$ for every $s^{\hat{l}m}\in\bS^{\hat{l}m}$ and $s^{\hat l}\in\bS^{\hat l}$.\\
Obviously, $\bar A_{s^{{\hat{l}}m}}(\rho_{1})<0$ and $\bar A_{s^{{\hat{l}}m}}(\rho_{2})>0$. Setting
\begin{align}C:=\max_{(s^{{\hat{l}}},X)\in\bS^{\hat{l}}\times\{\rho_{1},\rho_{2}\}}(\tr(A\cn_{s^{\hat{l}}}(X))/2)^{-2}(\tr(A^2\cn_{s^{\hat{l}}}(X))-\tr(A\cn_{s^{\hat{l}}}(X))^2)\end{align}
we arrive, by application of Chebyshev's inequality and for every $s^{{\hat{l}}m}=(s^{\hat{l}}_1,\ldots,s^{\hat{l}}_m)\in\bS^{{\hat{l}}m}$ at
\begin{eqnarray}
\tr(P_1^m\cn_{s^{{\hat{l}}m}}(\rho_{1}^{\otimes m}))&=&\sum_{x^m:\frac{1}{m}\sum_{i=1}^ma_{x_i}<0}\tr(A_{x_1}\otimes\ldots\otimes A_{x_m}\cn_{s^{\hat{l}}_1}(\rho_{1})\otimes\ldots\otimes\cn_{s^{\hat{l}}_m}(\rho_{1}))\\
&=&\sum_{x^m:\frac{1}{m}\sum_{i=1}^ma_{x_i}<0}p_{s^{\hat{l}}_1}(x_1)\cdot\ldots\cdot p_{s^{\hat{l}}_m}(x_m)\\
&\geq&\sum_{x^m:|\frac{1}{m}\sum_{i=1}^ma_{x_i}-\bar A_{s^{{\hat{l}}m}}(\rho_{1})|\leq|\bar A_{s^{{\hat{l}}m}}(\rho_{1})/2|}p_{s^{\hat{l}}_1}(x_1)\cdot\ldots\cdot p_{s^{\hat{l}}_m}(x_m)\\
&\geq&1-\frac{1}{m}(\bar A_{s^{{\hat{l}}m}}(\rho_{1})/2)^{-2}\sum_{s^{\hat{l}}\in\bS^{\hat{l}}}\frac{1}{m}N(s^{\hat{l}}|s^{{\hat{l}}m})(\tr(A^2\cn_{s^{\hat{l}}}(\rho_{1}))-\tr(A\cn_{s^{\hat{l}}}(\rho_{1}))^2)\\
&\geq&1-\frac{1}{m}\max_{s^{{\hat{l}}}}(\tr(A\cn_{s^{\hat{l}}}(\rho_{1}))/2)^{-2}(\tr(A^2\cn_{s^{\hat{l}}}(\rho_{1}))-\tr(A\cn_{s^{\hat{l}}}(\rho_{1}))^2)\\
&\geq&1-\frac{1}{m}\cdot C.\label{eqn-mep-10}
\end{eqnarray}
In the very same way, we can prove that
\begin{eqnarray}
\tr(P_2^m\cn_{s^{{\hat{l}}m}}(\rho_{2}^{\otimes m}))&=&\sum_{x^m:\frac{1}{m}\sum_{i=1}^ma_{x_i}\geq0}\tr(A_{x_1}\otimes\ldots\otimes A_{x_m}\cn_{s^{\hat{l}}_1}(\rho_{2})\otimes\ldots\otimes\cn_{s^{\hat{l}}_m}(\rho_{2}))\\
&=&\sum_{x^m:\frac{1}{m}\sum_{i=1}^ma_{x_i}\geq0}q_{s^{\hat{l}}_1}(x_1)\cdot\ldots\cdot q_{s^{\hat{l}}_m}(x_m)\\
&\geq&\sum_{x^m:|\frac{1}{m}\sum_{i=1}^ma_{x_i}-\bar A_{s^{{\hat{l}}m}}(\rho_{2})|\leq |\bar A_{s^{{\hat{l}}m}}(\rho_{2})|/2}q_{s^l
_1}(x_1)\cdot\ldots\cdot q_{s^{\hat{l}}_m}(x_m)\\
&\geq&1-\frac{1}{m}\max_{s^{{\hat{l}}}}(\tr(A\cn_{s^{\hat{l}}}(\rho_{2}))/2)^{-2}(\tr(A^2\cn_{s^{\hat{l}}}(\rho_{2}))-\tr(A\cn_{s^{\hat{l}}}(\rho_{2}))^2)\\
&\geq&1-\frac{1}{m}\cdot C.\label{eqn-mep-11}
\end{eqnarray}
Take any $0<\eps<1/4$. Let $m'=\min\{m\in\nn:\frac{1}{m}\cdot C<\eps\}$. Then
\begin{align}
& \tr(P_1^{m'}\cn_{s^{{\hat{l}}m'}}(\rho_{1}^{\otimes m'}))\geq1-\eps & \tr(P_2^{m'}\cn_{s^{{\hat{l}}m'}}(\rho_{2}^{\otimes m'}))\geq1-\eps \label{eqn-mep-12a}
\end{align}
hold. Consider the classical AVC given by the family $J:=\{c_{\nu,\delta}\}_{\delta,\nu\in[3/4,1]}$ of classical channels $c_{\nu,\delta}:\{0,1\}\rightarrow\{0,1\}$ with stochastic matrices defined via $c_{\nu,\delta}(1|1):=1-\nu,\ c_{\nu,\delta}(2|2):=1-\delta$. Clearly, $J$ is a convex set and, for every $c_{\nu,\delta}\in J$ we have that
\begin{align}
\max_{p\in\mathfrak P(\{0,1\})}I(p,c_{\nu,\delta})&\geq1-\frac{1}{2}(h(\nu)+h(\delta))&\\
&\geq1-h(3/4)&\\
&>0,&\label{eqn-mep-12}
\end{align}
where $I(p,c_{\nu,\delta})$ is the mutual information of the probability distribution $q$ on $\{1,2\}\times\{1,2\}$ which is generated by $p$ and $c_{\nu,\delta}$ through $q(i,j):=p(i)c_{\nu,\delta}(j|i)$ ($(i,j)\in\{1,2\}\times\{1,2\}$). The lower bound given here is calculated using an equidistributed input. Note further that for this special AVC, with notation taken from \cite{ahlswede-wolfowitz-2}, $\bar{\bar{J}}=\conv(J)=J$.\\
At this point in their proof of the classical zero-capacity-condition for AVCs \cite{kiefer-wolfowitz}, Kiefer and Wolfowitz made reference to a result by Gilbert \cite{gilbert}, who proved existence of codes that achieve a positive rate. Kiefer and Wolfowitz used these codes for message transmission over an AVC with binary input and output alphabet. Our strategy of proof is to use the existence of codes for AVCs with binary input and output that is guaranteed by Theorem 1 of \cite{ahlswede-wolfowitz-2} instead. Together with (\ref{eqn-mep-12}) this theorem gives us the existence of a number $C'>0$, a function $\kappa:\nn\rightarrow\mathbb R$ with $\lim_{r\rightarrow\infty}\kappa(r)=0$ and a sequence $(M^r,f^r,\eps_r,(D_1^r,\ldots,D_{|M^r|}^r))_{r\in\nn}$ where for each $r\in\nn$:
\begin{enumerate}
\item $M^r=\{1,\ldots,N\}$ is a finite set of cardinality $N=|M^r|=2^{r(C'-\kappa(r))}$,
\item $f^r:M^r\rightarrow\{1,2\}^r$,
\item $\eps_r\geq0$ and $\lim_{r\rightarrow\infty}\eps_r=0$,
\item $D^r_1,\ldots,D^r_{|M^r|}\subset\{1,2\}^n$ are pairwise disjoint and
\item for every sequence $x^r\in([3/4,1]\times[3/4,1])^r$ and every $i\in M^r$ we have that
\begin{align}
\sum_{y^n\in D_i^r}\prod_{j=1}^rc_{x_j}(y_j|f^r(i)_j)\geq1-\eps_r.\label{eqn-mep-13}
\end{align}
\end{enumerate}
For ${n}\in\nn$, take the unique numbers $r\in\nn$, $t\in\{0,\ldots,m'-1\}$ such that ${n}=m'r+t$ holds. The code for $\fri$ is then defined as follows:
\begin{align}
&M_{n}:=M^r,&\\
&f_{n}(i):=(\rho_{f^r(i)_1})^{\otimes m'}\otimes\ldots\otimes(\rho_{f^r(i)_r})^{\otimes m'}\otimes\sigma^{\otimes t},&\\
&P_i^{n}:=\sum_{y^r\in D_i^r}P^{m'}_{y_1}\otimes\ldots\otimes P^{m'}_{y_r}\otimes\mathbf1_\kr^{\otimes t}.&
\end{align}
Let, for every $s^{m'}\in\bS^{m'}$, $x=(\nu,\delta)\in[3/4,1]^2$ be such that
\begin{align}
& c_{\nu,\delta}(0|0):=\tr(P_1^{m'}\cn_{s^{{n}m'}}(\rho_{1}^{\otimes m'}))=1-\nu & c_{\nu,\delta}(1|1):=\tr(P_1^{m'}\cn_{s^{{n}m'}}(\rho_{2}^{\otimes m'}))=1-\delta. \label{eqn-mep-14}
\end{align}
Then for every $s^{n}\in\bS^{n}$ we use the decomposition $s^{n}=(s^{m'}_1,\ldots,s^{m'}_r,s^t)$ and get, using equation (\ref{eqn-mep-13}) and the definition (\ref{eqn-mep-14}), for every $i\in M_{n}$,
\begin{align}
\tr\{P_i^{n}f_{n}(i)\}&=\tr\{[\sum_{y^r\in D_i^r}P^{m'}_{y_1}\otimes\ldots\otimes P^{m'}_{y_r}\otimes\mathbf1_\kr^{\otimes t}](\rho_{f^r(i)_1})^{\otimes m'}\otimes\ldots\otimes(\rho_{f^r(i)_r})^{\otimes m'}\otimes\sigma^{\otimes t}\}&\\
&=\sum_{y^r\in D_i^r}\prod_{j=1}^r\tr\{P^{m'}_{y_j}(\rho_{f^r(i)_j})^{\otimes m'}\}&\\
&=\sum_{y^r\in D_i^r}\prod_{j=1}^rc_{\nu,\delta}(y_j|f^r(i)_j)&\\
&\geq1-\eps_r.\label{eqn-mep-15}
\end{align}
Obviously, this implies
\begin{equation}\lim_{{n}\rightarrow\infty}\min_{s^{n}\in\bS^{n}}\max_{i\in M_{n}}\tr\{P_i^{n}f_{n}(i)\}=0.\end{equation}
Together with
\begin{equation}\lim_{{n}\rightarrow\infty}\frac{1}{{n}}\log |M_{n}|=\frac{1}{m'}C'>0\end{equation}
we have shown that $C_{det,max}(\fri)>0$ holds.
\end{proof}
Notice that the statements made in (\ref{eqn-mep-3}) and (\ref{eq:equiv-max-1}) equivalent and a glance at Definition \ref{def:c-symmetrizability} reveals that the assertion of (\ref{eqn-mep-3}) is nothing else than the symmetrizability restricted to sets of states consisting of two elements.

\subsection{\label{subsec:Entanglement-transmission-capacity-with-random-codes}Entanglement transmission capacity with random codes}
The final issue in this section is a sufficient condition for $\mathcal{A}_{\textup{random}}(\fri)=0$ which is based on the notion of qc-symmetrizability.\\
Let $\mathfrak{F}_{\cc}(\bS)$ stand for the set of $\cc$-valued functions defined on $\bS$ in what follows and we consider the set of channels with quantum input and classical output (qc-channels)\footnote{Mere positivity is sufficient here because $\mathfrak{F}_{\cc}(\bS)$ is commutative, cf. \cite{paulsen}.}
\begin{equation}\textrm{QC}(\hr,\bS):=\{T:\mathcal{B}(\hr)\to \mathfrak{F}_{\cc}(\bS): T \textrm{ is linear, positive, and trace preserving}\}.    \end{equation}
The condition that $T\in\textrm{QC}(\hr,\bS) $ is trace preserving means that
\begin{equation}\sum_{s\in \bS}[T(b)](s)=\textrm{tr}(b)  \end{equation}
holds for all $b\in\mathcal{B}(\hr)$. By Riesz' representation theorem there is a one-to-one correspondence between elements $T\in \textrm{QC}(\hr,\bS) $ and (discrete) positive operator-valued measures (POVM) $\{E_s  \}_{s\in \bS}$.\\
For a given finite set of quantum channels $\fri=\{\cn_s  \}_{s\in \bS}$ and $T\in \textrm{QC}(\hr,\bS) $ we define a CPTP map $\mathcal{M}_{T,\bS}:\mathcal{B}(\hr)\otimes \mathcal{B}(\hr)\to\mathcal{B}(\kr)$ by
\begin{eqnarray}\label{def-M}
\mathcal{M}_{T,\bS}(a\otimes b)&:=&\sum_{s\in \bS}[T(a)](s)\cn_s (b)\\
&=& \sum_{s\in \bS}\textrm{tr}(E_s a)\cn_s(b),
\end{eqnarray}
where $\{ E_s \}_{s\in \bS}$ is the unique POVM associated with $T$.
\begin{definition}\label{def:symmetrizability}
An arbitrarily varying quantum channel, generated by a finite set $\fri=\{\cn_s  \}_{s\in \bS}$, is called qc-symmetrizable if there is $T\in \textrm{QC}(\hr,\bS) $ such that for all $a,b\in\mathcal{B}(\hr) $
\begin{equation}\label{eq:symmetrizable}
  \mathcal{M}_{T,\bS}(a\otimes b)=\mathcal{M}_{T,\bS}(b\otimes a)
\end{equation}
holds, where $\mathcal{M}_{T,\bS}:\mathcal{B}(\hr)\otimes \mathcal{B}(\hr)\to\mathcal{B}(\kr) $ is the CPTP map defined in (\ref{def-M}).
\end{definition}
The best illustration of the definition of qc-symmetrizability is given in the proof of our next theorem.
\begin{theorem}\label{symm-implies-0-capacity}
If an arbitrarily varying quantum channel generated by a finite set $\fri=\{\cn_s  \}_{s\in \bS}$ is qc-symmetrizable, then for any sequence of $(l,k_l)$-random codes $(\mu_l)_{l\in\nn}$ with $k_l=\dim \fr_l\ge 2$ for all $l\in\nn$ we have
\begin{equation}\inf_{s^l\in \bS^l}\int F_e(\pi_{\fr_l}, \crr^l\circ \cn_{s^l}\circ \mathcal{P}^l )d\mu_l (\crr^l,\mathcal{P}^l)\le \frac{1}{2}, \end{equation}
for all $l\in\nn$. Thus
\begin{equation}\mathcal{A}_{\textup{random}}(\fri)=0,  \end{equation}
and consequently
\begin{equation}\mathcal{A}_{\textup{det}}(\fri)=0.  \end{equation}
\end{theorem}
\emph{Proof.} We have to show that for the codes $(\mathcal{P}^l,\crr^l)$ with the properties as stated in the lemma the inequality
\begin{equation}\label{eq:0-cap-1}
  \inf_{s^l\in \bS^l}\int F_e(\pi_{\fr_l}, \crr^l\circ \cn_{s^l}\circ \mathcal{P}^l )d\mu_l (\crr^l,\mathcal{P}^l)\le \frac{1}{2}
\end{equation}
holds for all $l\in\nn$.\\
Let $\psi_l\in\cs(\fr_l\otimes \fr_l) $ be a purification of $\pi_{\fr_l}$ which is, clearly, maximally entangled. Inequality (\ref{eq:0-cap-1}) can then be equivalently reformulated as
\begin{equation}\label{eq:0-cap-2}
  \inf_{s^l\in \bS^l}\int \langle \psi_l, (id_{\fr_l}\otimes (\crr^l\circ \cn_{s^l}\circ \mathcal{P}^l))(|\psi_l\rangle\langle \psi_l| )  \psi_l    \rangle d\mu_l (\crr^l,\mathcal{P}^l )\le \frac{1}{2}.
\end{equation}
We fix $\sigma\in \cs(\hr)$ and define CPTP maps $E_1,E_2: \mathcal{B}(\hr)\to\mathcal{B}(\kr)$ by
\begin{equation}\label{eq:0-cap-3}
  E_1(a):=\mathcal{M}_{T,\bS}(\sigma\otimes a)=\sum_{s\in \bS}\textrm{tr}(E_s \sigma)\cn_s(a)
\end{equation}
and
\begin{equation}\label{eq:0-cap-4}
  E_2(a):=\mathcal{M}_{T,\bS}(a\otimes \sigma)=\sum_{s\in \bS}\textrm{tr}(E_s a)\cn_s(\sigma).
\end{equation}
Then
\begin{eqnarray}\label{eq:0-cap-5}
  \int F_e(\pi_{\fr_l}, \crr^l\circ E_1^{\otimes l}\circ \mathcal{P}^l)d\mu_l (\crr^l,\mathcal{P}^l ) &=& \sum_{s^l\in \bS^l}\textrm{tr}(E_{s^l}\sigma^{\otimes l})\int F_e(\pi_{\fr_l}, \crr^l\circ\cn_{s^l}\circ \mathcal{P}^l )d\mu_l (\crr^l,\mathcal{P}^l )\\
&\ge& \inf_{s^l\in \bS^l}\int F_e(\pi_{\fr_l}, \crr^l\circ \cn_{s^l}\circ \mathcal{P}^l )d\mu_l (\crr^l,\mathcal{P}^l ),
\end{eqnarray}
where $E_{s^l}:= E_{s_1}\otimes \ldots \otimes E_{s_l}$. Therefore, we are done if we can show that
\begin{equation}\label{eq:0-cap-6}
 \int F_e(\pi_{\fr_l}, \crr^l\circ E_1^{\otimes l}\circ \mathcal{P}^l)d\mu_l (\crr^l,\mathcal{P}^l )\le \frac{1}{2}
\end{equation}
for all $l\in\nn$. \\
On the other hand, choosing bases $\{ e_{i,j} \}_{i,j=1}^{k_l}$ and $\{f_{k,m}  \}_{k,m=1}^{d^l}$ of $\mathcal{B}(\fr_l)$ and $\mathcal{B}(\hr)^{\otimes l}$ respectively, we can write
\begin{equation}id_{\fr_l}\otimes \mathcal{P}^l(|\psi_l\rangle\langle \psi_l| )=:\rho_l=\sum_{i,j,k,m}\rho_{i,j,k,m}e_{i,j}\otimes f_{k,m},  \end{equation}
and obtain
\begin{eqnarray}
 id_{\fr_l}\otimes (\crr^l\circ E_1^{\otimes l} )(\rho_l)&=&\sum_{i,j,k,m}\rho_{i,j,k,m}e_{i,j}\otimes \crr^l (\mathcal{M}_{T,\bS}^{\otimes l}(\sigma^{\otimes l}\otimes f_{k,m}  )  ) \\
&=& \sum_{i,j,k,m}\rho_{i,j,k,m}e_{i,j}\otimes \crr^l (\mathcal{M}_{T,\bS}^{\otimes l}( f_{k,m}\otimes \sigma^{\otimes l}  )  )\\
&=& \sum_{i,j,k,m}\rho_{i,j,k,m}e_{i,j}\otimes \crr^l(E_2^{\otimes l}(f_{k,m}))\\
&=& id_{\fr_l}\otimes (\crr^l\circ E_2^{\otimes l} )(\rho_l),
\end{eqnarray}
where the second equality follows from the assumed qc-symmetrizability. Thus, we end up with
\begin{equation}\label{eq:0-cap-7}
F_e(\pi_{\fr_l}, \crr^l\circ E_1^{\otimes l}\circ \mathcal{P}^l)=F_e(\pi_{\fr_l}, \crr^l\circ E_2^{\otimes l}\circ \mathcal{P}^l),
\end{equation}
for any encoding operation $\mathcal{P}^l$ and any recovery operation $\crr^l$.
Consequently, by (\ref{eq:0-cap-7}) and (\ref{eq:0-cap-5}) we have to show that for all $l\in\nn$
\begin{equation}\label{eq:0-cap-8}
 \int F_e(\pi_{\fr_l}, \crr^l\circ E_2^{\otimes l}\circ \mathcal{P}^l)d\mu_l (\crr^l,\mathcal{P}^l)\le \frac{1}{2}
\end{equation}
holds. But the channel
\begin{equation} E_2(a)=\sum_{s\in S}\textrm{tr}(E_s a)\cn_s(\sigma)\qquad (a\in\mathcal{B}(\hr)) \end{equation}
is entanglement breaking implying that the state
\begin{equation} (id_{\fr_l}\otimes \crr^l\circ E_2^{\otimes l}\circ \mathcal{P}^l)(|\psi_l\rangle\langle \psi_l|) \end{equation}
is separable. A standard result from entanglement theory implies that
\begin{equation}\label{eq:0-cap-9}
\langle \psi_l, (id_{\fr_l}\otimes \crr^l\circ E_2^{\otimes l}\circ \mathcal{P}^l)(|\psi_l\rangle\langle \psi_l|)        \psi_l\rangle \le \frac{1}{k_l}
\end{equation}
holds for any $\crr^l$ and $\mathcal{P}^l$ since $\psi_l$ is maximally entangled with Schmidt rank $k_l$. Now, our assumption that for each $l\in\nn$ the relation $k_l\ge 2$ holds implies along with (\ref{eq:0-cap-9}) that for all $l\in\nn$
\begin{equation}\int F_e(\pi_{\fr_l}, \crr^l\circ E_2^{\otimes l}\circ \mathcal{P}^l)d\mu_l(\crr^l,\mathcal{P}^l)=\int \langle \psi_l, (id_{\fr_l}\otimes \crr^l\circ E_2^{\otimes l}\circ \mathcal{P}^l)(|\psi_l\rangle\langle \psi_l|)        \psi_l\rangle d\mu_l(\crr^l,\mathcal{P}^l) \le \frac{1}{2},   \end{equation}
and by (\ref{eq:0-cap-7}) and (\ref{eq:0-cap-5}) we obtain
\begin{equation} \inf_{s^l\in \bS^l}\int F_e(\pi_{\fr_l}, \crr^l\circ \cn_{s^l}\circ \mathcal{P}^l )d\mu_l(\crr^l,\mathcal{P}^l)\le \frac{1}{2}  \end{equation}
which concludes the proof.
\begin{flushright}$\Box$\end{flushright}
Our Definition \ref{def:symmetrizability} addresses the notion of qc-symmetrizability for block length $l=1$. Thus the question arises whether a less restrictive requirement, as stated in the following definition, gives a better sufficient condition for an arbitrarily varying quantum channel to have capacity $0$.
\begin{definition}\label{def:l-symmetrizability}
An arbitrarily varying quantum channel, generated by a finite set $\fri=\{\cn_s  \}_{s\in \bS}$, is called $l$-qc-symmetrizable, $l\in\nn$, if there is $T\in \textrm{QC}(\hr^{\otimes l},\bS^l) $ such that for all $a,b\in\mathcal{B}(\hr)^{\otimes l} $
\begin{equation}\label{eq:l-symmetrizable}
  \mathcal{M}_{T,\bS}^l(a\otimes b)=\mathcal{M}_{T,\bS}^{l}(b\otimes a)
\end{equation}
holds, where $\mathcal{M}_{T,\bS}^l:\mathcal{B}(\hr)^{\otimes l}\otimes \mathcal{B}(\hr)^{\otimes l}\to\mathcal{B}(\kr)^{\otimes l} $ is the CPTP map defined by
\begin{equation}\label{def-M-l}
\mathcal{M}_{T,\bS}^{l}(a\otimes b):= \sum_{s^l\in \bS^{l}}\textrm{tr}(E_{s^l} a)\cn_{s^l}(b),
\end{equation}
and $\{ E_{s^l} \}_{s^l\in \bS^l}$ is the unique POVM corresponding to  $T\in \textrm{QC}(\hr^{\otimes l},\bS^l) $.
\end{definition}
Obviously, qc-symmetrizability implies $l$-qc-symmetrizability for all $l\in\nn$. The next lemma states that the reverse implication is true too.
\begin{lemma}\label{l-symm-implies-symm}
For any finitely generated AVQC given by $\fri=\{\cn_s  \}_{s\in \bS}$ $l$-qc-symmetrizability implies qc-symmetrizability for any $l\in \nn$.
\end{lemma}
\emph{Proof.} For a given finite set of quantum channels $\fri=\{\cn_s \}_{s\in \bS}$ and $l\in\nn$ let  $T\in \textrm{QC}(\hr^{\otimes l},\bS^l) $ be such that for all $a,b\in\mathcal{B}(\hr)^{\otimes l} $
\begin{equation}\label{l-imp-1-1}
 \mathcal{M}_{T,\bS}^l(a\otimes b)=\mathcal{M}_{T,\bS}^{l}(b\otimes a),
\end{equation}
where  $\mathcal{M}_{T,\bS}^l$ is defined in (\ref{def-M-l}).\\
Let $b\in\mathcal{B}(\hr)$ and for each $s\in \bS$ define a linear functional
\begin{equation}\phi_s(b):=\textrm{tr}\left(\left(b\otimes \left(\frac{1}{\dim \hr}\idn_{\hr}\right)^{\otimes l-1} \right)\sum_{s_2^{l}\in \bS^{l-1}}E_{ss_2^{l}} \right),  \end{equation}
where $ss_2^{l}:=(s,s_2,\ldots ,s_l)\in \bS^l$. Clearly, $\phi_s$ is positive. Consequently, Riesz' representation theorem shows that there is a unique positive $\tilde{E}_s\in\mathcal{B}(\hr)$ with
\begin{equation}\label{l-imp-1-2}
  \phi_s(b)=\textrm{tr}(\tilde{E}_s b)\qquad \qquad (b\in\mathcal{B}(\hr)).
\end{equation}
Obviously, $\{\tilde{E}_s  \}_{s\in \bS}$ is a POVM and let $\tilde{T}\in\textrm{QC}(\hr,\bS) $ denote the associated qc-channel.\\
Some simple algebra shows that for each $a,b\in \mathcal{B}(\hr)$
\begin{equation}\label{l-imp-1-3}
  \mathcal{M}_{\tilde{T},\bS}(a\otimes b)=\textrm{tr}_{\kr^{\otimes l-1}}\mathcal{M}_{T,\bS}^l\left(\left(a\otimes \left(\frac{1}{\dim \hr}\idn_{\hr}\right)^{\otimes l-1}\right)\otimes \left(b\otimes \left(\frac{1}{\dim \hr}\idn_{\hr}\right)^{\otimes l-1}     \right)   \right)
\end{equation}
where $\textrm{tr}_{\kr^{\otimes l-1}} $ denotes the partial trace over the last $l-1$ tensor factors. The relation (\ref{l-imp-1-3}) immediately implies that for all $a,b\in\mathcal{B}(\hr)$
\begin{equation}\label{l-imp-1-4}
\mathcal{M}_{\tilde{T},\bS}(a\otimes b)= \mathcal{M}_{\tilde{T},\bS}(b\otimes a),
\end{equation}
and
\begin{eqnarray}\label{l-imp-1-5}
\mathcal{M}_{\tilde{T},\bS}(a\otimes b)&=& \sum_{s^l\in \bS^l}\textrm{tr}\left(\left(b\otimes \left(\frac{1}{\dim \hr}\idn_{\hr}\right)^{\otimes l-1}      \right)  E_{s^l}  \right)\cn_{s_1}(a)\\
&=& \sum_{s_1\in \bS} \textrm{tr}\left(\left(b\otimes \left(\frac{1}{\dim \hr}\idn_{\hr}\right)^{\otimes l-1}      \right)\sum_{s_2^l\in \bS^{l-1}}  E_{s^l}  \right)\cn_{s_1}(a)\\
&=& \sum_{s_1\in \bS}\phi_{s_1}(b)\cn_{s_1}(a)\\
&=& \sum_{s_1\in \bS}\textrm{tr}(b \tilde{E}_{s_1})\cn_{s_1}(a).
\end{eqnarray}
Equations (\ref{l-imp-1-4}) and (\ref{l-imp-1-5}) show that $\fri$ is qc-symmetrizable.
\begin{flushright}$\Box$\end{flushright}
\section{\label{sec:Conditions for single-letter-capacities}Conditions for single-letter capacity formulas}
In this section we give two conditions on the structure of a finite AVQC which guarantee that their quantum capacity is given by a single-letter formula. 
The first one is empty in the case of a single channel, while the second one generalizes the degradability condition from \cite{devetak-shor} that we repeat here for readers convenience:\\
A channel $\cn\in\mathcal C(\hr,\kr)$ is called degradable if for any Hilbert space $\kr_E$ and any partial isometry $V:\hr\rightarrow\kr\otimes\kr_E$ such that $\cn(\cdot)=\tr_{\kr_E}(V\cdot V^*)$ there is 
$\cn_V\in\mathcal C(\kr,\kr_E)$ such that $\cn_V\circ\cn=\tr_\kr(V\cdot V^*)$. Degradability is a property independent from the choice of $\kr_E$ and $V$ - if it holds for only one such choice, then it holds for all possible choices.
\begin{lemma}\label{lemma-single-letter-conditions} Let $\fri=\{\cn_s\}_{s\in\bS}$ satisfy $\A_{\mathrm{det}}(\fri)=\A_{\mathrm{random}}(\fri)$ (for example, $\fri$ might be non-symmetrizable). We have a single-letter formula for $\A_{\textup{det}}(\fri)$ in any of the following two cases:
\begin{enumerate}
\item There is $\cn_*\in \conv(\fri)$ such that for any $\cn\in \conv(\fri)$ there is $\D_\cn\in\mathcal C(\kr,\kr)$ with the property $\cn_*=\D_\cn\circ\cn$ and, 
additionally, $Q(\cn_*)=\max_{\rho\in\cs(\hr)}I_c(\rho,\cn_*)$ holds for the entanglement transmission capacity $Q(\cn_*)$ of the memoryless channel $\cn_*$.
\item Each $\cn\in \conv(\fri)$ is degradable.
\end{enumerate}
\end{lemma}
\begin{proof}
\emph{1.} It is clear that
\begin{equation}
\A_{\textup{det}}(\fri)\leq Q(\cn_*)=\max_{\rho\in\cs(\hr)}I_c(\rho,\cn_*).\label{eqn:single-letter-conditions-1}
\end{equation}
By assumption, we have
\begin{equation}
\A_{\textup{det}}(\fri)=\A_{\textup{random}}(\fri)=\lim_{l\rightarrow\infty}\frac{1}{l}\max_{\rho\in\cs(\hr^{\otimes l})}\min_{\cn\in \conv(\fri)} I_c(\rho,\cn^{\otimes l}).\label{eqn:single-letter-conditions-2}
\end{equation}
On the other hand by application of the data-processing inequality \cite{schumacher-nielsen} we have, for all $\rho\in\cs(\hr^{\otimes l})$, $\cn\in \conv(\fri)$ and $l\in\nn$,
\begin{align}
I_c(\rho,\cn^{\otimes l})&\geq I_c(\rho,\D^{\otimes l}_\cn\circ\cn^{\otimes l})&\\
&=I_c(\rho,\cn_*^{\otimes l}).&
\end{align}
It follows that
\begin{align}
\frac{1}{l}\max_{\rho\in\cs(\hr^{\otimes l})}I_c(\rho,\cn^{\otimes l})&\geq \frac{1}{l}\max_{\rho\in\cs(\hr^{\otimes l})}I_c(\rho,\cn_*^{\otimes l})&
\end{align}
and by (\ref{eqn:single-letter-conditions-2}):
\begin{align}
\A_{\textup{det}}(\fri)&\geq\lim_{l\rightarrow\infty}\frac{1}{l}\max_{\rho\in\cs(\hr^{\otimes l})}I_c(\rho,\cn_*^{\otimes l})&\\
&=Q(\cn_*)&\\
&=\max_{\rho\in\cs(\hr)}I_c(\rho,\cn_*).\label{eqn:single-letter-conditions-3}
\end{align}
Equations (\ref{eqn:single-letter-conditions-1}) and (\ref{eqn:single-letter-conditions-3}) give us the desired result:
\begin{equation}
\A_{\textup{det}}(\fri)=\max_{\rho\in\cs(\hr)}I_c(\rho,\cn_*).
\end{equation}
\emph{2.} It is well known that the following three properties are valid:
\begin{itemize}
\item[P1] If a $\cn\in \mathcal C(\hr,\kr)$ is degradable, then the map $\rho\mapsto I_c(\rho,\cn)$ is concave (\cite{yard-devetak-hayden}, Lemma 5).
\item[P2] For every fixed $\rho\in\cs(\hr)$, $\cn\mapsto I_c(\rho,\cn)$ is convex (see \cite{lieb-ruskai}, Theorem 1).
\item[P3] Let $\cn\in\mathcal C(\hr,\kr)$ be degradable. For an arbitrary $l\in\nn$, write $\hr^{\otimes l}=\hr_1\otimes\ldots\otimes\hr_l$ with $\hr_i:=\hr$ for every $i\in\nn$. Let $\rho\in\cs(\hr^{\otimes l})$ with marginal states $\rho_i:=\tr_{\hr_1\otimes\ldots\otimes\hr_{i-1}\otimes\hr_{i+1}\otimes\ldots\otimes\hr_l}(\rho)$. Then the inequality $I_c(\rho,\cn^{\otimes l})\leq\sum_{i=1}^nI_c(\rho_i,\cn)$ holds \cite{devetak-shor}.
\item[P4] The coherent information is continuous in both of its entries.
\end{itemize}
By the minimax-theorem \cite{von-neumann,kakutani}, properties P1, P2 and P4 imply that
\begin{align}\max_{\rho\in\cs(\hr)}\min_{\cn\in \conv(\fri)}I_c(\rho,\cn)=\min_{\cn\in \conv(\fri)}\max_{\rho\in\cs(\hr)}I_c(\rho,\cn).
\end{align}
Suppose now, that each $\cn\in \conv(\fri)$ is degradable. It then holds, for every $l\in\nn$,
\begin{align}
\frac{1}{l}\max_{\rho\in\cs(\hr^{\otimes l})}\min_{\cn\in \conv(\fri)}I_c(\rho,\cn^{\otimes l})&\leq\frac{1}{l}\min_{\cn\in \conv(\fri)}\max_{\rho\in\cs(\hr^{\otimes l})}I_c(\rho,\cn^{\otimes l})&\\
&\leq \min_{\cn\in \conv(\fri)}\max_{\rho\in\cs(\hr)}I_c(\rho,\cn)&\\
&=\max_{\rho\in\cs(\hr)}\min_{\cn\in \conv(\fri)}I_c(\rho,\cn),&
\end{align}
where the second inequality follows from P3 and the equality from P1, P2 via the minimax-theorem. It follows that
\begin{align}
\A_{\textup{det}}(\fri)\leq\A_{\textup{random}}(\fri)\leq\max_{\rho\in\cs(\hr)}\min_{\cn\in \conv(\fri)}I_c(\rho,\cn).
\end{align}
By assumption, we also have $\A_{\textup{det}}(\fri)=\A_{\textup{random}}(\fri)$. The obvious relation $\A_{\textup{random}}(\fri)\geq\max_{\rho\in\cs(\hr)}\min_{\cn\in \conv(\fri)}I_c(\rho,\cn)$ then implies the reverse inequality.
\end{proof}
\section{\label{sec:Applications and examples}An example and an application to zero-error capacities}
\subsection{\label{subsec:Erasure-AVQC}Erasure-AVQC}
As an application and illustration of most of the results obtained so far we calculate the quantum capacity of finite AVQC $\fri$ consisting of erasure quantum channels. As expected, we obtain that $\A_{\textup{det}}(\fri)$ equals the capacity of the worst erasure channel in the set $\fri$.
\begin{lemma}\label{lemma-erasure-example}
Let $d\in\nn$, $d\geq2$ and denote by $\{e_1,\ldots e_d\},\ \{e_1,\ldots,e_{d+1}\}$, the standard basis of $\mathbb C^d,\mathbb C^{d+1}$. Set $\hr=\mathbb C^d$, $\kr=\mathbb C^{d+1}$. Define, for $p\in[0,1]$, the erasure channel $\E_p\in\mathcal C(\hr,\kr)$ by
\begin{equation}\E_p(x):=(1-p)x+p\cdot\tr(x)|e_{d+1}\rangle\langle e_{d+1}|\qquad\forall x\in\mathcal B(\hr).\end{equation}
Let, for a finite collection $\{p_s\}_{s\in\bS}\subset[0,1]$, an AVQC be given by $\fri=\{\E_{p_s}\}_{s\in\bS}$.\\
The following are true.
\begin{enumerate}
\item If $p_s\geq1/2$ for some $s\in\bS$, then $\A_{\textup{det}}(\fri)=\A_{\textup{random}}(\fri)=0$.
\item If $p_s<1/2$ for every $s\in\bS$, then $\A_{\textup{det}}(\fri)=\A_{\textup{random}}(\fri)=\min_{s\in\bS}(1-2p_s)\log(d)$.
\end{enumerate}
\end{lemma}
\begin{proof}
We start with \emph{2.} by showing the validity of the following properties.
\begin{description}
\item[A.] For $q\in\mathfrak P(\bS)$, we have $\sum_{s\in\bS}q(s)\E_{p_s}=\E_{q(p)}$, where $q(p):=\sum_{s\in\bS}q(s)p_s$.
\item[B.] There is a set $\{\widehat{\E_{p_s}}\}_{s\in\bS}$ of complementary maps given by $\widehat{\E_{p_s}}=\E_{1-p_s}$, $s\in\bS$.
\item[C.] $\E_p$ is degradable for $p\in[0,1/2)$.
\item[D.] $\{\E_{p_s}\}_{s\in\bS}$ is non-symmetrizable if $p_s\in[0,1)$ for all $s\in\bS$. Additionally then, $C_{\textup{det,max}}(\{\E_{p_s}\}_{s\in\bS})>0$ holds.
\end{description}
\textbf{A.:} For every $x\in\mathcal B(\mathbb C^d)$,
\begin{align}
\sum_{s\in\bS}q(s)\E_{p_s}(x)&=\sum_{s\in\bS}q(s)[(1-p_s)x+p_s\cdot\tr(x)|e_{d+1}\rangle\langle e_{d+1}|]&\\
&=(1-\sum_{s\in\bS}q(s)p_s)x+\sum_{s\in\bS}q(s)p_s\cdot\tr(x)|e_{d+1}\rangle\langle e_{d+1}|.&
\end{align}
\textbf{B.:} Consider an environment defined by $\kr_{env}:=\kr$. For every $p\in[0,1]$ we can give a Stinespring isometry $V_p:\hr\rightarrow\kr\otimes\kr_{env}$ of $\E_p$ by
\begin{align}
V_pu:=\sqrt{1-p}\cdot u\otimes e_{d+1}+\sqrt{p}\cdot e_{d+1}\otimes u,\qquad u\in\mathcal B(\hr).
\end{align}
hbecomes clear by tracing out the first or second subsystem, depending on whether one wants to calculate $\hat\E_p$ or $\E_p$.\\
\textbf{C.:} Set $\mu:=\frac{1-2p}{1-p}$ and define $E_\mu\in\mathcal C(\kr,\kr)$ by
\begin{equation}E_\mu(x):=(1-\mu)\cdot x +\mu\cdot\tr(x)\cdot|e_{d+1}\rangle\langle e_{d+1}|,\qquad x\in\mathcal B(\kr).\end{equation}
Then by $p\in[0,1/2)$ we have $\mu\in(0,1]$. We show that $\E_{1-p}=E_\mu\circ\E_{p}$ holds. Let $x\in\mathcal B(\hr)$, then
\begin{align}
E_{\mu}\circ\E_p(x)&=(1-p)\cdot E_\mu(x)+p\cdot\tr(x)\cdot E_\mu(|e_{d+1}\rangle\langle e_{d+1}|)&\\
&=(1-p)\cdot(1-\mu)\cdot x+\mu\cdot(1-p)\cdot\tr(x)\cdot|e_{d+1}\rangle\langle e_{d+1}|+p\cdot\tr(x)\cdot|e_{d+1}\rangle\langle e_{d+1}|&\\
&=(1-p-1+2p)\cdot x+(1-2p)\cdot\tr(x)\cdot|e_{d+1}\rangle\langle e_{d+1}|+p\cdot\tr(x)\cdot|e_{d+1}\rangle\langle e_{d+1}|&\\
&=p\cdot x+(1-p)\cdot\tr(x)\cdot|e_{d+1}\rangle\langle e_{d+1}|&\\
&=\E_{1-p}(x).
\end{align}
\textbf{D.:} Let $\rho_1:=|e_1\rangle\langle e_1|,\ \rho_2:=|e_2\rangle\langle e_2|\in\cs(\hr)$. We show by contradiction that there are no two probability distributions $r_1,r_2\in\mathfrak P(\bS)$ such that
\begin{align}
\sum_{s\in\bS}r_1(s)\E_{p_s}(\rho_1)=\sum_{s\in\bS}r_2(s)\E_{p_s}(\rho_2).\label{eqn-erasure-avqc-1}
\end{align}
Assume there are $r_1,r_2\in\mathfrak P(\bS)$ such that (\ref{eqn-erasure-avqc-1}) is true. This is equivalent to
\begin{align}
\sum_{s\in\bS}(1-p_s)[r_1(s)|e_1\rangle\langle e_1|-r_2(s)|e_2\rangle\langle e_2|]=0,\qquad \sum_{s\in\bS}p_s[r_1(s)-r_2(s)]\cdot|e_{d+1}\rangle\langle e_{d+1}|=0.
\end{align}
By linear independence of $|e_1\rangle\langle e_1|,|e_2\rangle\langle e_2|$ and since $p_s\in[0,1)$ for every $s\in\bS$ the first equality implies $r_1(s)=r_2(s)=0\ \forall s\in\bS$, 
in clear contradiction to the assumption $r_1,r_2\in\mathfrak P(\bS)$.\\
Thus, $\{\E_{p_s}\}_{s\in\bS}$ with all $p_s\in[0,1)$ is non-symmetrizable. Moreover, checking the requirements in Theorem \ref{theorem:c-det=0-for-maximal-error} we see that also
$C_{\textup{det,max}}(\{\E_{p_s}\}_{s\in\bS})>0$ has to hold.
\\\\
Using \textbf{A} and the fact that $\{p_s\}_{s\in\bS}\subset[0,1/2)$ we see that for an arbitrary $q\in\mathfrak P(\bS)$ we have $\sum_{s\in\bS}q(s)\E_{p_s}=\E_{q(p)}$ with $q(p)\in[0,1/2)$.\\
Now \textbf{B} implies that for every $q\in\mathfrak P(\bS)$ the channel $\sum_{s\in\bS}q(s)\E_{p_s}$ is degradable.\\
Thus by Lemma \ref{lemma-single-letter-conditions}, \emph{2.}, the regularization in the identity
\begin{equation}\A_{\textup{random}}(\fri)=\lim_{l\to\infty}\frac{1}{l}\max_{\rho\in\cs(\hr^{\otimes l})}\inf_{\cn\in \conv(\fri)}I_c(\rho, \cn^{\otimes l})\end{equation}
is not necessary, so
\begin{align}\A_{\textup{random}}(\fri)=\max_{\rho\in\cs(\hr)}\inf_{\cn\in \conv(\fri)}I_c(\rho, \cn).\label{eqn-erasure-avqc-2}\end{align}
Further, for a fixed degradable channel, the coherent information is concave in the input state \cite{yard-devetak-hayden} and thus by the minimax theorem for concave-convex functions \cite{kakutani,von-neumann} we can interchange min and max in (\ref{eqn-erasure-avqc-2}). Now, to any given $\rho\in\cs(\hr)$, we may write $\rho=\sum_{i=1}^d\lambda_i|v_i\rangle\langle v_i|$ for some set $\{v_1,\ldots,v_d\}$ of orthonormal vectors that satisfy, by standard identification of $\mathbb C^d$ and $\mathbb C^{d+1}$, $v_i\perp e_{d+1}$ ($1\leq i\leq d$) and write a purification of $\rho$ as $|\psi_\rho\rangle\langle\psi_\rho|=\sum_{i,j=1}^d\lambda_i\lambda_j|v_i\rangle\langle v_j|\otimes|v_i\rangle\langle v_j|$. Then for every $\E_p\in \conv(\fri)$ we have
\begin{align}
\max_{\rho\in\cs(\hr)}I_c(\rho, \E_p)&=\max_{\rho\in\cs(\hr)}(S(\E_p(\rho))-S(Id_{\hr}\otimes\E_p(|\psi_\rho\rangle\langle\psi_\rho|))&\\
&=\max_{\rho\in\cs(\hr)}(S((1-p)\rho+p|e_{d+1}\rangle\langle e_{d+1}|)-S((1-p)|\psi_\rho\rangle\langle\psi_\rho|+p\rho\otimes|e_{d+1}\rangle\langle e_{d+1}|))&\\
&=\max_{\rho\in\cs(\hr)}((1-p)S(\rho)+pS(|e_{d+1}\rangle\langle e_{d+1}|)+H(p)&\\
&\ \ -(1-p)S(|\psi_\rho\rangle\langle\psi_\rho|)-pS(\rho\otimes|e_{d+1}\rangle\langle e_{d+1}|)-H(p))&\\
&=\max_{\rho\in\cs(\hr)}((1-p)S(\rho)-pS(\rho\otimes|e_{d+1}\rangle\langle e_{d+1}|))&\\
&=\max_{\rho\in\cs(\hr)}(1-2p)S(\rho)&\\
&=(1-2p)\log(d).\label{eqn-erasure-avqc-3}&
\end{align}
This leads to
\begin{align}\A_{\textup{random}}(\fri)=\min_{s\in\bS}(1-2p_s)\log(d),\end{align}
a formula that was first discovered for the case of a single memoryless channel and $d=2$ by \cite{bennet-divincenzo-smolin}.\\
From \textbf{D} it follows that $\A_{\textup{det}}(\fri)=\A_{\textup{random}}(\fri)$.\\
We can now prove \emph{1.}: Set $p_{\max}:=\max_{s\in\bS}p_s$. It holds $p_{\max}\geq1/2$, therefore the channel $\E_{\max}:=\E_{p_{\max}}$ satisfies 
$\A_{\textup{det}}(\{\E_{\max}\})=\A_{\textup{random}}(\{\E_{\max}\})=0$, since by $\mathbf B$ and $\mathbf C$ $\widehat{\E_{\max}}$ is degradable, 
hence $\E_{\max}$ is anti-degradable.\
Thus, for every $l\in\nn$, the adversary can always choose $\E_{\max}^{\otimes l}$ to ensure that transmission of entanglement will fail.
\end{proof}
\subsection{\label{subsec:Qualitative behavior of zero-error capacities}Qualitative behavior of zero-error capacities}
Let us, first, embark on the connection between AVQCs and zero-error capacities. Classical information theory exhibits an interesting connection between the zero-error capacity of certain channels and the deterministic capacity with asymptotically vanishing maximal error probability criterion. This connection is described in \cite{ahlswede-note}.\\
We give (following closely the lines of \cite{ahlswede-note}) the remaining part of this connection in the quantum case:\\
Let $\fri=\{\cn_s\}_{s\in\bS}$ be a finite AVQC. Consider
\begin{equation}
 \cn_\fri:=\frac{1}{|\bS|}\sum_{s\in\bS}\cn_s.
\end{equation}
By definition of zero-error capacity, to any $\delta>0$ there exists an $l\in\nn$, a maximally mixed state $\pi_{\fr_l}$ with 
$\frac{1}{l}\log\dim\fr_l\geq Q_0(\cn_\fri)-\delta$ and a pair $(\crr^l,\cP^l)$ of recovery and encoding map such that
\begin{equation}
 \min_{x\in\fr_l,||x||=1}\langle x,\crr^l\circ\cn_\fri^{\otimes l}\circ\cP(|x\rangle\langle x|)x\rangle=1
\end{equation}
holds. But this directly implies
\begin{equation}
 \min_{x\in\fr_l,||x||=1}\langle x,\crr^l\circ\cn_{s^l}\circ\cP(|x\rangle\langle x|)x\rangle=1\qquad\forall s^l\in\bS^l,
\end{equation}
so $(\pi_{\fr_l},\crr,\cP^l)$ is a zero-error code for the AVQC $\fri$ as well and therefore
\begin{equation}
 \A_{\det}(\fri)\geq Q_0(\cn_\fri)-\delta\qquad\forall\delta>0.\label{eqn:avqc-connection-to-zero-error}
\end{equation}
This in turn is equivalent to
\begin{equation}
\A_{\det}(\fri)\geq Q_0(\cn_\fri).
\end{equation}
One may now ask when exactly this is a meaningful (nonzero) lower bound. The answer is given by the proof of Lemma \ref{face-reduction}: 
On any face of $\mathcal C(\hr,\kr)$ the zero-error capacities are constant and the encoding and recovery maps are \emph{universal}. 
Thus, if $\fri$ is a subset of a face and $\fri\subset \ri(\mathcal C(\hr,\kr))^\complement$ then there is good hope to get a nonzero lower bound by means of inequality (\ref{eqn:avqc-connection-to-zero-error}). So far for the connection between AVQCs and zero-error capacities.\\
Motivated by the above observation, a closer study of zero-error capacities reveals some additional facts that are interesting in their own right.\\
To be more precise, we investigate continuity of zero-error capacities. This property is a highly desirable property both from the practical and the theoretical point of view. It is of particular importance in situations where full knowledge of the communication system cannot be achieved but only a narrow confidence set containing the unknown channel is given. In \cite{leung-smith} it has been shown that the ordinary capacities of stationary memoryless quantum channels are continuous in the finite-dimensional setting and it was demonstrated by examples that these functions become discontinuous in infinite dimensional situations.\\
In this subsection we show that quantum, entanglement-assisted, and classical zero-error capacities of quantum channels are discontinuous at every positivity point. Our approach is based on two simple observations. The first one is that the zero-error capacities mentioned above of each quantum channel belonging to the relative interior of the set of quantum channels are equal to $0$. The second one is the well known fact that the relative interior of any convex set is open and dense in that set, i.e. generic. Hence any channel can be approximated by a sequence belonging to the relative interior implying the discontinuity result.\\
Similar arguments can be applied to the recently introduced Lov\'asz $\tilde\theta$ function and zero-error distillable entanglement as well, leading to analogous conclusions as shall be shown in the last part of this subsection.     
We now show that all the zero-error capacities defined in subsection \ref{subsec:Zero-error capacities} are generically equal to $0$ and are discontinuous at any positivity point. Then we demonstrate that the zero-error capacities of quantum channels can be thought of as step functions subordinate to the partition built from the relative interiors of the faces of $\mathcal{C}(\hr,\kr)$.
\paragraph{Discontinuity of zero-error capacities}\label{discontinuity}
\begin{theorem}\label{ea-zero-in-interior} 
Let $\cn \in \ri \mathcal{C}(\hr,\kr)$. Then $k(l,\cn)=M(l,\cn)= M_{\textup{EA}}(l,\cn)=1$ for every $l\in\nn$. Consequently, $Q_0(\cn)=C_0(\cn)=C_{0\textup{EA}}(\cn)=0$.
\end{theorem}
In the proof of Theorem \ref{ea-zero-in-interior} we shall make use of the following elementary fact: 
\begin{lemma}\label{elementary}
Let $F$ be a non-empty convex set and $\cn_0,\cn\in\ri F$ with $\cn_0\neq \cn$. Then there exists $\cn_1\in F$ and $\lambda_0,\lambda_1\in (0,1)$, $\lambda_0+\lambda_1=1$ with $\cn=\lambda_0\cn_0+\lambda_1\cn_1$.
\end{lemma}
\begin{proof}[Proof of Lemma \ref{elementary}]
Since $\cn\in \ri F$ there is $\mu'>1$ such that
\begin{equation}\label{app-2}
  \cn_1:= (1-\mu')\cn_0+\mu'\cn\in F.
\end{equation}
We define now 
\begin{equation}\label{app-3}
  \lambda_1:=\frac{1}{\mu'}\in (0,1), \quad \lambda_0:=1-\lambda_1,
\end{equation}
and obtain using $\cn_1$ given in (\ref{app-2}) the desired convex decomposition
\begin{equation}\label{app-4}
  \cn=\lambda_0\cn_0 +\lambda_1\cn_1.
\end{equation}
\end{proof}

\begin{proof}[Proof of Theorem \ref{ea-zero-in-interior}] Let $\cn \in \ri \mathcal{C}(\hr,\kr) $. Observing that the fully depolarizing channel $\cn_{0}(a)=\frac{\textrm{tr}(a)}{d_{\kr}}\idn_{\kr}$, $a\in \B(\hr)$, belongs to $\ri \mathcal{C}(\hr,\kr)$ we obtain from Lemma \ref{elementary} a convex decomposition of $\cn$ as
\begin{equation}\label{ea-N-convex}
  \cn=\lambda_0 \cn_0 +\lambda_1 \cn_1,
\end{equation}
where $\lambda_0,\lambda_1\in(0,1)$, $\lambda_0+\lambda_1=1$.\\
Clearly, this decomposition implies that
\begin{equation}\label{ea-N-tensored}
\cn^{\otimes l}=\sum_{s^l\in \{0,1  \}^l}\lambda_{s^l}\cn_{s^l}, 
\end{equation}
with $\lambda_{s^l}:=\lambda_{s_1}\cdot\ldots\cdot\lambda_{s_l}>0$ and $\cn_{s^l}:=\cn_{s_1}\otimes\ldots\otimes \cn_{s_l}$ for all $s^l\in\{0,1  \}^{ l}$. Then for any zero-error $(l,M)$ ea-code $(\sigma_{\fr\fr'},\{\cP_m, D_m  \}_{m=1}^M )$ for $\cn$ we get for each $m\in [M]$
\begin{eqnarray}
1&=&\tr ((\cn^{\otimes l}\circ \cP_m \otimes     \textrm{id}_{\fr'})(\sigma_{\fr\fr'})D_m )\\
  &=& \sum_{s^l\in\{0,1  \}^l}\lambda_{s^l} \tr ((\cn_{s^l}\circ \cP_m \otimes     \textrm{id}_{\fr'})(\sigma_{\fr\fr'})D_m )
\end{eqnarray}
and, consequently, since $\lambda_{s^l}>0$ for all $s^l\in\{0,1  \}^{ l}$
\begin{equation}\label{ea-avc}
  \tr ((\cn_{s^l}\circ \cP_m \otimes     \textrm{id}_{\fr'})(\sigma_{\fr\fr'})D_m )=1\qquad \forall s^l\in\{0,1  \}^l,
\end{equation}
for all $m\in [M]$. Choosing $\bar s^l=(0,\ldots,0)$ we obtain from Eqn. (\ref{ea-avc}) that $(\sigma_{\fr\fr'},\{\cP_m, D_m  \}_{m=1}^M )$ is a zero-error ea-code for $\cn_0$. Since $M_{\textrm{EA}}(l,\cn_0)=1$ for all $l\in\nn$ we can conclude that $M_{\textrm{EA}}(l,\cn)\le1$ and thus $M_{\textrm{EA}}(l,\cn)=1$ holds. Consequently $C_{0\textrm{EA}}(\cn)=0$. The other assertions follow from the observation that $1\le k(l,\cn)\le M(l,\cn)\le M_{\textrm{EA}}(l,\cn)$. 
\end{proof}
\begin{corollary}\label{q-corollary}
 The function $Q_0:\mathcal{C}(\hr,\kr)\to \rr_{+}$ that assigns the zero-error quantum capacity to each quantum channel is discontinuous at any $\cn\in\mathcal{C}(\hr,\kr)$ with $Q_0(\cn)>0$. The same conclusion holds true for $C_0$ and $C_{0\textup{EA}}$.
\end{corollary}
\begin{proof} If $Q_0(\cn)>0$ holds then necessarily $\cn \in \rebd \mathcal{C}(\hr,\kr)$ by Theorem \ref{ea-zero-in-interior}. On the other hand $\ri \mathcal{C}(\hr,\kr)$ is dense in $\mathcal{C}(\hr,\kr)$ (cf. Theorem 2.3.8 in \cite{webster}). So there is a sequence of channels $(\cn_i)_{i\in\nn}\subset \ri \mathcal{C}(\hr,\kr) $ with $\lim_{i\to\infty}||\cn_i-\cn||_{\lozenge}=0$ and by Theorem \ref{ea-zero-in-interior} we have $Q_0(\cn_i)=0$ for all $i\in \nn$. The arguments for $C_0$ and $C_{0\textrm{EA}}$ follow the same line of reasoning.
\end{proof}
\paragraph{Relation to the facial structure of the set of quantum channels}\label{faces}
Here we shall show that the considered zero-error capacities are basically step functions, 
the underlying partition consisting of the relative interiors of the faces of $\mathcal{C}(\hr,\kr)$.

\begin{lemma}\label{face-reduction}
Let $F\subset \mathcal{C}(\hr,\kr)$ be convex and let $\tilde\cn\in\ri F $. Then for any $\cn\in \ri F$, $Q_0(\cn)=Q_0(\tilde\cn)$, $C_{0\textup{EA}}(\cn)=C_{0\textup{EA}}(\tilde\cn)$, and $C_0(\cn)=C_0(\tilde\cn)$ hold.
 \end{lemma}
\begin{proof} We assume w.l.o.g. that $\cn\neq \tilde\cn$ to avoid trivialities. Then setting $\cn_0:=\cn$ we can find $\cn_1\in F$ and $\lambda_0,\lambda_1\in (0,1), \lambda_0+\lambda_1=1$, with
\begin{equation}\label{face-reduction-1}
  \tilde\cn=\lambda_0\cn_0 +\lambda_1 \cn_1
\end{equation}
just by applying Lemma \ref{elementary} to $\cn_0,\tilde\cn\in \ri F$.\\
Let $(\fr_l,\cP,\crr)$ be an $(l,k_l)$ zero-error quantum code for $\tilde\cn $.
Then using the representation (\ref{face-reduction-1}) we obtain for any $x\in \fr_l, ||x||=1$
\begin{eqnarray}\label{face-reduction-2}
  1&=& \langle x, \crr\circ \tilde\cn^{\otimes l}\circ \cP (|x\rangle\langle x|  )x\rangle\\
&=& \sum_{s^l\in\{0,1  \}^l}\lambda_{s^l} \langle x, \crr\circ \cn_{s^l}\circ \cP (|x\rangle\langle x|  )x\rangle
\end{eqnarray}
and consequently, since $\lambda_{s^l}>0$ for all $s^l\in \{0,1  \}^l$, we are led to
\begin{equation}\label{face-reduction-3}
  \langle x, \crr\circ \cn_{s^l}\circ \cP (|x\rangle\langle x|  )x\rangle=1
\end{equation}
for all $s^l\in\{0,1  \}^l$ and all $x\in\fr_l, ||x||=1$. Choosing the sequence $s^l=(0,\ldots,0)$ and recalling that $\cn_0=\cn$ we arrive at
\begin{equation}\label{face-reduction-4}
  Q_0(\cn)\ge Q_0(\tilde\cn).
\end{equation}
The reverse inequality is derived by interchanging the roles of $\cn$ and $\tilde\cn$. The remaining assertions are shown in the same vein.
\end{proof}
We shall now pass to the set of faces $\mathfrak{F}:=\{F: \textrm{face of } \mathcal{C}(\hr,\kr) \}$ of $\mathcal{C}(\hr,\kr)$.
\begin{theorem}\label{face-theorem}
To each $\cn\in \mathcal{C}(\hr,\kr)$ there is a unique $F\in\mathfrak{F}$ with $\cn\in \ri F$. Moreover, each of the capacity functions $Q_0$, $C_{0\textup{EA}}$, and $C_{0}$ is constant on $\ri F$.
\end{theorem}
\begin{proof} According to Theorem 2.6.10 in \cite{webster} the family of sets $\{\ri F: F\in\mathfrak{F}  \}$ forms a partition of $\mathcal{C}(\hr,\kr)$. This shows the first assertion of the theorem. The second follows from Lemma \ref{face-reduction}.
\end{proof}

\begin{remark}
Notice that the results obtained so far show that the optimal (i.e. capacity achieving) code for any channel $\cn$ in the relative interior of any face  $F$ of 
$\mathcal{C}(\hr,\kr)$ is also optimal for any other channel in $\textrm{ri }F$.\end{remark}

\subsection{\label{subsec:lovasz}Discontinuity of quantum Lov\'asz $\tilde\theta$ function \& zero-error distillable entanglement}

In this final section we show that our methods are not only bound to the zero-error capacities of quantum channels. They apply to the quantum Lov\'asz $\tilde\theta$ function from 
\cite{duan-severini-winter} and also to zero-error distillable entanglement.
\paragraph{Discontinuity of quantum Lov\'asz $\tilde\theta$ function}\label{theta}
Preliminarily, following \cite{duan-severini-winter}, for a given channel $\cn\in\mathcal{C}(\hr,\kr)$ with a corresponding set of Kraus operators $\{ E_j \}_{j\in[K]}$ 
we define the non-commutative confusability graph following \cite{duan-severini-winter} by
\begin{eqnarray}\label{representation-of-S-N}
S(\cn)&:=&\textrm{span}\{E_j^{\ast}E_i: i,j\in [K]  \}\\
      &=& \hat\cn_{\ast}(\bo (\E)),  
\end{eqnarray}
where $\hat\cn_{\ast}$ is the adjoint of the complementary channel $\hat\cn\in\mathcal{C}(\hr,\E)$ defined via the Stinespring isometry $V:\hr\to \kr\otimes \E$
\begin{equation} Vx:=\sum_{j=1}^{K}E_j x\otimes f_j \end{equation}
with an orthonormal basis $\{f_1,\ldots, f_{K}  \}$ in $\E$.\\
Also, let us recall their definition of the quantum Lov\'asz $\tilde\theta$ function and its most fundamental property:
\begin{definition}[Quantum Lov\'asz $\tilde\theta$ function]
The quantum Lov\'asz $\tilde\theta$ function is, for a given confusability graph $S$ defined by 
\begin{equation}
 \tilde\theta(S):=\sup_{n\in\nn}\max\{\|\mathbf1_{\hr\otimes\mathbb C^n}+T\|:T\in S^\perp\otimes\mathcal B(\mathbb C^n),\ \mathbf1_{\hr\otimes\mathbb C^n}+T\geq0,\ T=T^*\},
\end{equation}
where $S^\perp:=\{a\in\mathcal B(\hr):\tr(ab)=0\ \forall b\in S\}$.
\end{definition}
This function gives an upper bound stemming from semi-definite programming (\cite{duan-severini-winter}, Theorem 8) on the entanglement-assisted capacity for transmission of classical messages with zero error: 
For every channel $\cn\in\mathcal C(\hr,\kr)$:
\begin{equation}
 C_{0EA}(\cn)\leq\log\tilde\theta(S(\cn))\label{eqn-zero-error-1}
\end{equation}
(Lemma 7 and Corollary 10 in \cite{duan-severini-winter}). We are going to employ the dual formulation:
\begin{theorem}[Theorem 9 in \cite{duan-severini-winter}]\label{thm:lovasz-theta-as-semidefinite-program}
For any $\cn\in\mathcal{C}(\hr,\kr)$ we have
\begin{equation}\label{dual-theta}
  \tilde\theta (S(\cn))=\min\left\{||\tr_{\hr} Y   ||: Y\in S(\cn)\otimes \bo(\hr'), Y\ge |\Phi\rangle\langle \Phi|  \right\},
\end{equation}
where $\hr'$ is just a copy of $\hr$ and $\Phi=\sum_{i=1}^{\dim \hr}e_i\otimes e'_i$ with ONBs $\{e_1,\ldots,e_{\dim\hr}  \}$ and 
$\{e'_1,\ldots,e'_{\dim\hr}  \}$ of $\hr$ and $\hr'$.
\end{theorem}
In the following we shall also need the next simple lemma.
\begin{lemma}\label{full-graph}
Let $\cn\in \ri \mathcal{C}(\hr,\kr)$. Then $S(\cn)=\bo (\hr)$.
\end{lemma}
\begin{proof} Again we can represent $\cn$ as
\begin{equation}\cn=\lambda_0 \cn_0 +\lambda_1 \cn_1   \end{equation}
with $\lambda_0,\lambda_1\in (0,1)$, $\lambda_0+\lambda_1=1$, $\cn_0$ being the fully depolarizing channel, and
$\cn_1\in\mathcal{C}(\hr,\kr)$. The proof is concluded by the following simple observation: Given any two channels $\cn_0,\cn_1$ and $\lambda_0,\lambda_1\in (0,1)$ with $\lambda_0+\lambda_1=1$. Then for the channel $\cn:=\lambda_0\cn_{0}+\lambda_1\cn_1$ it holds that
\begin{equation} S(\cn)\supseteq S(\cn_0), S(\cn_1). \end{equation}
Since in our case $S(\cn_0)=\bo (\hr)$ we are done.
\end{proof}
With Theorem \ref{thm:lovasz-theta-as-semidefinite-program} and Lemma \ref{full-graph} at our disposal we can deduce the following discontinuity result for $\tilde\theta$:
\begin{theorem}
The function $\tilde\theta: \mathcal{C}(\hr,\hr)\to \rr_{+}$ assigning the number $\tilde\theta(S(\cn))$ to each quantum channel $\cn$ is discontinuous at any $\cn$ with $C_{0\textup{EA}}(\cn)>0$.
\end{theorem}
\begin{proof} Note that for $\cn\in\ri \mathcal{C}(\hr,\kr)$ $S(\cn)=\bo(\hr)$ by Lemma \ref{full-graph}. Hence $|\Phi\rangle\langle \Phi|\in S(\cn)\otimes \bo(\hr')=\bo(\hr)\otimes \bo(\hr)$ and $|| \tr_{\hr}|\Phi\rangle\langle \Phi|   ||=1=\tilde\theta(S(\cn))$.
On the other hand, (\ref{eqn-zero-error-1}) implies that for any $\cn\in\mathcal{C}(\hr,\kr)$, $\tilde\theta (S(\cn))>1$ if $C_{0\textrm{EA}}(\cn)>0 $.\\
Since $ \ri \mathcal{C}(\hr,\kr) $ is dense in $\mathcal{C}(\hr,\kr)$ and since $\tilde\theta(S(\cn))=1$ for each $\cn\in \ri \mathcal{C}(\hr,\kr)$ we are done.
\end{proof}
Notice that the arguments given for the Lov\'asz $\tilde\theta$ function apply to any other upper bound to the entanglement-assisted zero-error capacity vanishing in the relative interior of $\mathcal{C}(\hr,\kr)$.
\paragraph{Zero-error distillation of entanglement}
The simple methods employed so far can also be applied to the problem of zero-error distillation of entanglement as we shall briefly indicate below. 
Assuming that $\rho \in \ri \cs (\hr_A\otimes \hr_B)$ we can find $\lambda_0,\lambda_1\in (0,1)$, $\lambda_0+\lambda_1=1$ such that
\begin{equation}\label{distillation-4}
  \rho=\lambda_0 \rho_0 +\lambda_1\rho_1,
\end{equation}
with $\rho_0=\frac{1}{d_A}\idn_{\hr_A}\otimes \frac{1}{d_B}\idn_{\hr_B}\in\ri \cs (\hr_A\otimes \hr_B)$, $d_A=\dim \hr_A, d_B=\dim\hr_B$, and $\rho_1\in \cs (\hr_A\otimes \hr_B)$. Then
\begin{equation}\label{distillation-5}
  \rho^{\otimes l}=\sum_{s^l\in\{ 0,1 \}^l}\lambda_{s^l}\rho_{s^l},
\end{equation}
and for any $(l,k_l)$ zero-error EDP $(\D, \vphi_{k_l})$ for $\rho$ we obtain
\begin{equation}\label{distillation-6}
  1= \sum_{s^l\in\{ 0,1 \}^l}\lambda_{s^l}\langle  \vphi_{k_l}, \D(\rho_{s^l})  \vphi_{k_l}     \rangle ,
\end{equation}
leading to
\begin{equation}\label{distillation-7}
  1= \langle  \vphi_{k_l}, \D(\rho_{s^l})  \vphi_{k_l}     \rangle
\end{equation}
for all $s^l\in\{ 0,1 \}^l  $. Choosing $s^l=(0,\ldots, 0)$ and noting that due to the fact that $\D$ is a LOCC operation the state $\D(\rho_0^{\otimes l})$ is separable, we obtain from \cite{horodecki-reduction}
\begin{equation}\label{distillation-8}
  1=\langle \vphi_{k_l}, \D(\rho_0^{\otimes l})\vphi_{k_l}  \rangle\le \frac{1}{k_l}. \end{equation}
Thus $k_l=1$ and $d(l,\rho)=1$ for all $l\in\nn$. We collect these observations in the following corollary.
\begin{corollary}
Let $\rho \in \ri \cs (\hr_A\otimes \hr_B) $. Then $d(l,\rho)=1$ for all $l\in\nn$ and $D_0(\rho)=0$. Moreover, the function $D_0$ is discontinuous at any $\rho\in \cs (\hr_A\otimes \hr_B)$ with $D_0(\rho)>0$.
\end{corollary}
\section{Conclusion}
We have been able to derive a multi-letter analog of Ahlswede's dichotomy for quantum capacities of arbitrarily varying quantum channels: Either the classical, deterministic 
capacity of such a channel with average error criterion is zero , or else its deterministic and common-randomness-assisted entanglement transmission capacities are equal. 
Moreover, we have shown that the entanglement and strong subspace transmission capacities for this channel model are equal. It should be noted, however, that our proof of 
this does not rely on a strategy of ``hiding'' randomness in the encoding operation. In fact, by using a probabilistic variant of Dvoretzky's theorem we achieve this equality 
of capacities just  by restricting to an appropriate code subspace of comparable dimension on the exponential scale. Here we have left open the question whether the quantum 
capacity of arbitrarily varying quantum channels can be achieved with isometric encoding operations.\\
Simple conditions that guarantee single-letter capacity formulas have been provided. They are generalizations of those for memoryless and stationary quantum channels.\\ 
The major unresolved problem of this paper is the question whether there are AVQCs for which $C_{\textup{det}}(\mathfrak{I})=0$ and 
$\mathcal{A}_{\textup{random}}(\mathfrak{I})>0$ can occur. Notice that $C_{\textup{det}}(\mathfrak{I})=0$ immediately implies
$\mathcal{A}_{\textup{det}}(\mathfrak{I})=0$ and that the quantum version of Ahlswede's dichotomy proved in this paper shows that
$C_{\textup{det}}(\mathfrak{I})>0$ leads to $ \mathcal{A}_{\textup{det}}(\mathfrak{I})=\mathcal{A}_{\textup{random}}(\mathfrak{I})$. 
Therefore, the actual question behind the open problem stated above is, whether it can happen that 
$\mathcal{A}_{\textup{det}}(\mathfrak{I})=0$ and $\mathcal{A}_{\textup{random}}(\mathfrak{I})>0$ for some AVQC $\mathfrak{I}$. An affirmative example
to this question would display a rather striking super-activation phenomenon: a polynomial amount (in block length) of common randomness might
boost the capacity of an AVQC from $0$ to a positive value. On the other hand, showing that  $\mathcal{A}_{\textup{det}}(\mathfrak{I})=0$ implies
$\mathcal{A}_{\textup{random}}(\mathfrak{I})=0$ for all AVQCs $\mathfrak{I}$ would guarantee that 
$ \mathcal{A}_{\textup{det}}(\mathfrak{I})=\mathcal{A}_{\textup{random}}(\mathfrak{I})$ for all AVQCs $\mathfrak{I}$. The latter identity would relieve
us from checking the two rather intractable symmetrizability conditions developed in Sections \ref{subsec:classical-deterministic-average-error} and 
\ref{subsec:Classical-capacity-with-deterministic-codes-and-maximal error}.\\
Either affirmative or negative resolution of this issue is surprising for its own reason. At the present time we think that the latter answer is 
correct. 
\\\\
\emph{Acknowledgment.} We thank Andreas Winter for encouragement and his repeated insistence on launching this research.\\
We are indebted to an anonymous referee for numerous comments, questions, corrections, and advice that helped us a lot to improve the overall
structure and readability of the manuscript. Thank you!\\
Several very stimulating discussions with Toby Cubitt and Debbie Leung during the fall program 2010 at the Institut Mittag-Leffler have triggered our interest in zero-error capacities which resulted in the last part of the present paper. We thank both of them for these conversations.\\
Support by the Institut Mittag-Leffler (Djursholm, Sweden) is gratefully
acknowledged\\
This work is supported by the DFG via grants Bj 57/1-2 (IB), BO 1734/20-1 (IB and HB) and by the BMBF via grants 01BQ1050 (IB, HB, and JN) and 01BQ1052 (RA).





\begin{thebibliography}{99}
\bibitem{ahlswede-note} R. Ahlswede, ``A Note on the Existence of the Weak Capacity for Channels with Arbitrarily
Varying Channel Probability Functions and Its Relation to Shannon's Zero Error
Capacity'' \emph{The Annals of Mathematical Statistics}, Vol. 41, No. 3. (1970)

\bibitem{ahlswede-elimination}
R. Ahlswede, ``Elimination of Correlation in Random Codes for Arbitrarily Varying Channels'', \emph{Z. Wahrscheinlichkeitstheorie verw. Gebiete} 44, 159-175 (1978)

\bibitem{ahlswede-coloring}
R. Ahlswede, ``Coloring Hypergraphs: A New Approach to Multi-user Source Coding-II'', \emph{Journal of Combinatorics, Information \& System Sciences} Vol. 5, No. 3, 220-268 (1980)

\bibitem{ahlswede-gelfand-pinsker}
R. Ahlswede, ``Arbitrarily Varying Channels with States Sequence Known to the Sender'', \emph{IEEE Trans. Inf. Th.} Vol. 32, 621-629, (1986)

\bibitem{ahlswede-blinovsky}
R. Ahlswede, V. Blinovsky, ``Classical Capacity of Classical-Quantum Arbitrarily Varying Channels'', \emph{IEEE Trans. Inf. Th.} Vol. 53, No. 2, 526-533 (2007)

\bibitem{ahlswede-wolfowitz-2} R. Ahlswede, J. Wolfowitz, ``The Capacity of a Channel with Arbitrarily Varying Channel Probability Functions and Binary Output Alphabet'' Z. Wahrscheinlichkeitstheorie verw. Geb. 15, 186-194 (1970)

\bibitem{barnum-knill-nielsen} H. Barnum, E. Knill and M.A. Nielsen, ``On Quantum Fidelities and Channel Capacities'', \emph{IEEE Trans. Inf. Theory}, VOL. 46, NO. 4, (2000)

\bibitem{bennet-divincenzo-smolin}
C.H. Bennett, D.P. DiVincenzo, and J.A. Smolin, ``Capacities of Quantum Erasure Channels'', Phys. Rev. Lett. 78, 3217–3220 (1997)

\bibitem{bbn-1}
I. Bjelakovi\'c, H. Boche, J. N\"otzel, ``Quantum capacity of a class of compound channels'', \emph{Phys. Rev. A} 78, 042331, (2008)

\bibitem{bbn-2}
I. Bjelakovi\'c, H. Boche, J. N\"otzel, ``Entanglement transmission and generation under channel uncertainty: Universal quantum channel coding'', \emph{Commun. Math. Phys.} 292, 55-97 (2009)\\
I. Bjelakovi\'c, H. Boche, J. N\"otzel, ``Erratum to 'Entanglement transmission and generation under channel uncertainty: Universal quantum channel coding' ''...
\bibitem{bbt-avc}
D. Blackwell, L. Breiman, A.J. Thomasian, ``The capacities of certain channel classes under random coding'', \emph{Ann. Math. Stat.} 31, 558-567 (1960)

\bibitem{choi}
M.-D. Choi, ``Completely Positive Linear Maps on Complex Matrices'',
\emph{Linear Algebra and Its Applications} 10, 285-290 (1975)

\bibitem{csiszar}
I. Csiszar, J. K\"orner, \emph{Information Theory; Coding Theorems for Discrete Memoryless Systems}, Akad\'emiai Kiad\'o, Budapest/Academic Press Inc., New York 1981

\bibitem{csiszar-narayan}
I. Csiszar, P. Narayan, ``The Capacity of the Arbitrarily Varying Channel Revisited: Positivity, Constraints'', \emph{IEEE Trans. Inf. Th.} Vol. 34, No. 2, 181-193 (1989)

\bibitem{devetak-shor} I. Devetak and P.W. Shor, ``The Capacity of a Quantum Channel for Simultaneous Transmission of Classical and Quantum Information'', \emph{Commun. Math. Phys.} Vol. 256, Nr. 2 (2005)

\bibitem{duan-severini-winter} 
R. Duan, S. Severini, A. Winter, ``Zero-error communication via quantum channels, non-commutative graphs and a quantum Lov\'asz $\theta$ function'', arXiv:1002.2514v2

\bibitem{ericson}
T. Ericson, ``Exponential Error Bounds for Random Codes in the Arbitrarily Varying Channel'', \emph{IEEE Trans. Inf. Th.} Vol. 31, No. 1, 42-48 (1985)

\bibitem{fekete}
M. Fekete, ``\"Uber die Verteilung der Wurzeln bei gewissen algebraischen Gleichungen mit ganzzahligen Koeffizienten'', Mathematische Zeitschrift 17, 228 (1923).

\bibitem{gilbert} E.N. Gilbert, ``A comparison of signaling alphabets'', \emph{Bell System  Tech.  J.}
31,  504-522. (1952)

\bibitem{horodecki-reduction}
M. Horodecki, P. Horodecki, ``Reduction criterion of separability and limits for a class of distillation protocols'', \emph{Phys. Rev. A} Vol. 59, No. 6, 4206 (1999)

\bibitem{horodecki-horodecki-horodecki}M. Horodecki, P. Horodecki, R. Horodecki, ``General teleportation channel, singlet fraction, and quasidistillation
'', \emph{Phys. Rev. A} 60, 1888–1898 (1999)

\bibitem{kakutani} S. Kakutani, ``A Generalization of Brouwer's Fixed Point Theorem'', \emph{Duke Math. J.}, Volume 8, Number 3, 457-459 (1941)

\bibitem{kiefer-wolfowitz} J. Kiefer, J. Wolfowitz, ``Channels with arbitrarily varying channel probability functions'', \emph{Information and Control} 5, 44-54 (1962)

\bibitem{kitaev}
A.Y. Kitaev, A.H. Shen, M.N. Vyalyi, \emph{Classical and Quantum Computation}, Graduate Studies in Mathematics 47, American Mathematical Society, Providence, Rhode Island 2002

\bibitem{knill-laflamme}
E. Knill, R. Laflamme, ``Theory of quantum error-correcting codes'', \emph{Phys. Rev. A} Vol. 55, No. 2, 900-911 (1997)

\bibitem{koerner-orlitsky}
J. K\"orner, A. Orlitsky, ``Zero-error Information Theory'', \emph{IEEE Trans. Inf. Theory} Vol. 44, No. 6, 2207-2229 (1998)

\bibitem{leung-smith}
D. Leung, G. Smith, ``Continuity of quantum channel capacities'', \emph{Commun. Math. Phys.} 292, 201-215, (2009)

\bibitem{lieb-ruskai} E.H. Lieb and M.B. Ruskai, ``Proof of the strong subadditivity of quantum-mechanical entropy'', \emph{J. Math. Phys.} 14, 1938 (1973)

\bibitem{matousek} J. Matousek, \emph{Lectures on Discrete Geometry}, Graduate Texts in Mathematics, Vol. 212, Springer 2002

\bibitem{milman-schechtman} V.D. Milman, G. Schechtman \emph{Asymptotic Theory of Finite Dimensional Normed Spaces}, Lecture Notes in Mathematics vol. 1200, Springer-Verlag 1986

\bibitem{paulsen}
V. Paulsen, \emph{Completely Bounded Maps and Operator Algebras}, Cambridge Studies in Advanced Mathematics vol. 78, Cambridge University Press 2002

\bibitem{polya-szegoe} G. P\'olya, V. Szeg\"o \emph{Problems and Theorems in Analysis I}, Springer 1998

\bibitem{schumacher-nielsen} B. Schumacher, M.A. Nielsen, Quantum data processing and error correction. \emph{Phys. Rev. A} Vol. 54, No. 4,
2629 (1996)

\bibitem{shannon} C. E. Shannon,  ``The zero error capacity  of a  noisy channel''. \emph{IRE  Trans.  Information Theory  IT-2},
8-19  (1956)

\bibitem{von-neumann} J. von Neumann, ``Zur Theorie der Gesellschaftsspiele'', \emph{Math. Ann.} Vol. 100, 295-320 (1928)

\bibitem{webster}
R. Webster, \emph{Convexity}, Oxford University Press 1994

\bibitem{yard-devetak-hayden}
J. Yard, I. Devetak, P. Hayden, ``Capacity theorems for quantum multiple access channels: Classical-quantum and quantum-quantum capacity Regions'', \emph{IEEE Trans. Inf. Theory} 54, 3091 (2008) e-print arXiv:quant-ph/0501045.

\end{thebibliography}
\end{document}